\documentclass[10pt, twoside]{IEEEtran} \onecolumn 

\usepackage{graphicx}
\usepackage{amsmath, amssymb, amsthm, algorithm, algorithmic}
\usepackage{epsfig,color,epstopdf}
\usepackage{bm}
\usepackage{datetime}

\usepackage{caption}
\usepackage{subcaption}

\usepackage{tikz}
\usetikzlibrary{decorations.pathreplacing}

\usepackage{enumerate}
\usepackage[normalem]{ulem}

\usepackage{color}

\newtheorem{theorem}{Theorem}[section]
\newtheorem{lem}[theorem]{Lemma}
\newtheorem{sigmodel}[theorem]{Model}
\newtheorem{corollary}[theorem]{Corollary}
\newtheorem{definition}[theorem]{Definition}
\newtheorem{remark}[theorem]{Remark}

\newtheorem{fact}[theorem]{Fact}

\newcommand{\ds}{\displaystyle}

\newcommand{\R}{\mathbb{R}}

\renewcommand{\Pr}{\mathbb{P}}

\newcommand{\bi}{\begin{itemize}}
\newcommand{\ei}{\end{itemize}}
\newcommand{\ben}{\begin{enumerate}}
\newcommand{\een}{\end{enumerate}}

\newcommand{\bean}{\begin{eqnarray*} }
\newcommand{\eean}{\end{eqnarray*} }
\newcommand{\bea}{\begin{eqnarray} }
\newcommand{\eea}{\end{eqnarray} }
\newcommand{\nn}{\nonumber}
\newcommand{\ba}{\begin{array} }
\newcommand{\ea}{\end{array} }
\newcommand{\beq}{\begin{equation}}
\newcommand{\eeq}{\end{equation}}

\newcommand{\vect}[2]{\left[\begin{array}{cccccc}
     #1 \\
     #2
   \end{array}
  \right]
  }

\renewcommand\thetheorem{\arabic{section}.\arabic{theorem}}

\newcommand{\wt}{\bm{w}_t}
\newcommand{\mt}{\bm{m}_t}
\newcommand{\xt}{\bm{x}_t}
\newcommand{\xhatt}{\hat{\bm{x}}_t}
\newcommand{\lt}{\bm{\ell}_t}
\renewcommand{\l}{\bm{\ell}}

\newcommand{\e}{\bm{e}}
\newcommand{\bnu}{\bm{\nu}}
\newcommand{\lhatt}{\hat{\bm{\ell}}_t}
\newcommand{\et}{\bm{e}_t}
\newcommand{\Pt}{\bm{P}_t}
\newcommand{\Pj}{\bm{P}_{(j)}}

\newcommand{\Pjs}{\bm{P}_{(j),*}}

\newcommand{\rmnew}{\mathrm{new}}
\newcommand{\new}{\mathrm{new}}

\newcommand{\Pjnew}{\bm{P}_{(j),\mathrm{new}}}
\newcommand{\at}{\bm{a}_t}

\newcommand{\atnew}{\bm{a}_{t,\mathrm{new}}}
\newcommand{\I}{\bm{I}}
\newcommand{\Lamt}{\bm{\Lambda}_t}
\newcommand{\Lamtnew}{\bm{\Lambda}_{t,\mathrm{new}}}

\newcommand{\Lamjex}{\bm{\Lambda}_{(j)}}
\newcommand{\Lamj}{\bm{\Lambda}_{(j)}}
\newcommand{\T}{\mathcal{T}}
\newcommand{\J}{\mathcal{J}}

\newcommand{\cs}{\text{cs}}

\newcommand{\bigo}{\mathcal{O}}

\newcommand{\old}{\mathrm{old}}

\newcommand{\lhat}{\hat{\bm{\ell}}}

\renewcommand{\P}{\bm{P}}
\renewcommand{\b}{\bm{b}}
\newcommand{\Phat}{\hat{\bm{P}}}

\newcommand{\Span}{\operatorname{range}}
\newcommand{\del}{\mathrm{del}}

\newcommand{\add}{\mathrm{add}}

\newcommand{\rank}{\operatorname{rank}}
\newcommand{\E}{\mathbb{E}}

\newcommand{\calc}{\mathcal{C}}

\newcommand{\cov}{\operatorname{Cov}}

\newcommand{\train}{\mathrm{train}}
\newcommand{\thresh}{\mathrm{thresh}}  
\newcommand{\That}{\hat{\mathcal{T}}}

\newcommand{\SE}{\mathrm{SE}}

\newcommand{\that}{{\hat{t}}}
\newcommand{\rmend}{\mathrm{end}}
\newcommand{\M}{\bm{\mathcal{M}}}
\newcommand{\range}{\operatorname{range}}
\newcommand{\llceil}{\left\lceil}
\newcommand{\rrceil}{\right\rceil}

\newcommand{\uhat}{{\hat{u}}}
\newcommand{\lammin}{{\lambda}^-}
\newcommand{\lamtrain}{{\hat{\lambda}_{\train}^-}}
\newcommand{\jhat}{j}

\newcommand{\rrho}{\rho}
\newcommand{\Ijkt}{\tilde{\mathcal{I}}_{j,k}}
\newcommand{\ghatp}{\hat{g}^+}

\newcommand{\nut}{\bm{\nu}_t}
\newcommand{\bmG}{\bm{G}}
\renewcommand{\L}{\bm{L}}
\renewcommand{\S}{\bm{S}}

\newcommand{\undet}{\mathrm{undet}}

\newcommand{\ITt}{{\bm{I}_{\mathcal{T}_t}}}
\newcommand{\invterm}{[(\bm{\Phi}_{t})_{\mathcal{T}_t}{}'(\bm{\Phi}_{t})_{\mathcal{T}_t}]^{-1}}
\newcommand{\ttau}{{\tilde{\tau}}}

\newcommand{\bb}{q}
\newcommand{\bbb}{q_2}

\begin{document}
\title{Online (and Offline) Robust PCA: Novel Algorithms and Performance Guarantees}
\author{
Jinchun~Zhan, Brian~Lois, and Namrata~Vaswani
\thanks{J. Zhan and N. Vaswani are with the ECE department at Iowa State University. B. Lois was with the Mathematics and ECE departments at Iowa State when this work was done. He is currently with AT\&T Big Data in Plano, TX. Email: \{jzhan,blois,namrata\}@iastate.edu.
A shorter version of this work will appear in the proceedings of AISTATS 2016.
}
}

\maketitle


\begin{abstract}
In this work, we study the online robust principal components' analysis (RPCA) problem. In recent work, RPCA has been defined as a problem of separating a low-rank matrix (true data), $\L:=[\ell_1, \ell_2, \dots \ell_{t}, \dots , \ell_{t_{\max}}]$, and a sparse matrix (outliers), $\S:=[x_1, x_2, \dots x_{t}, \dots, x_{t_{\max}}]$, from their sum, $\bm{M}:=\L + \S$. A more general version of this problem is to recover $\L$ and $\S$ from $\bm{M}:=\L + \S + \bm{W}$ where $\bm{W}$ is the matrix of unstructured small noise/corruptions. An important application where this problem occurs is in video analytics in trying to separate sparse foregrounds (e.g., moving objects) from slowly changing backgrounds. While there has been a large amount of recent work on solutions and guarantees for the batch RPCA problem, the online problem is largely open.``Online" RPCA is the problem of doing the above on-the-fly with the extra assumptions that the initial subspace is accurately known and that the subspace from which $\lt$ is generated changes slowly over time.

We develop and study a novel ``online" RPCA algorithm based on the recently introduced Recursive Projected Compressive Sensing (ReProCS) framework. Our algorithm improves upon the original ReProCS algorithm and it also returns even more accurate offline estimates. The key contribution of this work is a correctness result (complete performance guarantee) for this algorithm under reasonably mild assumptions. By using extra assumptions -- accurate initial subspace knowledge, slow subspace change, and clustered eigenvalues -- we are able to remove one important limitation of batch RPCA results and two key limitations of a recent result for ReProCS for online RPCA. To our knowledge, this work is among the first few correctness results for online RPCA. Most earlier results were only partial results, i.e., they required an assumption on intermediate algorithm estimates.%
\end{abstract}

\section{Introduction}
Principal Components Analysis (PCA) is a tool that is frequently used for dimension reduction. Given a matrix of data, PCA computes a small number of orthogonal directions that contain most of the variability of the data. PCA for relatively noise-free data is easily accomplished via singular value decomposition (SVD). 
The robust PCA (RPCA) problem, which is the problem of PCA in the presence of outliers, is much harder. In recent work, Cand\`{e}s et al. \cite{rpca} posed it as a problem of separating a low-rank matrix, $\L$, (true data) and a sparse matrix, $\S$, (outliers\footnote{Since an outlier is something that occurs occasionally, it is well modeled using a sparse matrix of corruptions.}) from their sum, $\bm{M}:=\L+\S$. They proposed a convex program called principal components' pursuit (PCP) that provided a provably correct batch solution to this problem under mild assumptions. The same program was also analyzed in Chandrasekharan et al. \cite{rpca2} and later in Hsu et al. \cite{rpca_zhang}. Since these works, there has been a large amount of work on batch RPCA methods and performance guarantees. The more general case, $\bm{M}:=\L + \S + \bm{W}$ where $\bm{W}$ is unstructured small noise/corruptions, has also been studied in later works, e.g., \cite{stablePCP}.

When RPCA needs to be solved in a recursive fashion for sequentially arriving data vectors it is referred to as incremental or recursive or dynamic or ``online" RPCA. 
``Online" RPCA assumes that (i) a short sequence of outlier-free (sparse component free) data vectors is available or that there is another way to get an estimate of the initial subspace of the true data (without outliers); and that (ii) the subspace from which $\lt$ is generated is either fixed or changes slowly over time. We put ``online" in quotes here to stress that the ``online" problem formulation uses extra assumptions beyond what are used by batch RPCA.
An important application where the RPCA problem occurs is one of separating a video sequence into foreground and background layers \cite{rpca}. Video layering is a key first step to simplifying many video analytics and computer vision tasks, e.g., video surveillance (to track moving foreground objects), background video recovery and subspace tracking in the presence of frequent foreground occlusions or low-bandwidth mobile video chats or video conferencing (can transmit only the foreground layer).
In videos, the foreground typically consists of one or more moving persons or objects and hence is a sparse image. The background images (in a static camera video) usually change only gradually over time, e.g., moving lake waters or moving trees in a forest, and the changes are global \cite{rpca}. Hence they are well modeled as being dense and lying in a low-dimensional subspace that is fixed or slowly changing. We show an example in Fig. \ref{expt}. In many videos, it is also valid to assume that a short initial sequence is available without any foreground objects, i.e., (i) holds.
Other RPCA applications include recommendation system design, survey data analysis \cite{rpca}, anomaly detection in dynamic social (or computer) networks \cite{mateos_anomaly} or dynamic magnetic resonance imaging (MRI) based region-of-interest tracking \cite{candes_mri}. In many of these, an online solution is desirable.%


\subsection{Problem Definition} \label{probdef}
At time $t$ we observe a data vector $\mt \in \R^n$ that satisfies
\bea
\mt = \lt + \xt + \wt, \ \text{for} \  t = t_{\train}+1, t_{\train}+2, \dots, t_{\max}.
\label{orpca_eq}
\eea
For $t=1,2,\dots, t_\train$, $\xt=0$, i.e., $\mt = \lt + \wt$.
Here  $\lt$ is a vector that lies in a low-dimensional subspace that is fixed or slowly changing in such a way that the matrix $\bm{L}_t:=[\l_1, \l_2, \dots, \lt]$ is a low-rank matrix for all but very small values of $t$; $\xt$  is a sparse (outlier) vector; and $\wt$ is small modeling error or noise. We use $\T_t$ to denote the support set of $\xt$ and we use $\P_t$ to denote a basis matrix for the subspace from which $\lt$ is generated.
For $t > t_\train$, the goal of online RPCA is to recursively estimate $\lt$ and its subspace $\Span(\P_t)$,  and $\xt$ and its support, $\T_t$, as soon as a new data vector $\mt$ arrives or within a short delay\footnote{By definition, a subspace of dimension $r>1$ cannot be estimated immediately since it needs at least $r$ data points to estimate}.
Sometimes, e.g., in video analytics, it is often also desirable to get an improved offline estimate of $\xt$ and $\lt$ when possible. 
We show that this is an easy by-product of our solution approach.

The initial $t_\train$ outlier-free measurements are used to get an accurate estimate of the initial subspace via PCA. For video, this assumption corresponds to having a short initial sequence of background-only images, which can often be obtained.
%
%

In many applications, it is actually the sparse outlier $\xt$ that is the quantity of interest. The above problem can thus also be interpreted as one of {\em online sparse matrix recovery in large but structured noise $\lt$ and unstructured small noise $\wt$}.
%
The unstructured noise, $\wt$, often models the modeling error. For example, when some of the corruptions/outliers are small enough to not significantly increase the subspace recovery error, these can be included into $\wt$ rather than $\xt$.
Another example is when the $\lt$'s form an approximately low-rank matrix.

\begin{figure}[t!]
\centering
\includegraphics[width=\linewidth]{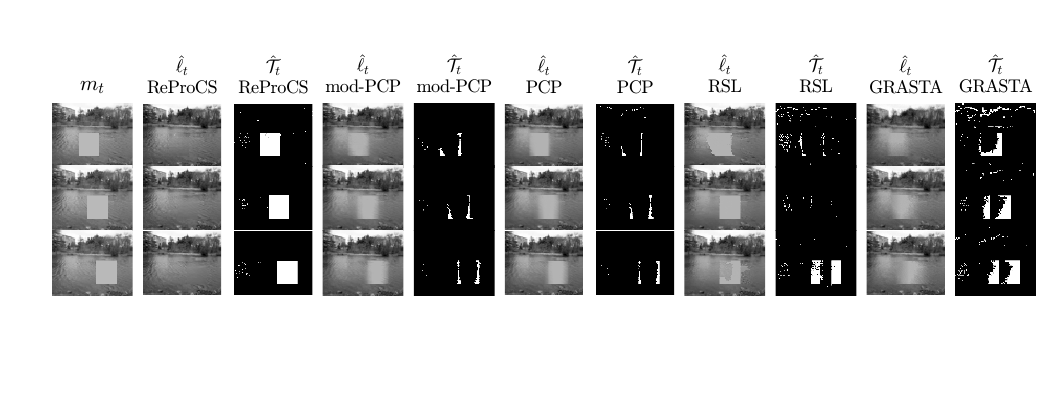} \vspace{-0.9in}
\caption{\small{The first column shows the video of a moving rectangular object against moving lake waters' background. The object and its motion are simulated while the background is real. In the next two columns, we show the recovered background ($\lhatt$) and the recovered foreground support ($\That_t$) using  Automatic ReProCS-cPCA (labeled ReProCS in the figure). The algorithm parameters are set differently for the experiments (see Sec. \ref{expts}) than in our theoretical result. Notice that the foreground support is recovered mostly correctly with only a few extra pixels and the background appears correct too (does not contain the moving block). The quantitative comparison is shown later in Fig. \ref{sims}.
The next few columns show background and foreground-support recovery using some of the existing methods discussed in Sec. \ref{relwork}.}}
\label{lake_view} \label{expt}
\end{figure}

\subsection{Related Work} \label{relwork}
Solutions for online RPCA have been analyzed in recent works  \cite{rrpcp_perf}, \cite{xu_nips2013_1}, \cite{rrpcp_icassp15,rrpcp_isit15}. The work of \cite{rrpcp_perf} introduced the Recursive Projected Compressive Sensing (ReProCS) algorithmic framework and obtained a partial result for it. Another approach for online RPCA (defined differently from above) and a partial result for it were provided in \cite{xu_nips2013_1}. We use the term {\em partial result} to refer to a performance guarantee that depends on intermediate algorithm estimates satisfying certain properties. We will see examples of this in Sec. \ref{discussion} when we discuss the above results.
In very recent work \cite{rrpcp_icassp15,rrpcp_isit15}, a {\em correctness result} for ReProCS was obtained. The term {\em correctness result} refers to a complete performance guarantee, i.e., a  guarantee that only puts assumptions on the input data (here $\mt$) and/or on the algorithm initialization, but not on intermediate algorithm estimates.%

Other somewhat related work includes \cite{xu_nips2013_2} (online PCA with contaminated data that is not modeled as being sparse) and \cite{zhan_pcp_jp} (modified-PCP, a piecewise batch method). All the above results are discussed Sec. \ref{discussion}.

Some other works, such as \cite{grass_undersampled}(GRASTA), \cite{sequentialSVD} (adaptive-iSVD), \cite{Li03anintegrated} (incremental Robust Subspace Learning) or \cite{xu2013gosus} (GOSUS), \cite{mateos_anomaly,mateos_anomaly2}, \cite{chouvardas2015robust}, \cite{songimage}, \cite{vinjamuri2015sparse}
only provide an online RPCA algorithm without guarantees. We do not discuss these here. As demonstrated by the experimental comparisons shown in \cite{rrpcp_tsp} and in
\cite[Fig 6]{zhan_pcp_jp}, when the outlier support is large and changes in a correlated fashion over time, ReProCS-based algorithms significantly outperform most of these, besides also outperforming batch methods such as PCP and robust subspace learning (RSL) \cite{rpca,Torre03aframework}. This is also evident from Fig. \ref{expt} and Fig. \ref{sims}.
%


\subsection{Contributions}
In this work we develop and study an algorithm based on the ReProCS idea introduced and studied in \cite{rrpcp_perf,rrpcp_icassp15,rrpcp_isit15}. We call it Automatic ReProCS with cluster PCA (ReProCS-cPCA). This is an improved ReProCS algorithm compared to the ones studied in previous work. (1) It is able to automatically detect subspace changes within a short delay; is able to correctly estimate the number of directions added or deleted; and is also able to correctly estimate the clusters of eigenvalues along the existing directions. This is important because it is impractical to assume that a subspace change time or the exact number of added or removed directions is known. Additionally, these estimates themselves are relevant for applications such as understanding dynamic social networks' structural changes in the presence of outliers. While many heuristics exist to detect sudden subspace changes, we provide an approach for correctly detecting slow subspace changes within a short delay. (2) Moreover it is able to accurately estimate both the newly added subspace as well as the newly deleted subspace. The latter is done by re-estimating the current subspace using an approach called cluster PCA (cPCA). The basic cPCA idea was introduced in \cite{rrpcp_perf}. The current work uses that idea to develop an automatic algorithm. The cPCA step ensures that the estimated subspace dimension does not keep increasing with time. (3) The current algorithm also returns more accurate offline estimates.
The algorithms studied in \cite{rrpcp_perf,rrpcp_icassp15} could not do (1) and (3). The algorithms studied in \cite{rrpcp_icassp15,rrpcp_isit15} did not do (2) and (3).


The {\em main} contribution of this work is a correctness result (complete performance guarantee) for the proposed algorithm under relatively mild assumptions on $\lt$, $\xt$, and $\wt$. To our knowledge, this and \cite{rrpcp_icassp15,rrpcp_isit15} are the first correctness results for online RPCA. The result obtained here removes two key limitations of \cite{rrpcp_icassp15,rrpcp_isit15}. (1) First, we obtain a result for the case where the $\lt$'s can be correlated over time (follow an autoregressive (AR) model) where as the result of \cite{rrpcp_icassp15,rrpcp_isit15} needed mutual independence of the $\lt$'s. This models mostly static backgrounds in which changes are only due to independent variations at each time, e.g., light flickers. However, a large class of background image sequences change due to factors that are correlated over time, e.g., moving waters. This can be better modeled using an AR model.
(2) Second, with one extra assumption -- that the eigenvalues of the covariance matrix of $\lt$ are clustered for a period of time after the previous subspace change has stabilized -- we are able to remove another significant limitation of \cite{rrpcp_icassp15,rrpcp_isit15}. That result needed the rank of $\L$ to grow as $\bigo(\log n)$ while our result allows it to grow as $\bigo(n)$. Batch methods such as PCP allow the rank to grow almost linearly with $n$. The clustered eigenvalues assumption is valid for data that has variability at different scales - large scale variations would result in the first (largest eigenvalues') cluster and the smaller scale variations would form the later clusters.

Because we use extra assumptions -- accurate initial subspace knowledge, slow subspace change, and clustered eigenvalues -- we are able to remove an important limitation of batch methods \cite{rpca,rpca2,rpca_zhang}. As we explain in Sec. \ref{discussion}, our result requires an order-wise looser bound on the number of time instants for which a particular index $i$ can be outlier-corrupted compared to these results. In other words, it allows significantly more correlated changes of the outlier support over time. This is important in practice, e.g., in video, foreground objects do not randomly jump around; in social networks, once an anomalous pattern starts to occur, it remains on many of the same edges for a while. The clustered eigenvalues assumption is discussed above. Accurate initial subspace knowledge and slow subspace change were discussed earlier (just above Sec. \ref{probdef}).

The novelty in the proof techniques used in this work is summarized in Sec. \ref{novelty}. The proof relies on the $\sin \theta$ theorem \cite{davis_kahan} (that bounds the effect of a perturbation on a Hermitian matrix's top eigenvectors) and the matrix Azuma inequality \cite{tail_bound}.

\subsection{Notation}


We use the interval notation $[a, b]$ to mean all of the integers between $a$ and $b$, inclusive, and similarly for $[a,b)$ etc.
 For a set $\mathcal{T}$, $|\mathcal{T}|$ denotes its cardinality and $\bar{\mathcal{T}}$ denotes its complement set. We use $\emptyset$ to denote the empty set.

We use $'$ to denote a vector or matrix transpose.
The  $l_p$-norm of a vector and the induced $l_p$-norm of a matrix are denoted by $\| \cdot \|_p$.
For a vector $\bm{x}$ and set $\T$, $\bm{x}_{\mathcal{T}}$ is a smaller vector containing the entries of $\bm{x}$ indexed by entries in $\mathcal{T}$.
We use $\I$ to denote the identity matrix.
Define $\I_{\mathcal{T}}$ to be an $n \times |\mathcal{T}|$ matrix of those columns of the identity matrix indexed by entries in $\mathcal{T}$. For a matrix $\bm{A}$, define $\bm{A}_{\mathcal{T}} := \bm{AI}_{\mathcal{T}}$.
For matrices $\bm{P}$, $\bm{Q}$ where the columns of $\bm{Q}$ are a subset of the columns of $\bm{P}$, $\bm{P} \setminus \bm{Q}$ refers to the matrix of columns in $\bm{P}$ and not in $\bm{Q}$.
For a matrix $\bm{H}$, $\bm{H}\overset{\mathrm{EVD}}= \bm{U\Lambda U}'$ denotes its reduced eigenvalue decomposition.
For Hermitian matrices $\bm{A}$ and $\bm{B}$, the notation $\bm{A}\preceq\bm{B}$ means that $\bm{B}-\bm{A}$ is positive semi-definite. 

For a matrix $\bm{A}$, the restricted isometry constant (RIC) $\delta_s(\bm{A})$ is the smallest real number $\delta_s$ such that
\[
(1-\delta_s)\|\bm{x}\|_2^2 \leq \|\bm{Ax}\|_2^2 \leq (1+\delta_s) \|\bm{x}\|_2^2
\]
for all $s$-sparse vectors $\bm{x}$ \cite{candes_rip}.  A vector $\bm{x}$ is $s$-sparse if it has $s$ or fewer non-zero entries.

We refer to a matrix with orthonormal columns as a {\em basis matrix}. Thus, for a basis matrix $\bm{P}$, $\bm{P}'\bm{P} = \bm{I}$.
For basis matrices $\Phat$ and $\P$, $\mathrm{dif}(\Phat,\P):= \|(\I - \Phat \Phat') \P \|_2$ quantifies error between their range spaces.%

\subsection{Paper organization}
This paper is organized as follows. We discuss the data models and the main results for the proposed algorithm in Sec. \ref{models}. The Automatic ReProCS-cPCA algorithm is developed in Sec. \ref{algo_sec}. The stepwise algorithm is summarized in Algorithm \ref{reprocsdet}. The proof outline of our main result is given in Sec. \ref{outline}. This section also helps understand the algorithm better and explains the novelty in the proof techniques. The lemmas for proving the main result, the proof of the main result and the proofs of the main lemmas are given in Sec. \ref{pf_thm1_cor}. The key lemmas needed to prove the main lemmas are proved in Sec. \ref{3_pfs} (lemmas for analyzing the projection-PCA based subspace addition step) and in Sec. \ref{3_pfs_del} (lemmas for analyzing the cluster PCA based subspace deletion step). These are the long sections that contain the new proofs that rely on the matrix Azuma inequality \cite{tail_bound}. This is needed because the $\lt$'s are now correlated over time. Simulation experiments comparing the proposed algorithm to some existing batch and online RPCA algorithms are described in Sec. \ref{expts}. Conclusions are given in Sec. \ref{conclude}.

\section{Data models and main results} \label{models}
In this section, we give the data models and correctness results for our proposed algorithm, Automatic ReProCS-cPCA, and for its simplification, Automatic ReProCS. The algorithm itself is developed in Sec \ref{algo_sec} and the complete stepwise algorithm is summarized in Algorithm \ref{reprocsdet}. We give below the model on the outlier support sets $\T_t$,  the model on $\lt$, and the denseness assumption. Using these, we state the result for Automatic ReProCS in Sec. \ref{noclust_result}. In Sec. \ref{clust_result}, we state the clustering assumption and give the correctness result for Automatic ReProCS-cPCA. The results are discussed in Sec. \ref{discussion}.


\subsection{Model on the outlier support set, $\T_t$} 
We give here one simple and practically relevant special case of the most general assumptions (Model \ref{general_model}) on the outlier support sets $\T_t$. It requires that the $\T_t$'s have {\em some} changes over time and have size less than $s$. An example of this is a video application consisting of a foreground with a 1D object of length $s$ or less that remains static for at most $\beta$ frames at a time. When it moves, it moves {\em downwards (or upwards, but always in one direction)} by at least $\frac{s}{\rrho}$ pixels, and at most $\frac{s}{\rho_2}$ pixels. Once it reaches the bottom of the scene, it disappears. The maximum motion is such that, if the object were to move at each frame, it still does not go from the top to the bottom of the scene in a time interval of length $\alpha$. This is ensured if $\frac{s}{\rho_2} \alpha \le n$. Anytime after it has disappeared another object could appear. A visual depiction of this model is shown in Fig. \ref{supportfig}.
We have used this ``one object moving in one direction" example  to only explain the idea in a simple fashion. Instead, one could also have multiple moving objects and arbitrary motions, as long as the union of their supports follows the assumptions of Model \ref{sbyrho} below or those given later in Model \ref{general_model}. These models were introduced in \cite{rrpcp_isit15}.

\begin{sigmodel}[model on $\T_t$]\label{sbyrho}
Let $t^k$, with $t^k < t^{k+1}$, denote the times at which $\T_t$ changes and let $ \T^{[k]}$ denote the distinct sets.
For an integer $\alpha$, 
\ben
\item assume that $\T_t = \T^{[k]}$ for all times $t \in [t^k, t^{k+1})$ with $(t^{k+1} - t^k) < \beta$ and $|\T^{[k]}| \le s$;
\item let $\rrho$ be a positive integer so that for any $k$,
$
\T^{[k]} \cap \T^{[k+\rrho ]} = \emptyset;
$
assume that
$
{\rrho}^2 \beta \leq  0.0001 \alpha;
$

\item for any $k$,
$
\sum_{i=k+1}^{k+\alpha} \left|\T^{[i]} \setminus \T^{[i+1]}\right| \le n
$ and for any $k < i \le k+\alpha$,
$
(\T^{[k]} \setminus \T^{[k+1]}) \cap (\T^{[i]} \setminus \T^{[i+1]}) = \emptyset
$
(one way to ensure the first condition is to require that for all $i$, $|\T^{[i]} \setminus \T^{[i+1]}| \le \frac{s}{\rho_2}$ with $\frac{s}{\rho_2}\alpha \le n$).
\een
In this model, $k$ takes values $1,2, \dots $; the largest value it can take is $t_{\max}$. We set $\alpha$ in the Theorem.
\end{sigmodel}

\begin{figure}[t]
\centering
\begin{subfigure}{.5\linewidth}
\centering
\includegraphics[height=5cm, width=3.5cm]{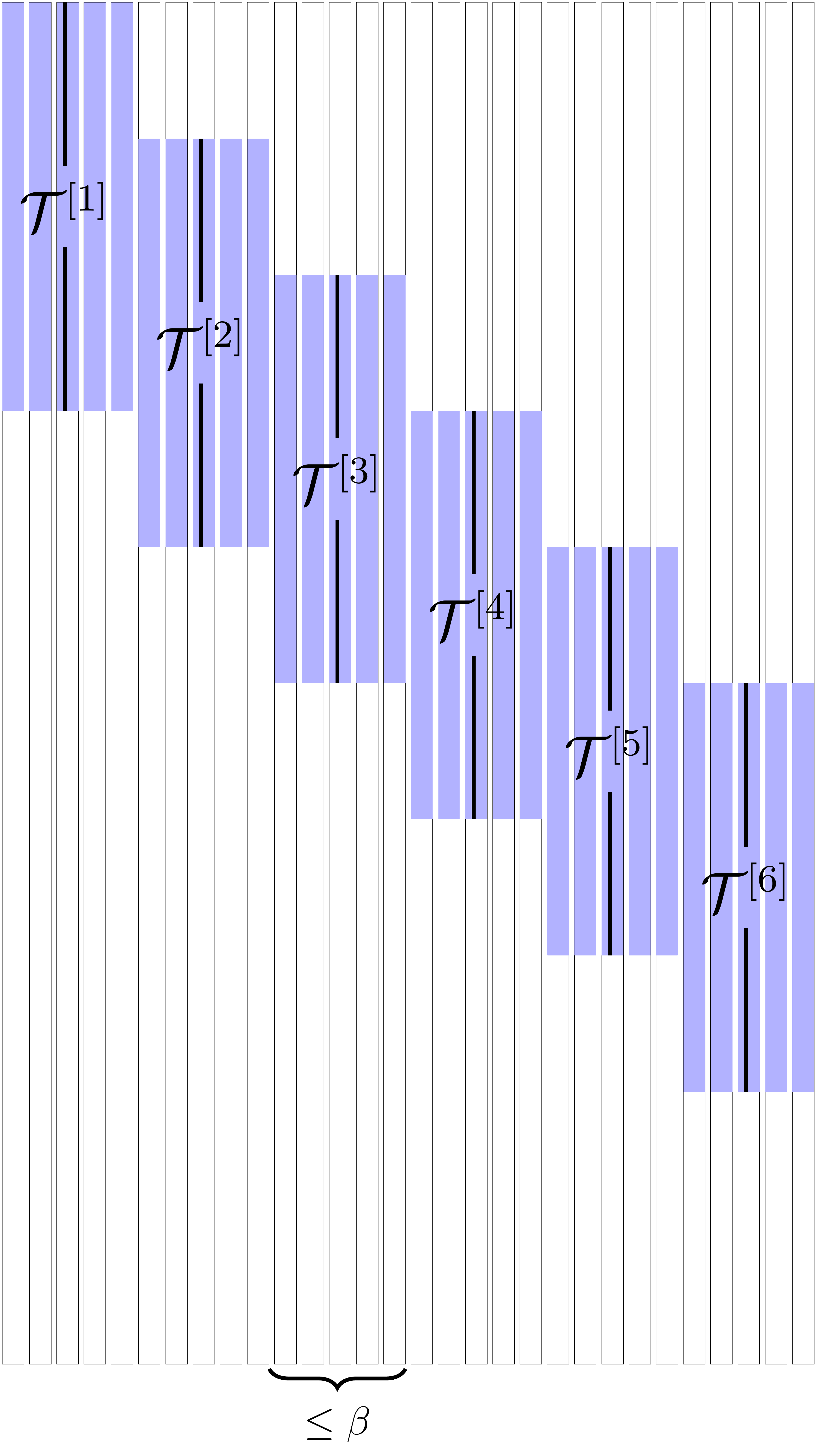}
\caption{$\rho=3$ and $\beta=5$ case}
\end{subfigure}
\begin{subfigure}{.3\linewidth}
\centering
\includegraphics[height=5cm, width=2cm]{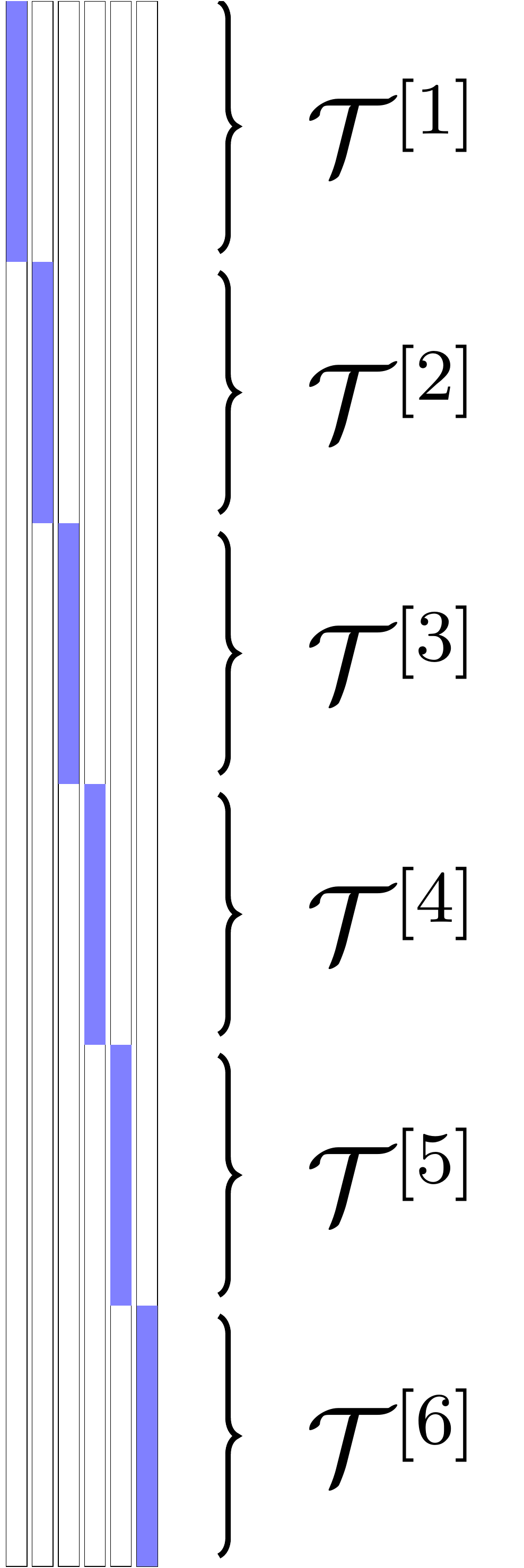}
\caption{$\rho=1$ and $\beta=1$ case}
\end{subfigure}
\caption{\small{Examples of Model \ref{sbyrho}. (a) shows a 1D object of length $s$ that moves by at least $s/3$ pixels at least once every $5$ frames (i.e., $\rho=3$ and $\beta=5$). (b) shows the object moving by $s$ pixels at every frame  (i.e., $\rho=1$ and $\beta=1$). (b) is an example of the best case for our result - the case with the smallest $\rho,\beta$ ($\T_t$'s  mutually disjoint)
}}
\vspace{-0.1in}
\label{supportfig}
\end{figure}



\subsection{Model on $\ell_t$}
A common model for data that lies in a low-dimensional subspace is to assume that, at all times, it is independent and identically distributed (iid) Gaussian with zero mean and a fixed {\em low-rank} covariance matrix $\bm\Sigma$. However this can be restrictive since, in many applications, data statistics change with time, albeit slowly. To model this perfectly, one would need to assume that $\lt$ is zero mean with covariance matrix $\bm{\Sigma}_t$ at time $t$. If $\bm{\Sigma}_t \overset{\mathrm{EVD}}{=} \P_t \Lamt \P_t'$, this means that both $\P_t$ and $\Lamt$ can change at each time $t$, though slowly. This is the most general model but it has an identifiability problem if the goal is to estimate the subspace from which $\lt$ was generated, $\Span(\P_t)$.
The subspace cannot be estimated with one data point. If it is $r$-dimensional, it needs at least $r$ data points. So, if $\P_t$ changes at each time, it is not clear how one can estimate all the subspaces. To resolve this issue, a general enough but tractable option is to assume that $\P_t$ is piecewise constant with time and $\Lamt$ can change at each time. To ensure that $\bm{\Sigma}_t$ changes ``slowly", we assume that, when $\P_t$ changes, the eigenvalues along the newly added directions are small initially for the first $d$ frames, and after that they can increase gradually or suddenly to any large value. One precise model for this is specified next.

The model below assumes boundedness of $\lt$. This is more practically valid than the usual Gaussian assumption since most sensor data or noise is bounded. We also replace independence of $\lt$'s by an AR model with independent perturbations $\bnu_t$ and we place the above assumptions on $\bnu_t$. As explained earlier, this is a more practical model and includes independence as a special case.%

\begin{sigmodel}[Model on $\lt$]\label{cor_model} Assume the following.
\ben
\item Let $\bm{\ell}_0=\bm{0}$ and for $t=1,2, \dots t_{\max}$, assume that
\[
\lt = b\bm{\ell}_{t-1} + \nut
\]
for a $b<1$. Assume that the $\nut$ are zero mean, mutually independent and bounded random vectors with covariance matrix
\[
\cov(\nut) = \bm{\Sigma}_{t} \overset{\mathrm{EVD}}{=} \bm{P}_t \bm{\Lambda}_t {\bm{P}_t}'.
\]


\item Let $t_1, t_2, \dots t_J$ denote the subspace change times. The basis matrices $\Pt$  change as
\begin{align*} \label{Pt_def}
\bm{P}_t =
%
\begin{cases}
 [(\bm{P}_{t-1} \bm{R}_t \setminus  \bm{P}_{t,\old})  \  \bm{P}_{t,\new}]  & \text{if} \ t = t_1, t_2, \dots \ t_J \\ 
\bm{P}_{t-1}  &   \text{otherwise.}
\end{cases}
\end{align*}
where $\bm{R}_t$ is a rotation matrix, $\P_{t_j,\new}$ and $\P_{t_j,\old}$ are basis matrices of size $n \times r_{j,\new}$ and $n \times r_{j,\old}$ respectively, $\bm{P}_{t_j,\old}$ contains a subset of columns of $\bm{P}_{t_j-1} \bm{R}_t$, and $\P_{t_j,\new}{}' \P_{t_j-1} = \bm{0}$ (new directions are orthogonal to previous subspace).

\item Define
\[
\lambda^-:= \lambda_{\min}\left( \frac{1}{t_{\train}}\sum_{t=1}^{t_{\train}} \Lamt \right) \ \text{and} \ \lambda^+:= \lambda_{\max}\left( \frac{1}{t_{\train}}\sum_{t=1}^{t_{\train}} \Lamt \right).
\]
The eigenvalues' matrices $\Lamt$ are such that (i) $\lambda_{\max}(\Lamt) \le \lambda^+ $ and (ii) for a $d< t_{j+1} - t_j$,
\begin{eqnarray} \label{anew_small}
&& 0< \lambda^- \leq  \lambda_{\new}^-  \leq \lambda_{\new}^+  \leq 3 \lambda^- \  \text{where} \  \nn \\
&& \lambda_{\new}^- := \min_j \min_{t\in [t_j,t_j + d]} \lambda_{\min}\left( {\bm{P}_{t_j,\new}}'\bm{\Sigma}_{t}\bm{P}_{t_j,\new} \right), \nn \\
&& \lambda_{\new}^+ := \max_j \max_{t\in [t_j,t_j + d]} \lambda_{\max}\left( {\bm{P}_{t_j,\new}}'\bm{\Sigma}_{t}\bm{P}_{t_j,\new} \right).
%
\end{eqnarray}

\item Assume that $d\geq(K+2)\alpha$. This also implies that $t_{j+1} - t_{j} > d\geq(K+2)\alpha$. We set $K$ and $\alpha$ in the Theorem. This along with \eqref{anew_small} quantifies ``slow subspace change".

\item Other assumptions: (i) define $t_0:= 1$ and assume that $t_\train \in [t_0, t_1)$;
(ii) for $j=0,1,2, \dots, J$, define
$
r_j: = \rank(\bm{P}_{t_j}) , \ r_{j,\new}:=\rank(\bm{P}_{t_j,\new}), \ r_{j,\old}:=\rank(\bm{P}_{t_j,\old})
$
Clearly, $r_j = r_{j-1}+r_{j,\new} - r_{j,\old}$. Assume that $r_{j,\new}$ is small enough compared to $r_{j,\old}$ so that $r_j \le r$  and $r_{j,\new} \le r_{\new}$ for all $j$ for constants $r$ and $r_\new$. Assume that $r+r_\new < \min(n, t_{j+1}-t_j)$ and $r_\new \leq r_0$. 

\item Since the $\bnu_t$'s are bounded random variables, there exists a $\gamma< \infty$ and a $\gamma_\new \le \gamma$ such that
\[
\max_t \| \P_t{}' \bnu_t \|_2 \le \gamma, \  \max_j  \max_{t \in[t_j,t_j+d]} \| \P_{t_j,\new}{}' \nut\|_{\infty} \le \gamma_\new.
\]
We assume an upper bound on $\gamma_\new$ in the Thoerem. 

\een
\end{sigmodel}

\begin{figure}
\includegraphics[width=\textwidth]{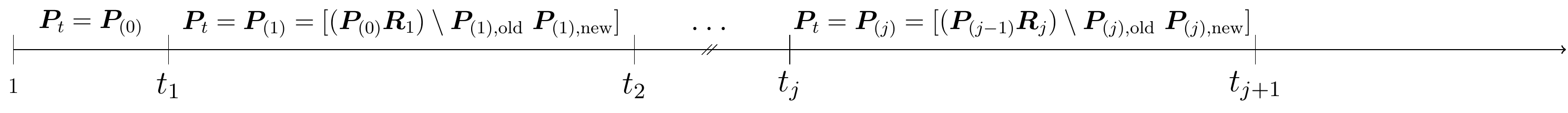}
\vspace{-0.1in}
\caption{A diagram of Model \ref{cor_model} \label{modelfig}}
\vspace{-0.1in}
\end{figure}

A visual depiction of Model \ref{cor_model} is shown in Figure \ref{modelfig}. The above model is similar to the ones introduced in \cite{rrpcp_perf,rrpcp_isit15}. Various low-rank and ``slow changing" models on $\bm{\Sigma}_t$ are special cases of the above model. One interesting special case is one that allows the variance along new directions to increase slowly as follows: for $t \in [t_j, t_{j}+d]$, let $\Lamtnew:= {\bm{P}_{t_j,\new}}'\bm{\Sigma}_{t}\bm{P}_{t_j,\new}$ and assume that $(\bm{\Lambda}_{t,\new})_{i,i} = (v_i)^{t-t_j} q_i  \lambda^-  \ \text{ for } \ i = 1,\dots,r_{j,\new}$. Here $q_i \ge 1$ and $v_i>1$. An upper bound on $v_i$ of the form $q_i (v_i)^d \leq 3$ ensures that \eqref{anew_small} holds.

\begin{remark} \label{new_eval_inc}
Model \ref{cor_model} requires the upper bound on the eigenvalues along the new directions to hold only for the first $d$ time instants after $t_j$. At any time $t > t_j+d$, the eigenvalues along $\P_{t_j,\new}$ could increase to any large value up to $\lambda^+$ either gradually or suddenly.
\end{remark}

%
%

The above model requires the directions to get deleted and added at the same set of times $t=t_j$. This is assumed for simplicity. In general, directions from $\Span(\P_{t_j-1})$ could get deleted at any other time as well. The lower bound in \eqref{anew_small}  requires the energy of $\lt$ along the new directions at {\em all} times $t \in [t_j, t_j+d]$ to be above $\lambda^-$. With very minor changes to the proof (of Lemma \ref{Ak_cor}), we can relax this to the following: we can let $\lambda_\new^-$ be the minimum eigenvalue along the new directions of any $\alpha$-frame {\em average} covariance matrix over the period $[t_j, t_j+d]$ and require this to be larger than $\lambda^-$. For video analytics, this translates to requiring that, after a subspace change, {\em enough (but not necessarily all)} background frames have ``detectable" energy along the new directions, so that the minimum eigenvalue of the average covariance along the new directions is above a threshold. For the recommendation systems' application, this means that the initial set of users may only be influenced by a few, say five, factors, but as more users come in to the system, {\em some (not necessarily all)} of them may also get influenced by a sixth factor (newly added direction).

There is a trade off between the upper bound on $\lambda_\new^+$ in \eqref{anew_small} in Model \ref{cor_model} above and the bound on ${\rrho}^2 \beta$ assumed in Model \ref{sbyrho}. Allowing a larger value of $\lambda_\new^+$ will require a tighter bound on $\rho^2\beta$. We chose one set of bounds, but many other pairs would also work. For video analytics, this means that if the background subspace changes are faster, then we also need the foreground objects to be moving more so we can `see' enough of the background behind them. 


\subsection{Denseness}
To separate sparse $\xt$'s from the $\lt$'s, the basis vectors for the subspace from which the $\lt$'s are generated cannot be sparse. We quantify this using an incoherence condition similar to \cite{rpca}.%
\begin{sigmodel}[Denseness] \label{dense_model}
Let $\mu$ be the smallest real number such that
$
\max_{i} \|{\bm{P}_{t_j}}' \I_i\|_2^2 \leq \frac{\mu r_j}{n}
\ and \
\max_{i} \|{\bm{P}_{t_j,\new}}' \I_i\|_2^2 \leq \frac{\mu r_{j,\new}}{n} \text{ for all } j
$
($\I_i$ is the $i^{\text{th}}$ column of the identity matrix; thus $\P'I_i$ is the $i$-th row of $\P$). 
Assume that
\[
2s r \mu\leq{0.09n} \text{ and  } 2s r_\new  \mu \leq  {0.0004n}.
\]
\end{sigmodel}

\begin{fact}
Model \ref{dense_model} is one way to ensure that $\|{\bm{P}_{t_j}}' \I_\T\|_2 \le 0.3$ and $\|{\bm{P}_{t_j,\new}}' \I_\T\|_2 \le 0.02$ for all sets $\T$ with $|\T|\le 2s$.
This follows using the fact that for an $r \times s$ matrix $M$, $\|M\|_2 \le \sqrt{s} \max_i \|M_i\|_2$ where $M_i$ is the $i$-th column vector of $M$.
\end{fact}

\subsection{Assumption on the unstructured noise $\wt$} 

\begin{sigmodel}\label{wt_model}
Assume that the noise $\wt$ is zero mean, mutually independent over time, and bounded with $\|\wt\|_2 \le \epsilon_w$. 
\end{sigmodel}

\subsection{Main result for Automatic ReProCS} \label{noclust_result} 
In this section, we give a correctness result for Automatic ReproCS, i.e., for Algorithm \ref{reprocsdet} with the cluster PCA (cPCA) step removed. This is exactly the algorithm studied in our earlier work \cite{rrpcp_isit15}. The result given in \cite{rrpcp_isit15} for it required mutual independence of the $\lt$'s over time. For the video application, this means that background changes at different times are due to independent causes, e.g., independent light flickers. This is often a restrictive assumption. The current result replaces this requirement with an autoregressive model which is a much better model for background changes due to correlated factors such as moving lake or sea waters.

The main idea of Automatic ReProCS is as follows.  It estimates the initial subspace as the top $r_0$ left singular vectors of $[\bm{m}_1,\bm{m}_2,\dots, \bm{m}_{t_\train}]$. 
At time $t$, if the previous subspace estimate, $\Phat_{t-1}$, is accurate enough, because of the ``slow subspace change" assumption, projecting $\mt = \xt + \lt + \wt$ onto its orthogonal complement nullifies most of $\lt$. Specifically, we compute $\bm{y}_t:= \bm{\Phi}_t \mt$ where $\bm{\Phi}_t:= \I - \Phat_{t-1}\Phat_{t-1}{}'$. Clearly, $\bm{y}_t = \bm{\Phi}_t \xt + \bm{b}_t$ with $\| \bm{b}_t \|_2$ being small. 
Thus recovering $\xt$ from $\bm{y}_t$ is a traditional sparse recovery problem in small noise \cite{candes_rip}. We recover $\xt$ by $l_1$ minimization with the constraint $\| \bm{y}_t - \bm{\Phi}_t x\|_2 \le \xi$ and estimate its support by thresholding using a threshold $\omega$. We use the estimated support, $\That_t$, to get an improved debiased estimate of $\xt$, denoted $\xhatt$, by least squares (LS) estimation on $\That_t$ \cite{dantzig}. We then estimate $\lt$ as $\lhat_t = \mt - \xhatt$. 
The estimates $\lhatt$ are used in the subspace estimation step which involves (i) detecting subspace change; and (ii) $K$ steps of projection-PCA, each done with a new set of $\alpha$ frames of $\lhatt$, to get an accurate enough estimate of the new subspace. This step is explained in detail later in Sec. \ref{algo_sec}. 
Automatic ReProCS has four algorithm parameters - $\alpha$, $K$, $\xi$, $\omega$ - whose values will be set in the result below.


%

\begin{theorem}\label{thm1_cor_nodel}
Consider Algorithm \ref{reprocsdet} without the cluster PCA step.
Assume that, for $t > t_\train$, $\mt = \lt + \wt + \xt$ and, for $t \le t_\train$,  $\mt = \lt + \wt$.
Pick a $\zeta$ that satisfies 
\[
\zeta \leq \min \left\{ \frac{10^{-4}}{(r_0+Jr_\new)^2} , \frac{0.003\lambda^-}{(r_0+Jr_\new)^2 \lambda^+} , \frac{1}{(r_0+Jr_\new)^3\gamma^2}, \frac{0.05\lambda^-}{(r_0+Jr_\new)^3\gamma^2} \right\}.
\]
Let $b_0=0.1$.
Suppose that the following hold.
\begin{enumerate}

\item enough initial training data is available: $t_\train \ge \frac{32 (2(r_0+Jr_\new) \gamma^2)^2}{ (1-b_0)^2 (0.001 r_\new\zeta \lambda^-)^2} (11 \log n + \log 8) $

\item algorithm parameters are set as:
\\ $\xi =  \xi_{\mathrm{cor}} := \epsilon_w + \frac{2\sqrt{\zeta} + \sqrt{r_{\new}}\gamma_{\new}}{1-b_0}$; $\omega = 7\xi$;  $K = \left\lceil \frac{\log(0.85r_\new\zeta)}{\log(0.2)}\right\rceil$;
\\ $\alpha = \alpha_{\add}$ where
$
\alpha_{\add}  \ge  32 \frac{1.2^2 (2\sqrt{\zeta} + \sqrt{r_\new}\gamma_\new + 2\epsilon_w)^4}{(1-b_0)^6} \frac{(1-b_0^2)^2}{(0.001 r_\new \zeta\lambda^-)^2} (11 \log n + \log ((52K+44)J) ) 
$

\item model on $\T_t$: Model \ref{sbyrho} holds; 

\item model on $\lt$:
\\ Model \ref{cor_model} holds with $\P_{t_j,\new}{}' [\P_0, \P_{t_1,\new}, \P_{t_2,\new}, \dots \P_{t_{j-1},\new}] = \bm{0}$, $b \le b_0=0.1$, and with $\sqrt{r_\new}\gamma_{\new}$ small enough so that $14 \xi \le \min_{t} \min_{i \in \T_t} |(x_t)_i|$;
\\ Model \ref{dense_model} (denseness) holds with $r$ replaced by $(r_0+Jr_\new)$.
\label{xmin_cond}

\item model on $\wt$: Model \ref{wt_model} holds with $\epsilon_w^2 \le 0.03\zeta  \lambda^-$ 

\item independence: Let $\T:=\{\T_{\tilde{t}}\}_{\tilde{t}=1,2,\dots,t_{\max}}$.
Assume that $\T,  \bm{w}_1, \bm{w}_2, \dots, \bm{w}_{t_{\max}}, \bnu_1, \bnu_2, \dots, \bnu_{t_{\max}}$ are mutually independent random variables. 
\label{indep_cond}


\end{enumerate}
Then, with probability  $\ge 1 - 2n^{-10}$, at all times $t$,
\ben
\item $\T_t$ is exactly recovered, i.e. $\hat{\T}_t = \T_t$ for all $t$;
\item  $\|\xt - \xhatt\|_2 \le  1.34 \left(2\sqrt{\zeta} +  \sqrt{r_\new}\gamma_{\rmnew}  + \epsilon_w  \right)$ and $\|\lhat_t - \lt \|_2 \le \|\xt - \xhatt\|_2 + \epsilon_w$;
\item the subspace error $\SE_t := \|( \I - \hat{\bm{P}}_t \hat{\bm{P}}_t{}' ) \Pt \|_2 \le 10^{-2} \sqrt{\zeta}$ for all $t\in[t_{j}+d,t_{j+1})$.
\item the subspace change time estimates satisfy $t_j \leq \hat{t}_j \leq t_j + 2\alpha$; and its estimates of the number of new directions are correct: $\hat{r}_{j,\new,k} = r_{j,\new}$ for $j=1,\dots,J$.

\een
\end{theorem}

{\em Proof: } The above result follows as a corollary of the more general result, Theorem \ref{thm1_cor}, that is given below. For its proof, please see Appendix \ref{proof_thm1_cor_nodel}.

\begin{remark} \label{wt_model_weaker}
Consider condition \ref{indep_cond}). If it is not practical to assume that $\wt$'s are independent of $\T$ (e.g., if $\wt$ contains the smaller magnitude outlier entries and $\xt$ the larger ones and so $\T_t = \text{support}(\xt)$ cannot be independent of $\wt$), the following weaker assumption can be used with small changes to the proof (see Fact \ref{wt_model_weaker_proof} in Sec. \ref{general_decomp_w}). Let $Q:=\{\T,\{\wt\}_{t=1,2,\dots,t_{\max}}\}$. Assume that $Q, \bnu_1, \bnu_2, \dots, \bnu_{t_{\max}}$ are mutually independent.
\end{remark}

Theorem \ref{thm1_cor_nodel} says the following. If an accurate estimate of the initial subspace is available ($t_\train$ is large enough); the algorithm parameters are set appropriately;  the outlier support at time $t$, $\T_t$, has enough changes over time; $\lt$ follows an AR model with parameter $b \le b_0=0.1$ (i.e., the $\lt$'s are not too correlated over time); the low-dimensional subspace from which $\bnu_t$ is generated (this is also approximately the subspace from which $\lt$ is generated) is fixed or changes ``slowly" enough, i.e. (i) the delay between change times is large enough ($t_{j+1}-t_j > d \ge (K+2)\alpha$) and (ii) the eigenvalues along the newly added directions are small enough for $d$ frames after a subspace change; the basis vectors whose span defines the low-dimensional subspaces are dense enough; the noise $\wt$ is small enough; then, with high probability (whp), the error in estimating $\lt$ or $\xt$ will be bounded by a small value at all times $t$. Also, whp, the outlier support will be exactly recovered at all times; and the error in estimating the low-dimensional subspace will decay to a small constant times $\sqrt{\zeta}$ within a finite delay of a subspace change.  Moreover, subspace changes will get detected within a short delay, and the dimension of the newly added subspaces will get correctly estimated.

The condition ``$14 \xi \le \min_{t} \min_{i \in \T_t} |(x_t)_i|$" in condition \ref{xmin_cond}) can be interpreted either as another slow subspace change condition or as a requirement that the minimum magnitude nonzero entry of $\xt$ (the smallest magnitude outlier) be large enough compared to $\epsilon_w + \sqrt{r_\new} \gamma_\new$. Interpreted this way, it says the following. If $\lt$ is the true data, $\mt-\lt = \wt+\xt$ is the vector of corruptions with $\wt$ being the small corruptions and the nonzero entries of $\xt$ being the large ones (outliers). We need $\wt$ to be small enough to not affect subspace recovery error too much ($\|\wt\|_2 \le \epsilon_w \le \sqrt{0.03 \zeta \lambda^-}$) and we need the nonzero entries of $\xt$ to be large enough  to be detectable ($\min_{t} \min_{i \in \T_t} |(x_t)_i| \ge 14 \xi \approx 14(\epsilon_w + \sqrt{r_\new} \gamma_\new)$). 


\subsection{Eigenvalues' clustering assumption and main result for Automatic ReProCS-cPCA}  \label{clust_result}

The ReProCS algorithm studied above (which is the same as the one introduced in \cite{rrpcp_isit15}) does not include a step to delete old directions from the subspace estimate. As a result, its estimated subspace dimension can only increase over time. This necessitates a bound on the number of subspace changes, $J$. The bound is imposed by the denseness assumption - notice that Theorem \ref{thm1_cor_nodel} requires the bound in Model \ref{dense_model} to hold with $r$ replaced by $r_0+Jr_\new$. In this section, we relax this requirement by analyzing automatic ReProCS-cPCA (Algorithm \ref{reprocsdet}) which includes cluster PCA to delete the old directions from the subspace estimate. This is done by re-estimating the current subspace.

In order to be able to design an accurate algorithm to delete the old directions by re-estimating the current subspace, we need one of the following  for a period of $d_2$ frames within the interval $[t_j,t_{j+1})$. We either need the condition number of $\Lamt$ (or equivalently of $\bm\Sigma_{t}$) to be small, or we need a generalization of it: we need its eigenvalues to be ``clustered" into a few (at most $\vartheta$) clusters in such a way that the condition number within each cluster is small and the distance between consecutive clusters is large (clusters are well separated). 
The problem with requiring a small upper bound on the condition number of $\bm\Sigma_t$ is that it disallows situations where the $\lt$'s constitute large but structured noise. This is why the ``clustered" generalization is needed. This would be valid for data that has variations at different scales. For example, for data that has variations at two scales, there would be two clusters, the large scale variations would form the first cluster and the small scale ones the second cluster. These clusters would naturally be well separated.

Let $\vartheta$ denote the maximum number of clusters. As we will explain in Sec. \ref{algo_sec}, the subspace deletion via re-estimation step is done after the new directions are accurately estimated. As explained later, with high probability (whp), this will not happen until $t_j+K \alpha$. Thus, we assume that the clustering assumption holds for the period $[t_j+K \alpha+1, t_j+K \alpha+d_2]$ with $d_2 > (\vartheta+3)\alpha$ and $t_{j+1}-t_j > K \alpha + d_2$. In the algorithm, cluster PCA is done starting at $\that_j + K \alpha$.

\begin{sigmodel} \label{clust_model}
Assume the following.
\ben
\item Assume that $t_{j+1} - t_j > K \alpha + d_2$ for an integer $d_2 \ge (\vartheta+3)\alpha$ (where $\vartheta$ is defined below). Assume that for all $t \in [t_j+K \alpha+1, t_j+K \alpha+d_2]$, $\Lamt$ is constant; let $\Lamj$ be this constant matrix and assume that $\lambda_{\min}(\Lamj) \ge \lambda^-$.

\item Define a partition of the index set $\{1,2, \dots r_{j} \}$ into sets $\mathcal{G}_{j,1}, \mathcal{G}_{j,2}, \dots , \mathcal{G}_{j,\vartheta_j}$ as follows. Sort the eigenvalues of $\bm{\Lambda}_{(j)}$ in decreasing order of magnitude. To define $\mathcal{G}_{j,1}$, start with the first (largest) eigenvalue and keep adding smaller eigenvalues to the set. Stop when the ratio of the maximum to the minimum eigenvalue first exceeds $g^+ = 3$ or when there are no more nonzero eigenvalues.
Suppose this happens for the $i$-th eigenvalue. Then, define $\mathcal{G}_{j,1} = \{1,2, \dots i-1\}$. For  $\mathcal{G}_{j,2}$, start with the $i$-th eigenvalue and repeat the same procedure. Keep doing this until there are no more nonzero eigenvalues. Let $\vartheta_j$ denote the number of clusters for the $j$-th subspace and let $\vartheta:= \max_j \vartheta_j$.
Define
\[
\lambda_{j,k}^+:= {\max_{i\in\mathcal{G}_{j,k}}\lambda_i\left(\bm{\Lambda}_{(j)}\right)}, \ \ \lambda_{j,k}^-:= {\min_{i\in\mathcal{G}_{j,k}}\lambda_i\left(\bm{\Lambda}_{(j)}\right)}
\]

Assume that the clusters are well-separated, i.e., 
 \begin{equation}\label{chi}
\frac{\lambda_{j,k+1}^+}{\lambda_{j,k}^-} \le {\chi}^+ = 0.2
 \end{equation}
\een
\end{sigmodel}

\begin{fact}
The above way of defining the clusters is one way to ensure that the condition number of the eigenvalues within each cluster (ratio of the maximum to minimum eigenvalue of the cluster) is below $g^+ = 3$, i.e.,
for all $k=1,2,\dots, \vartheta_j$,
\begin{equation}\label{gjk}
\frac{\lambda_{j,k}^+}{\lambda_{j,k}^-} \le g^+ = 3.
\end{equation}
\end{fact}

A model similar to Model \ref{clust_model} was first introduced in \cite{rrpcp_perf} where the cluster PCA idea was introduced.

\begin{remark}
The case when, for the entire period  $[t_j+K \alpha+1, t_j+K \alpha+d_2]$, the condition number of $\bm\Sigma_t$ is below $g^+$ is a special case of Model \ref{clust_model} with $\vartheta_j=\vartheta=1$ and $\chi^+ = 0$.
\end{remark}

\begin{remark}\label{new_eval_inc_2}
Model \ref{cor_model} requires the eigenvalues along $\P_{t_j,\new}$ to be small for $t \in [t_j, t_j+d]$ with $d \ge (K+2)\alpha$ while  Model \ref{clust_model} requires all eigenvalues to be constant for $t \in [t_j+K \alpha+1, t_j+K \alpha+d_2]$.
Taken together, this means that for all $t \in [t_j, t_j+K \alpha+d_2]$, we are requiring that the eigenvalues along $\P_{t_j,\new}$ be small. However after $t=t_j+K \alpha+d_2$, there is no constraint on its eigenvalues until $t=t_{j+1}+K \alpha$ at which time Model \ref{clust_model} again requires all eigenvalues to be constant.
Thus, in the interval $[t_j+K \alpha+d_2+1, t_{j+1}+K \alpha]$, or in later intervals of the form  $[t_{j+j'}+K \alpha+d_2+1, t_{j+j'+1}+K \alpha]$ for any $j'>0$, the eigenvalues along $\P_{t_j,\new}$ could increase to any large value up to $\lambda^+$ either gradually or suddenly. Or they could also decrease to any small value.
\end{remark}


With small changes to the proof, one can relax the $\Lamt$ constant requirement to the following. Let $\mathrm{ClustInterval}$ denote the interval $[t_j+K\alpha+ 1, t_j+K\alpha +d_2]$ 
and let $t_0$ denote the first time instant of $\mathrm{ClustInterval}$.
Define a partition of the index set $\{1,2, \dots r_{j} \}$ into sets $\mathcal{G}_{j,1}, \mathcal{G}_{j,2}, \dots , \mathcal{G}_{j,\vartheta_j}$ as in Model \ref{clust_model} but by using $\bm\Lambda_{t_0}$ to replace $\Lamj$.
Assume that for all $k=1,2,\dots, \vartheta_j$,
$\lambda_{j,k}^- \le \min_{i \in \mathcal{G}_{j,k}} \min_{t \in \mathrm{ClustInterval}} \lambda_i( \bm{\Lambda}_{t}) \le \max_{i \in \mathcal{G}_{j,k}} \max_{t \in \mathrm{ClustInterval}} \lambda_i( \bm{\Lambda}_{t}) \le \lambda_{j,k}^+.$

At the cost of making our model more complicated, the requirement discussed in Remark \ref{new_eval_inc_2} can also be relaxed, i.e., we can allow the eigenvalues along $\P_{t_j,\new}$ to increase to a large value before imposing Model \ref{clust_model}. To do this we need to assume an upper bound on $d$. Suppose that $(K+2)\alpha \le d \le (K+3)\alpha$. Suppose also that we allow a period of $\Delta = 4 \alpha$ frames for the new eigenvalues to increase. We can assume Model \ref{clust_model} holds for the period $[t_j+K\alpha + 3\alpha+ \Delta+ 1, t_j+K\alpha + 3\alpha+ \Delta +d_2]$ with $d_2> (\vartheta+3)\alpha$. In addition, we would also need $t_{j+1}-t_j >(K+3)\alpha+\Delta+d_2$. With this, we would run the cluster PCA algorithm starting at $\that_j+K\alpha + 3\alpha+ \Delta$ instead of at $\that_j+K\alpha$ as we do now.

We give below a correctness result for Automatic ReproCS-cPCA (Algorithm \ref{reprocsdet}) that uses the above model. It has one extra parameter, $\ghatp$, other than the four used by Automatic ReProCS. $\ghatp$ is used to estimate the eigenvalue clusters automatically from an empirical covariance matrix computed using an appropriate set of $\lhatt$'s.


\begin{theorem}\label{thm1_cor}
Consider Algorithm \ref{reprocsdet}.
%
Assume that, for $t > t_\train$, $\mt = \lt + \wt + \xt$ and, for $t \le t_\train$,  $\mt = \lt + \wt$.
Pick a $\zeta$ that satisfies 
\[
\zeta \leq \min \left\{ \frac{10^{-4}}{(r+r_\new)^2} , \frac{0.003\lambda^-}{(r+r_\new)^2 \lambda^+} , \frac{1}{(r+r_\new)^3\gamma^2}, \frac{0.05\lambda^-}{(r+r_\new)^3\gamma^2} \right\}.
\]
Let $b_0=0.1$.
Suppose that the following hold.
\begin{enumerate}

\item enough initial training data is available: $t_\train \ge \frac{32 (2r \gamma^2)^2}{ (1-b_0)^2 (0.001 r_\new\zeta \lambda^-)^2} (11 \log n + \log 8)$

\item algorithm parameters are set as:
\\ $\xi =  \xi_{\mathrm{cor}} := \epsilon_w + \frac{2\sqrt{\zeta} + \sqrt{r_{\new}}\gamma_{\new}}{1-b_0}$; $\omega = 7\xi$;  $K = \left\lceil \frac{\log(0.85r_\new\zeta)}{\log(0.2)}\right\rceil$; $\ghatp: =  \frac{g^+ + 0.06}{1-0.06}  = 3.26$;
\\ $\alpha = \max\{\alpha_{\add},\alpha_{\del}\}$ where
$
\alpha_{\add}  \ge  32 \frac{1.2^4 (2\sqrt{\zeta} + \sqrt{r_\new}\gamma_\new + 2\epsilon_w)^4}{(1-b_0)^6} \frac{(1-b_0^2)^2}{(0.001 r_\new \zeta\lambda^-)^2} (11 \log n + \log ( (52K+44)J) )
$
and
$\alpha_{\del} \ge 32 \frac{1.2^4 r^2\gamma^4}{(1-b_0)^6} \frac{(1-b_0^2)^2}{(0.001 r_\new \zeta\lambda^-)^2} (11 \log n + \log ( (52\vartheta+36)J) );$


\item model on $\T_t$: Model \ref{sbyrho} holds; 


\item model on $\lt$:
\\ Model \ref{cor_model} holds with $b \le b_0=0.1$ and with $\sqrt{r_\new}\gamma_{\new}$ small enough so that $14 \xi \le \min_{t} \min_{i \in \T_t} |(x_t)_i|$;
\\ Model \ref{clust_model} holds with $|\mathcal{G}_{j,k}| \ge 0.15 (r+r_\new) $; 
\\ Model \ref{dense_model} (denseness) holds. 
\label{xmin_cond}

\item model on $\wt$: Model \ref{wt_model} holds with $\epsilon_w^2 \le 0.03\zeta  \lambda^-$ 
\label{wt_cond}

\item independence: Let $\T:=\{\T_{\tilde{t}}\}_{\tilde{t}=1,2,\dots,t_{\max}}$. Assume that $\T,  \bm{w}_1, \bm{w}_2, \dots, \bm{w}_{t_{\max}}, \bnu_1, \bnu_2, \dots, \bnu_{t_{\max}}$ are mutually independent random variables.

\end{enumerate}
Then, with probability  $\ge 1 - 3n^{-10}$, at all times $t$,
\ben
\item $\T_t$ is exactly recovered, i.e. $\hat{\T}_t = \T_t$ for all $t$;

\item $\|\xt - \xhatt\|_2 \leq 1.34 \left(2\sqrt{\zeta} +  \sqrt{r_\new}\gamma_{\rmnew}  + \epsilon_w  \right)$ and $\|\lhat_t - \lt \|_2 \le \|\xt - \xhatt\|_2 + \epsilon_w$;

\item the subspace error $\SE_t := \|( \I - \hat{\bm{P}}_t \hat{\bm{P}}_t{}' ) \Pt \|_2 \le 10^{-2} \sqrt{\zeta}$ for all $t\in[t_{j}+d,t_{j+1})$;

\item the subspace change time estimates given by Algorithm \ref{reprocsdet} satisfy $t_j \leq \hat{t}_j \leq t_j + 2\alpha$;

\item its estimates of the number of new directions are correct: $\hat{r}_{j,\new,k} = r_{j,\new}$ for $j=1,\dots,J$; 

\item  eigenvalue clusters are recovered exactly: $\hat{\mathcal{G}}_{j,k} = \mathcal{G}_{j,k}$ for all $j$ and $k$; thus its estimate of the number of deleted directions is also correct.
\een
\end{theorem}

{\em Proof: } The proof outline is given in Section \ref{outline}. The proof is given in Sections \ref{pf_thm1_cor}, \ref{3_pfs}, \ref{3_pfs_del}.

\begin{remark}
Notice that the lower bound $|\mathcal{G}_{j,k}| \ge 0.15 (r+r_\new) $ can hold only if the number of clusters $\vartheta_j$ is at most 6. This is one choice that works along with the given bounds on other quantities such as $\rho^2 \beta$. It can be made larger if we assume a tighter bound on $\rho^2 \beta$ for example. But what will remain true is that our result requires the number of clusters to be $\bigo(1)$.
\end{remark}

\begin{remark} The independence assumption can again be replaced by the weaker one of Remark \ref{wt_model_weaker}.  \end{remark}

The extra assumption needed by the above result compared to Theorem \ref{thm1_cor_nodel} is the clustering one. Using this, ReProCS-cPCA is able to correctly estimate the current subspace. Thus,for $t \in [t_j, \that_j+\alpha]$, $\Phat_{t-1}$ is an accurate estimate of $\Span(\P_{t_j-1})$ where as when using ReProCS (and Theorem \ref{thm1_cor_nodel}), it is an estimate of $\Span([\P_0, \P_{t_1,\new}, \P_{t_2,\new}, \dots \P_{t_{j-1},\new}])$. Because of this, (i) the above result needs a much weaker denseness assumption, (ii) it does not need a bound on $J$, and (iii) it requires the new directions to only be orthogonal to $\Span(\P_{t_j-1})$.
We discuss the results in detail in Sec. \ref{discussion}.

\begin{corollary} \label{thm1_cor_corol}
The following conclusions also hold under the assumptions of Theorem \ref{thm1_cor} with probability at least $1 - 3n^{-10}$.
\begin{enumerate}
\item The recovery error satisfies $\|\lhat_t - \lt \|_2 \leq \|\xt - \xhatt\|_2 + \epsilon_w$ and
\begin{align*}
\|\xt - \xhatt\|_2 \le
\begin{cases}
1.34 \left(2\sqrt{\zeta} +  \sqrt{r_{\new}}\gamma_{\rmnew} + \epsilon_w \right) & t \in \left[t_j, \that_j + \alpha \right] \\
1.34 \left(2.15\sqrt{\zeta} +  0.19\cdot(0.1)^{k-1}\sqrt{r_{\new}}\gamma_{\rmnew} + \epsilon_w \right) &   t \in \left[ \that_j+(k-1)\alpha + 1,  \that_j + k\alpha\right], \ k=2,3, \dots, K \\
2.67(\sqrt{\zeta} + \epsilon_w) & t \in \left[\that_j + K \alpha + 1,  \hat{t}_j + K \alpha + (\vartheta+1)\alpha \right] \\
2.67(\frac{r}{r+r_\new}\sqrt{\zeta} + \epsilon_w) & t \in \left[\that_j + K \alpha + (\vartheta+1)\alpha + 1, t_{j+1} - 1\right];
\end{cases}
\end{align*}
\item The subspace error satisfies,
\begin{align*}
\SE_t \leq
\begin{cases}
1  & t \in \left[t_j,\that_j + \alpha \right] \\
10^{-2} \sqrt{\zeta} +  0.19\cdot0.1^{k-1} & t \in \left[\that_j+(k-1)\alpha + 1,  \that_j + k\alpha \right], \ k=2,3, \dots, K  \\
10^{-2} \sqrt{\zeta}   &  t \in \left[\that_j + K \alpha + 1,  \hat{t}_j + K \alpha + (\vartheta+1)\alpha \right] \\
10^{-2} \frac{r}{r+r_\new}\sqrt{\zeta}   &  t \in \left[\that_j + K \alpha + (\vartheta+1)\alpha + 1, t_{j+1} - 1\right];
\end{cases}
\end{align*}
\end{enumerate}
\end{corollary}



{\bf Online matrix completion (MC). }
MC can be interpreted as a special case of RPCA and hence the same is true for online MC and online RPCA \cite{rpca,rrpcp_isit15}. In \cite{rrpcp_isit15}, we explicitly stated results for both. In a similar fashion, an analog of either of the above results can also be obtained for online MC.

{\bf Offline RPCA. }
In certain applications such as video analytics, an improved offline estimate of both the background and the foreground is desirable. In some other applications, there is no real need for an online solution. We show here that, with a delay of at most $(K+2)\alpha$ frames, by using essentially the same ReProCS algorithm with one extra step, it is possible to recover $\xt$ and $\lt$ with close to zero error.
\begin{corollary}[Offline RPCA]
Consider the estimates given in the last two lines of Algorithm \ref{reprocsdet}.
Under the assumptions of Theorem \ref{thm1_cor}, with probability at least $1 - 3n^{-10}$, at all times $t$,
$
\|\xt - \xhatt^{\mathrm{offline}}\|_2
\le 
2.67 (\sqrt{\zeta} +  \epsilon_w)
$, $\|\lhat_t^{\mathrm{offline}} - \lt \|_2 \le 2.67 (\sqrt{\zeta} +  2\epsilon_w)$,
and all its other conclusions hold.
\end{corollary}

Observe that the offline recovery error can be made smaller and smaller by reducing $\zeta$ (this, in turn, will result in an increased delay between subspace change times).
As can be seen from the last two lines of Algorithm \ref{reprocsdet}, the offline estimates are obtained at $t=\that_j+K \alpha$. Since $\that_j \le t_j+2 \alpha$, this means that the offline estimates are obtained after a delay of at most $(K+2)\alpha$ frames.

\subsection{Discussion} \label{discussion}

{\bf Online versus offline. }
We analyze an online algorithm that is faster and needs less storage. It needs to store only a few $n \times \alpha$ or $n \times r$ matrices, while  PCP needs to store matrices of size $n \times t_{\max}$.
Other results for online algorithms include correctness results from \cite{rrpcp_icassp15,rrpcp_isit15} (discussed below), and partial results of Qiu et al. \cite{rrpcp_perf} and Feng et al. \cite{xu_nips2013_1}.
In \cite{xu_nips2013_1}, Feng et al. proposed a method for online RPCA and proved a partial result for their algorithm. Their approach was to reformulate the PCP program and to use this reformulation to develop a recursive algorithm that converged asymptotically to the solution of PCP as long as the basis estimate $\hat{\bm{P}}_t$ was full rank at each time $t$. Since this result assumed something about an intermediate algorithm estimate, $\hat{\bm{P}}_t$, it was  a {\em partial} result.
In \cite{rrpcp_perf}, Qiu et al. obtained a performance guarantee for ReProCS and ReProCS-cPCA that also needed intermediate algorithm estimates to satisfy certain properties. In particular, their result required that the basis vectors for the currently unestimated subspace, $\Span((I - \Phat_{t-1}\Phat_{t-1}{}')\P_{t_j,\new})$, be dense vectors. Thus, their result was also a  {\em partial} result.
In the current work, we remove this requirement and provide a {\em correctness result} for both ReProCS and ReProCS-cPCA. The assumption that helps us get this is Model \ref{sbyrho} on $\T_t$ (or its generalization given in Model \ref{general_model} later). Secondly, unlike \cite{rrpcp_perf}, we provide a correctness result for an automatic algorithm that does not assume knowledge of subspace change times, number of directions added or removed, or of the eigenvalue-based subspace clusters. Thirdly, we allow the $\lt$'s to follow an AR model where as \cite{rrpcp_perf} required independence over time.

To our knowledge, our work and \cite{rrpcp_icassp15,rrpcp_isit15} are the only correctness results for an online RPCA method. 
Our work significantly improves upon the results of \cite{rrpcp_icassp15,rrpcp_isit15}. We allow the $\lt$'s to be correlated over time and use a first order AR model to model the correlation. As discussed earlier, this is significantly more practically valid than the independence assumption used in \cite{rrpcp_icassp15,rrpcp_isit15}. It includes independence as a special case. Moreover, with the extra clustering assumption, we are able to analyze Automatic ReProCS-cPCA in Theorem \ref{thm1_cor}. It needs a much weaker rank-sparsity assumption than what is needed by the result of \cite{rrpcp_isit15}, and it does not need a bound on $J$. 
We discuss this below.

{\bf Bounds on rank and sparsity. }
Let $\L:=[\ell_1, \ell_2 \dots \ell_{t_{\max}}]$, $\S:=[x_1, x_2 \dots x_{t_{\max}}]$, $r_{\text{mat}}:=\rank(\L)$ and let $s_{\text{mat}}$ be the number of nonzero entries in $\S$. With our models, $s_{\text{mat}} \le s t_{\max}$ and $r_{\text{mat}} \le r_0 + J r_\new$ with both bounds being tight. Models \ref{sbyrho} and \ref{dense_model} constrain $s$ and $s,r,r_\new$ respectively. Model \ref{sbyrho} needs $s \le \rho_2 n /\alpha$ and Model \ref{dense_model} needs $r s \in \bigo(n)$ and $r_\new s \in \bigo(n)$.
Using the expression for $\alpha$, it is easy to see that if $J \in \bigo(n)$, $r_\new \in \bigo(1)$ and $r \in \bigo(\log n)$, then
$\frac{1}{\alpha} \in \bigo(\frac{\zeta^2}{r^2 \log n}) = \bigo(\frac{1}{(\log n)^9})$
Alternatively, if $r \in \bigo(1)$, then $\frac{1}{\alpha} \in \bigo(\frac{1}{\log n})$.
Thus, Theorem \ref{thm1_cor} definitely holds in two regimes of interest. The first is $J \in \bigo(n)$, $r_\new \in \bigo(1)$, $r \in \bigo(\log n)$, $s_{\text{mat}} \in \bigo(\frac{n t_{\max}}{(\log n)^9})$ and $r_{\text{mat}} \in \bigo(n)$.  The second is $J \in \bigo(n)$, $r_\new \in \bigo(1)$, $r \in \bigo(1)$, $s_{\text{mat}} \in \bigo(\frac{n t_{\max}}{\log n})$ and $r_{\text{mat}} \in \bigo(n)$. The second regime is more favorable when comparing bounds on $s_{\text{mat}}$ and $r_{\text{mat}}$, but, it also implies that the dimension of the subspace at any given time is $\bigo(1)$. This can be restrictive. The first regime allows the subspace dimension at any time to be $\bigo(\log n)$ which is more reasonable, but, because of this, it needs a tighter bound on $s$ and hence on $s_{\text{mat}}$.

In either regime, our requirements are weaker than those of the PCP results from \cite{rpca2,rpca_zhang}: they need $r_{\text{mat}} s = \bigo(n)$ which implies  $r_{\text{mat}} s_{\text{mat}} \in \bigo(n t_{\max})$; thus if $s_{\text{mat}} \in \bigo(\frac{n t_{\max}}{\log n})$, they would require $r_{\text{mat}}$ to be $\bigo(\log n)$. In the first regime, our conditions are slightly stronger than those of the PCP result from \cite{rpca} while in the second, they are comparable: \cite{rpca} needs $r_{\text{mat}} \in \bigo (\frac{n}{(\log n)^2})$ and $s_{\text{mat}} \in \bigo(n t_{\max})$.

Either set of requirements for Theorem \ref{thm1_cor} is significantly weaker than what is needed by Theorem \ref{thm1_cor_nodel} or by the results of \cite{rrpcp_icassp15,rrpcp_isit15}: both need $r_{\text{mat}} \in \bigo(\log n)$. This is because both analyze ReProCS without the cluster PCA based subspace deletion step. Suppose that $r_{j,\new}=r_\new$ for each $j$. For ReProCS without cluster PCA, this means that the dimension of the estimated subspace grows by $r_\new$ with each subspace change time. Thus, the maximum dimension of the estimated subspace is $r_{\text{mat}} = r_0 + J r_\new$ and this is what was used in place of $r$ in the denseness assumption as well in the bound on $\zeta$. This is why these results need $r_{\text{mat}}$ to be $\bigo(\log n)$. However, in Theorem \ref{thm1_cor}, we analyze ReProCS with cluster PCA. Cluster PCA is used to re-estimate the current subspace and thus effectively delete the subspace corresponding to the old directions. This ensures that the rank of the estimated subspace is also bounded by the rank of the true subspace at any time, i.e. by $r$. Thus, Theorem \ref{thm1_cor} only needs $r \in \bigo(\log n)$ while $r_{\text{mat}}$ can as large as $\bigo(n)$.

{\bf No bound on the number of subspace changes, $J$. }
Notice that the result for ReProCS-cPCA given in Theorem \ref{thm1_cor} does not require an upper bound on the number of subspace changes, $J$. On the other hand, the results for ReProCS (both Theorem \ref{thm1_cor_nodel} and the results from \cite{rrpcp_icassp15,rrpcp_isit15}) require a bound on $J$ that is imposed by the denseness assumption: they need $(r_0 + J r_\new) 2s \mu \le 0.09 n$. All results for PCP need a bound on $r_{\mathrm{mat}}$. Under our model of subspace change, $r_{\mathrm{mat}}$ is at most $r_0+J r_\new$ with the bound being tight and hence the PCP results also need a bound on $J$. Of course, even for Theorem \ref{thm1_cor}, $J$ does affect bounds on other quantities: the result needs $t_{j+1}-t_j > d > K \alpha + (\vartheta+3) \alpha$ where $\alpha$ is an algorithm parameter that depends linearly on $\log J$. Thus, for any given value of $J$, the delay between subspace change times, $t_{j+1}-t_j$, and the duration for which the eigenvalues along the new directions need to be small (quantified in \eqref{anew_small}), $d$, need to grow as $\log J$.

{\bf Assumptions on how often the outlier support $\T_t$ needs to change. }
An important advantage of our work over PCP and other batch methods is that we allow more correlated changes of the set of outliers over time. From the assumption on $\T_t$, it is easy to see that we allow the number of outliers per row of $\L$ to be $\bigo(t_{\max})$, as long as the sets follow Model \ref{sbyrho}\footnote{In a period of length $\alpha$, the set $\T_t$ can occupy index $i$ for at most $\rrho\beta$ time instants, and this pattern is allowed to repeat every $\alpha$ time instants. So an index can be in the support for a total of $\rrho\beta\frac{t_{\max}}{\alpha}$ time instants and the model assumes $\rrho \beta \le \frac{0.0001 \alpha}{\rrho}$ for a constant $\rrho$. Thus an index $i$ can be part of the support $\T_t$ for at most $\frac{0.0001}{\rho} t_{\max} \in \bigo(t_{\max})$ time instants.}. 
This is the same as what our previous results \cite{rrpcp_icassp15,rrpcp_isit15} also allowed.
On the other hand, the PCP results from \cite{rpca2,rpca_zhang} need this number to be $\bigo(\frac{t_{\max}}{r_{\text{mat}}})$ which is stronger. 
The PCP result from \cite{rpca} needs that the set $\cup_{t=1}^{t_{\max}} \T_t$ should be generated uniformly at random which is even stronger. 




{\bf Other assumptions. } 
The above advantages are obtained because we use extra assumptions on $\lt$. We assume (i) accurate knowledge of the initial subspace (or available outlier free data from which this can be obtained), (ii) slow subspace change as quantified by \eqref{anew_small} and the lower bound on the delay between subspace change times, and (iii) for a period of time after the previous subspace change has stabilized, we assume that the eigenvalues along the various subspace directions can be clustered into a few clusters. The result of \cite{rrpcp_isit15} required (i) and (ii) but not (iii). On the other hand, the PCP results \cite{rpca,rpca2,rpca_zhang} do not need any of the above. But they need other extra assumptions. They require denseness of the right singular vectors of $\L$ and a bound on the maximum absolute entry of the matrix $U V'$ where $U$ is the matrix of left singular vectors of $\L$ and $V$ is the matrix of its right singular vectors. In our notation $\Span(U) = \Span([P_0, P_{1,\new}, \dots P_{J,\new}])$. We assume denseness of $U$ but not of the right singular vectors.

{\bf Setting algorithm parameters. }
Our result needs five algorithm parameters to be appropriately set. Some of these require knowing at least an upper bound on the model parameters. Our result needs to know upper bounds on $\gamma, \gamma_\new$, $r_0, r, r_\new$, $b$, and $g^+$. The PCP results need this for none \cite{rpca} or at most one \cite{rpca2,rpca_zhang} algorithm parameter.
We briefly explain in Sec. \ref{param_set} how to set algorithm parameters automatically for practical experiments.

{\bf Other work. }
A recent work that uses knowledge of the initial subspace estimate but performs recovery in a piecewise batch fashion is modified-PCP \cite{zhan_pcp}. Like PCP, the result for modified PCP also needs uniformly randomly generated support sets which is stronger than what we need. But, like PCP, it does not need the other extra assumptions that ReProCS needs. Another somewhat related work is the algorithm and correctness result of Feng et al. \cite{xu_nips2013_2} on online PCA with contaminated data. This does not model the outlier as a sparse vector but defines anything that is far from the data subspace as an outlier. 



%
%

\begin{algorithm}[t!] 
\caption{Automatic ReProCS-cPCA}\label{reprocsdet}
{\bf Parameters}: $\alpha$, $K$, $\xi$, $\omega$, $\ghatp$, \ {\bf Inputs}: $\bm{m}_t$ for each $t$, \ {\bf Output}: $\xhatt$, $\lhatt$,  $\Phat_{t}$, $\that_{\jhat}$, $\hat{r}_{\jhat,\new,k}$, $\hat{G}_{j,k}$

Compute $\lamtrain$ as the $r_0$-th eigenvalue of $\frac{1}{t_\train} \sum_{t=1}^{t_\train} \mt \mt'$ and $\Phat_{t_\train}$ as its top $r_0$ eigenvectors.

Set $\thresh=\frac{\lamtrain}{2}$. 
Set $\Phat_{t,*} \leftarrow \Phat_{t_\train}$, $\Phat_{t,\new} \leftarrow [.]$, $\jhat \leftarrow 0$,  $\mathrm{phase} \leftarrow \mathrm{detect}$.

For every $t > t_\train$, do
\ben

\item {\bf Estimate $\T_t$ and $\xt$: }
\ben
\item \label{othoproj} compute $\bm{\Phi}_{t} \leftarrow \bm{I} - \Phat_{t-1} \Phat_{t-1}{}'$ and $\bm{y}_t \leftarrow \bm{\Phi}_{t} \bm{m}_t$

\item \label{Shatcs} solve $\min_{\bm{x}} \|\bm{x}\|_1 \ s.t. \ \|\bm{y}_t - \bm{\Phi}_{t} \bm{x}\|_2 \leq \xi$ and let $\hat{\bm{x}}_{t,\cs}$ denote its solution
\item \label{That} compute $\hat{\mathcal{T}}_t = \{i: \ |(\hat{\bm{x}}_{t,\cs})_i| > \omega\}$
\item \label{LS} LS: compute $\hat{\bm{x}}_t= \I_{\hat{\mathcal{T}}_t} ((\bm{\Phi}_t)_{\hat{\mathcal{T}}_t})^{\dag} \bm{y}_t$
\een

\item {\bf Estimate $\bm{\ell}_t$: }
$\hat{\bm{\ell}}_t \leftarrow \bm{m}_t - \hat{\bm{x}}_t $ 

\item {\bf Subspace Update: }

If $t \mod \alpha \neq 0$ then $\Phat_{t,*} \leftarrow \Phat_{t-1,*}$, $\Phat_{t,\new} \leftarrow \Phat_{t-1,\new}$, $\Phat_{t} \leftarrow [\Phat_{t,*} \ \Phat_{t,\new}]$

 If $t \mod \alpha = 0$  then 
\\
{\bf if  $\mathrm{phase} = \mathrm{detect}$ then}
\begin{enumerate}

\item \label{detect} Set $ u = \frac{t}{\alpha}$ and compute $\bm{\mathcal{D}}_u = (\I - \Phat_{u\alpha-1,*} \Phat_{u\alpha-1,*}{}') [\lhat_{(u-1)\alpha+1}, \dots \lhat_{u\alpha}]$

\item $\Phat_{t,*} \leftarrow \Phat_{t-1,*}$, $\Phat_{t,\new} \leftarrow \Phat_{t-1,\new}$, $\Phat_{t} \leftarrow [\Phat_{t,*} \ \Phat_{t,\new}]$ 

\item \label{change} If $\lambda_{\max} (\frac{1}{\alpha} \bm{\mathcal{D}}_u \bm{\mathcal{D}}_u{}'  ) \ge \thresh$  then
\begin{enumerate}
\item $\mathrm{phase} \leftarrow \mathrm{pPCA}$, $\jhat \leftarrow \jhat+1$, $k \leftarrow 0$, $\that_{\jhat} = t$
\end{enumerate}

\end{enumerate}
{\bf else if $\mathrm{phase} = \mathrm{pPCA}$ then}
\begin{enumerate}
\item Set $ u = \frac{t}{\alpha}$ and compute
$\bm{\mathcal{D}}_u = (\I - \Phat_{u\alpha-1,*} \Phat_{u\alpha-1,*}{}') [\lhat_{(u-1)\alpha+1}, \dots \lhat_{u\alpha}]$

\item\label{PCA} $\Phat_{t,\new} \leftarrow \text{eigenvectors}\left(\frac{1}{\alpha} \bm{\mathcal{D}}_u \bm{\mathcal{D}}_u{}',\thresh\right)$,  $\Phat_{t,*} \leftarrow \Phat_{t-1,*}$, $\Phat_{t} \leftarrow [\Phat_{t,*} \ \Phat_{t,\new}]$

\item  $k \leftarrow k+1$, set $\hat{r}_{j,\new,k} = \rank(\Phat_{t,\new})$

\item If $k == K$, then
\ben \item $\mathrm{phase} \leftarrow \mathrm{cPCA}$,  reset $k \leftarrow 0$ \een 
\end{enumerate}
{\bf else  if $\mathrm{phase} = \mathrm{cPCA}$ then}
\ben
\item cluster PCA (summarized in Algorithm \ref{clusterPCA}); 
\item set $\Phat_{t,*} \leftarrow \Phat_{t}$, $\Phat_{t,\new}\leftarrow [.]$,
\item $\mathrm{phase} \leftarrow \mathrm{detect}$, reset $k \leftarrow 0$
\een
{\bf end-if}
\een
$\mathrm{eigenvectors}(\bm{\mathcal{M}},\thresh)$ returns a basis matrix for the span of eigenvectors with eigenvalue above $\thresh$.
$\mathrm{eigenvectors}(\bm{\mathcal{M}},,r)$ returns a basis matrix for the span of the top $r$ eigenvectors.

{\bf Offline RPCA: } at $t = \that_j + K \alpha$, for all $t \in [\that_{j-1}+ K \alpha+1,  \that_j + K \alpha]$, compute
\[
\xhatt^{\mathrm{offline}} \leftarrow \I_{\hat{\mathcal{T}}_t} ((\bm{\Phi}_{\that_j + K \alpha})_{\hat{\mathcal{T}}_t})^{\dag} \bm{\Phi}_{\that_j + K \alpha} \bm{m}_t
\text{  and  }
\lhatt^{\mathrm{offline}} \leftarrow \mt - \xt
\]
\end{algorithm}

\begin{algorithm}[t!]
\caption{\small{cluster PCA}}
\label{clusterPCA}
\ben

\item If $k==0$, estimate the clusters
\ben
\item Set $ u = \frac{t}{\alpha}$ and compute $\hat{\bm{\Sigma}}_{\mathrm{sample}} =  \frac{1}{\alpha}\sum_{t =(u-1)\alpha+1}^{u \alpha} \lhatt\lhatt{}'$. Let $\hat\lambda_i$ denote its $i$-th largest eigenvalue.

\item To get the first cluster $\hat{\mathcal{G}}_{j,1}$, we start with the index of the first (largest) eigenvalue and keep adding indices of the smaller eigenvalues to it until $\frac{\hat\lambda_1}{\hat\lambda_{i+1}} > \ghatp$ but $\frac{\hat\lambda_1}{\hat\lambda_{i}} \le \ghatp$ or until $\hat\lambda_{i+1} < 0.25 \lamtrain$. We set  $\hat{\mathcal{G}}_{j,1} = \{1,2, \dots i\}$.
\\
For  $\hat{\mathcal{G}}_{j,2}$, start with the $(i+1)$-th eigenvalue and repeat the above procedure.
Repeat the above for each new cluster and stop when there are no more eigenvalues larger than $0.25 \lamtrain$. 

\item  $k \leftarrow k+1$, $\Phat_{t,*} \leftarrow \Phat_{t-1,*}$, $\Phat_{t,\new} \leftarrow \Phat_{t-1,\new}$, $\Phat_{t} \leftarrow [\Phat_{t,*} \ \Phat_{t,\new}]$
\een

\item If $1 \le k \le \vartheta$, estimate the $k$-th cluster's subspace by cluster PCA
\ben
\item Set $ u = \frac{t}{\alpha}$,  set $\hat{\bm{G}}_{j,0} \leftarrow [.]$. 
  \bi
  \item let $\hat{\bm{G}}_{j,\det,k}:= [\hat{\bm{G}}_{j,0},\hat{\bm{G}}_{j,1}, \dots \hat{\bm{G}}_{j,k-1}]$ and let $\bm\Psi_k:=(\I -\hat{\bm{G}}_{j,\det,k} \hat{\bm{G}}_{j,\det,k}{}')$ (notice that $\bm\Psi_{j,1}=\I$);  compute
  $\bm{\mathcal{M}}_{\mathrm{cpca}} = \bm\Psi_k \left(\frac{1}{\alpha}\sum_{t \in (u-1)\alpha+1}^{u \alpha} \lhat_t \lhat_t{}'\right) \bm\Psi_k$

  \item compute $\hat{\bm{G}}_{j,k} \leftarrow \mathrm{eigenvectors}(\bm{\mathcal{M}}_{\mathrm{cpca}},,|\hat{\mathcal{G}}_{j,k}|)$
  \ei

\item $k \leftarrow k+1$, $\Phat_{t,*} \leftarrow \Phat_{t-1,*}$, $\Phat_{t,\new} \leftarrow \Phat_{t-1,\new}$, $\Phat_{t} \leftarrow [\Phat_{t,*} \ \Phat_{t,\new}]$
\een

\item If $k==\vartheta$, set $\Phat_{t} \leftarrow [\hat{\bm{G}}_{j,1} \cdots \hat{\bm{G}}_{j,\vartheta}]$.

\een
\end{algorithm}

\section{Automatic ReProCS-cPCA} \label{algo_sec}

The automatic ReProCS-cPCA algorithm is summarized in Algorithm \ref{reprocsdet}. It proceeds as follows. It begins by estimating the initial subspace as the top $r_0$ left singular vectors of $[\bm{m}_1,\bm{m}_2,\dots, \bm{m}_{t_\train}]$. Let $\Phat_t$ denote the basis matrix for the subspace estimate at time $t$.
At time $t$, if the previous subspace estimate, $\Phat_{t-1}$, is accurate enough, because of the ``slow subspace change" assumption, projecting $\mt = \xt + \lt + \wt$ onto its orthogonal complement nullifies most of $\lt$. Specifically, we compute $\bm{y}_t:= \bm{\Phi}_t \mt$ where $\bm{\Phi}_t:= \I - \Phat_{t-1}\Phat_{t-1}{}'$. Clearly, $\bm{y}_t = \bm{\Phi}_t \xt + \bm{b}_t$ where $\bm{b}_t:= \bm{\Phi}_t\lt + \bm{\Phi}_t \wt$ and it can be argued that $\| \bm{b}_t \|_2$ is small: $\| \bm{\Phi}_t\lt\|_2$  is small due to the slow subspace change assumption and $\|\wt\|_2 \le \epsilon_w$.
Thus recovering $\xt$ from $\bm{y}_t$ becomes a traditional sparse recovery problem in small noise \cite{candes_rip}. We recover $\xt$ by $l_1$ minimization with the constraint $\| \bm{y}_t - \bm{\Phi}_t x\|_2 \le \xi$ and estimate its support by thresholding using a threshold $\omega$. We use the estimated support, $\That_t$, to get an improved debiased estimate of $\xt$, denoted $\xhatt$, by least squares (LS) estimation on $\That_t$. We then estimate $\lt$ as $\lhat_t = \mt - \xhatt$.
By the denseness assumption given in Model \ref{dense_model}, it can be argued that the restricted isometry constant (RIC) of $\bm{\Phi}_t$ will be small. Under the theorem's assumptions, we can bound it by 0.14. This ensures that a sparse $\xt$ is indeed accurately recoverable from $\bm{y}_t$.
%
With the support estimation threshold $\omega$ set as in Theorem \ref{thm1_cor}, it can be argued that the support will be exactly recovered, i.e., $\That_t=\T_t$. Let $\et: = \lt - \lhat_t$.
With this, it is clear that $\et =  (\xhatt - \xt) - \wt$ satisfies
\bea
\et = \I_{\T_t} [(\bm{\Phi}_t)_{\T_t}{}'(\bm{\Phi}_t)_{\T_t}]^{-1} \I_{\T_t}{}' \bm{b}_t - \wt = \I_{\T_t} [(\bm{\Phi}_t)_{\T_t}{}'(\bm{\Phi}_t)_{\T_t}]^{-1} \I_{\T_t}{}' \bm{\Phi}_t (\lt+\wt) - \wt.
\label{etdef_0}
\eea
Using the bound on the RIC of $\bm\Phi_t$, clearly $\|(\bm{\Phi}_t)_{\T_t}{}'(\bm{\Phi}_t)_{\T_t}^{-1}\|_2 \le (1-0.14)^{-1} < 1.2$. Thus,  $\|\et\|_2 \le 1.2 \|\bm{b}_t\|_2 + \epsilon_w$, i.e., it is small too. In other words, $\lt$ is accurately recovered.  

The estimates $\lhatt$ are used in the subspace estimation step which involves (i) detecting subspace change; (ii) $K$ steps of projection-PCA, each done with a new set of $\alpha$ frames of $\lhatt$, to get an accurate enough estimate of the newly added subspace; and (iii) cluster PCA to delete the old subspace by re-estimating the current subspace. At the end of the projection PCA step, the estimated subspace dimension is at most $r+r_\new$, and after cluster PCA, it comes down to at most $r$.

{\em Subspace update. }
In the subspace update step, the algorithm switches between the ``detect" phase, the ``pPCA" phase and the ``cPCA" phase. It starts in the ``detect" phase. When a subspace change is detected, i.e. at $t= \that_j$, it enters the ``pPCA" phase. After $K$ iterations of projection-PCA, i.e. at $t=\that_j+K\alpha$, the new subspace has been accurately estimated. At this time, it enters the ``cPCA" phase. At $t=\that_j+K\alpha+ (\vartheta+1)\alpha$, cluster PCA is done. At this time, it enters the ``detect" phase again and remains in it until the next subspace change is detected.
We detect the $j$-th subspace change as follows. Let $\Phat_* := \Phat_{\that_{j-1}+K\alpha+(\vartheta+1)\alpha}$. We detect change by comparing the eigenvalues of $\frac{1}{\alpha}\sum_{t} (\I-\hat{\bm{P}}_*\hat{\bm{P}}_*{}')\lhat_t \lhat_t' (\I-\hat{\bm{P}}_*\hat{\bm{P}}_*{}')$ to a chosen threshold at every $t=u\alpha$ when the algorithm is in the ``detect" phase.

{\em Projection-PCA (p-PCA). }
We use projection-PCA to estimate the newly added subspace. The reason this cannot be done using standard PCA is as follows \cite{rrpcp_perf}. Let $\sum_t$ denote a sum over an $\alpha$ length time interval. 
Because of how $\lt$ is recovered, the error, $\et$, in the estimate of $\lt$, $\lhat_t$, is correlated with $\lt$. This is evident from \eqref{etdef_0}. 
Due to this, the dominant terms in the perturbation seen by standard PCA, $\frac{1}{\alpha} \sum_t \lhatt \lhatt{}' - \frac{1}{\alpha} \sum_t \lt {\lt}'$,
are $\frac{1}{\alpha} \sum_t \lt {\et}'$ and its transpose\footnote{When $\lt$ and $\et$ are uncorrelated and one of them is zero mean, it can be argued by law of large numbers that, whp, these two terms will be close to zero and $\frac{1}{\alpha} \sum_t \et {\et}'$ will be the dominant term. 
}.
Thus, when the condition number of $\cov(\lt)$ is large, it is not possible to argue that the perturbation will be small compared to the smallest eigenvalue of $\cov(\lt)$. With a large perturbation, either the $\sin \theta$ theorem \cite{davis_kahan} (that bounds the subspace error between the eigenvectors of the true and estimated sample covariance matrices) cannot be applied or it gives a very large and useless bound.

Projection-PCA addresses the above issue as follows. Consider the $j$-th subspace change. 
Let $\bm{P}_* := \P_{t_{j-1}}$, $\P_{\new}:=\P_{t_j,\new}$, and $\Phat_* := \Phat_{\that_{j-1}+K\alpha+(\vartheta+1)\alpha}$. 
Denote the time at which this change is detected by $\that_j$. As explained in \cite{rrpcp_isit15}, it is easy to show that, whp, $t_j \le \that_j \le t_j+2\alpha$.
After $\that_j$ we use SVD on $K$ different sets of $\alpha$ frames of the $\lhat_t$'s projected orthogonal to $\Phat_*$ to get $K$ estimates of the new subspace $\Span(\bm{P}_{\new})$.
We get the $k$-th estimate, $\Phat_{\new,k}$, as the left singular vectors of $(\I - \Phat_* \Phat_*{}')[\lhat_{\that_j+(k-1)\alpha+1}, \dots , \lhat_{\that_j+ k\alpha}]$ with singular values above a threshold. After each projection-PCA step, we update $\Phat_t$ as $\Phat_t = [\Phat_* \ \textit{} \Phat_{\new,k}]$. This ensures that the error $\et$ is smaller for the next projection-PCA interval compared to the previous one and hence the subspace estimates also improve with each iteration. The above is done $K$ times with $K$ chosen so that, by $t=\that_j+K\alpha$, the error in estimating the new subspace is below $r_\new \zeta$. This ensures that, at this time, $\SE_t \le r \zeta +r_\new\zeta$.

{\em Cluster PCA for deleting directions by re-estimating the subspace. } The next step is to delete the subspace $\Span(\P_{j,\old})$ from $\Phat_t$.  The goal of doing this is to reduce the subspace error from $(r+r_\new) \zeta$ to $r \zeta$. The simplest way to do this would be to re-estimate $\Span(P_t)$ by standard PCA, i.e. compute the eigenvectors of $\frac{1}{\alpha}\sum_{t = \that_j+K \alpha + 1}^{t = \that_j+K \alpha + \alpha} \lhat_t \lhat_t'$ with eigenvalues above a threshold. However, since $\lt$ and $\et$ are correlated, this will cause a problem similar to the one described above. It will work only if the condition number of $\cov(\lt)$ is small. This is impractical though since we assume that $\lt$ can be large but structured noise.
Hence we re-estimate the subspace by developing a generalization of the projection-PCA idea that we call {\em cluster PCA (cPCA)}. This relies on the clustering assumption given in Model \ref{clust_model}.

cPCA proceeds as follows. We first estimate the clusters as follows. We compute the empirical covariance matrix of $\lhatt$'s after the new subspace is accurately estimated: $\hat{\bm\Sigma}_{\mathrm{sample}} = \frac{1}{\alpha} \sum_{t=\that_j+K\alpha+1}^{t=\that_j+K\alpha+\alpha} \lhatt \lhatt'$ and obtain its EVD. Let $\hat\lambda_i$ denote its $i$-th largest eigenvalue.
To get the first cluster $\hat{\mathcal{G}}_{j,1}$, we start with the index of the first (largest) eigenvalue and keep adding indices of the smaller eigenvalues to it until $\frac{\hat\lambda_1}{\hat\lambda_{i+1}} > \ghatp$ but
$\frac{\hat\lambda_1}{\hat\lambda_{i}} \le \ghatp$ or until the next eigenvalue $\hat\lambda_{i+1} < 0.25 \lamtrain$. We set  $\hat{\mathcal{G}}_{j,1} = \{1,2, \dots i\}$.
To get the second cluster we repeat the same procedure but starting with the $(i+1)$-th eigenvalue. We repeat this until there is no eigenvalue larger than $0.25 \lamtrain$. Observe that $\ghatp$ is set to a value that is a little larger than $g^+$ (see Theorem \ref{thm1_cor}). This is needed to allow for the fact that $\hat\lambda_i$ is not equal to the $i$-th eigenvalue of $\Lamj$ but is within a small margin of it. For the same reason, we need to also use a ``zeroing" threshold of $0.25 \lamtrain$ (notice that $\hat{\bm\Sigma}_{\mathrm{sample}}$ is not exactly low rank). This, along with appropriately setting $\ghatp$, and with using the separation condition from Model \ref{clust_model} ensures that, whp, all the clusters are correctly recovered.

Let $\bm{G}_{j,k}:= (\P_j)_{\hat{\mathcal{G}}_{j,k}}$.
Next, we estimate the subspace corresponding to the first cluster, $\Span(\bm{G}_{j,1})$ by standard PCA on $[\lhat_{\that_j+(K+1)\alpha+1}, \dots , \lhat_{\that_j+ (K+1)\alpha+\alpha}]$, i.e., by computing its top $|\hat{\mathcal{G}}_{j,1}|$ left singular vectors. Since the cluster's condition number is small (bounded by $g^+$), this works. Denote the basis for the estimated subspace by $\hat{\bm{G}}_{j,1}$. To estimate the subspace corresponding to the second cluster, we project the next set of $\alpha$ $\lhatt$'s orthogonal to $\hat{\bm{G}}_{j,1}$, followed by standard PCA to compute the top $|\hat{\mathcal{G}}_{j,2}|$ left singular vectors \cite{rrpcp_perf}. To estimate the $k$-th cluster's subspace, we do a similar thing but with projecting orthogonal to the estimated subspace corresponding to the previous $k-1$ clusters \cite{rrpcp_perf}.

%


\section{Proof Outline for Theorem \ref{thm1_cor} and Corollary \ref{thm1_cor_corol}} \label{outline}
The proof proceeds by induction.
Consider the $j$-the subspace change interval. 
Let $\bm{P}_* := \P_{t_{j-1}} = \P_{t_j-1}$, $\P_{\new}:=\P_{t_j,\new}$, and $\Phat_* := \Phat_{\that_{j-1}+K\alpha+(\vartheta+1)\alpha}$.
Assume that there have been no (false) change detects in the interval $[\that_{j-1}+K\alpha+(\vartheta+1)\alpha+1, t_j-1]$. Thus, $\Phat_{t_j-1} = \Phat_*$. Assume also that the subspace, $\Span(\P_{t_j-1})=\Span(\P_{*})$, has been accurately recovered, i.e., $\SE_{t_j-1} = \mathrm{dif}(\Phat_*,\P_*) \le  r \zeta$.
Conditioned on this, we use the following steps to show that, whp, the same conclusions hold at $t=t_{j+1}-1$ as well.
\ben
\item First, we  show that the subspace change is detected within a short delay of $t_j$. We show that $t_j \le \that_j \le t_j + 2 \alpha$ whp. This is done in Lemma \ref{det}.

\item At $t=\that_j + \alpha$, the first projection-PCA step is done to get the first estimate, $\Phat_{\new,1}$, of $\Span(\P_{\new})$. This computes the top singular vectors of $[\lhat_{\that_j+1}, \lhat_{\that_j+2}, \dots, \lhat_{\that_j+\alpha}]$ projected orthogonal to $\Span(\Phat_*)$.
In the interval $[t_j, \that_j + \alpha-1]$, the new subspace is not estimated at all, i.e., $\Phat_t = \Phat_*$ while $\P_t = [\P_* \ \P_\new]$ and so $\SE_t \le 1$. Thus, the noise seen by the projected sparse recovery step, $\bm{b}_t$, is the largest in this interval. Hence the error $\et$ is also the largest for the $\lhatt$'s used in the first projection-PCA step. However, due to slow subspace change, even this error is not too large. Because of this, and because $\P_\new$ is dense, 
 we can argue that $\Phat_{\new,1}$ is a good estimate. We show that $\mathrm{dif}([\Phat_* \ \Phat_{\new,1}], \P_\new) \le 0.19 < 1$.  Thus, at this time, $\SE_t=\mathrm{dif}([\Phat_* \ \Phat_{\new,1}], [\P_* \ \P_\new]) \le r \zeta + 0.19$. This is shown in Lemmas \ref{pPCA} and \ref{zetadecay}.

\item At $t=\that_j + k\alpha$,  for $k=1,2,\dots, K$, the $k$-th projection-PCA step is done to get the $k$-th estimate, $\Phat_{\new,k}$.  This computes the top singular vectors of $[\lhat_{\that_j+(k-1)\alpha+1}, \lhat_{\that_j+(k-1)\alpha+2}, \dots, \lhat_{\that_j+k\alpha}]$ projected orthogonal to $\Span(\Phat_*)$.
    After the first projection-PCA step, $\Phat_t = [\Phat_* \ \Phat_{\new,1}]$ and this reduces $\bm{b}_t$ and hence $\et$ for the $\lhatt$'s in the next $\alpha$ frames. This fact, along with the fact that $\et$ is approximately sparse with support $\T_t$ and $\T_t$ follows Model \ref{sbyrho}, in turn, imply that the perturbation seen by the second projection-PCA step is even smaller. So $\Phat_{\new,2}$ is a more accurate estimate of $\Span(\bm{P}_\new)$ than $\Phat_{\new,1}$. Repeating the same argument, the third estimate is even better and so on. Under the theorem's assumptions, we can show that $\mathrm{dif}([\Phat_* \ \Phat_{\new,k}], \P_\new) \le 0.19\cdot0.1^{k-1} + 0.15 r_\new \zeta$ and so, at $t=\that_j+k \alpha$, $\SE_t \le r \zeta + 0.19\cdot0.1^{k-1} + 0.15 r_\new \zeta$. This is shown in Lemmas \ref{pPCA} and \ref{zetadecay}. The most important idea here is to use the fact that $\et$ is approximately supported on $\T_t$ (shown in Lemma \ref{cslem_cor}) and the support change model on $\T_t$ (this is used in Lemma \ref{blockdiag1}).



\item The above is repeated $K$ times with $K$ set to ensure that, by $t=\that_j + K\alpha$, $\mathrm{dif}([\Phat_* \ \Phat_{\new,K}], \P_\new) \le r_\new \zeta$ and so, at this time, $\SE_t \le (r + r_\new) \zeta$.

\item In the interval $[\that_j + K\alpha+1, \that_j + K\alpha+ (\vartheta+1) \alpha]$, cluster PCA is done to delete $\Span(P_{t_{j}, \old})$. At the end of this step, we can show that the bound on $\SE_t$ has reduces from $(r + r_\new) \zeta$ to $r \zeta$.  This is proved in Lemmas \ref{cluster_lem}, \ref{cPCA} and \ref{zetadecay}.

\item Finally, we also argue that there are no (false) subspace change detects for any $t \in [\that_j + K \alpha + (\vartheta+1) \alpha+1, t_{j+1}-1]$. This ensures that $\that_{j+1} \ge t_{j+1}$. This is done in Lemma \ref{falsedet_del}.
\een
To prove the theorem, we first show that the initial subspace is recovered accurately enough, i.e., $\SE_t \le r \zeta$ at $t=t_\train+1$, whp. This is done in Lemma \ref{initsub_cor}. Then, repeating the above argument for each subspace change period, we can obtain the subspace error bounds of the theorem. We set $t_\train$ and $\alpha$ to ensure that the probability of the good events is at least $1-3n^{-10}$. The sparse recovery error bounds can be obtained by using these bounds and quantifying the discussion of Sec. \ref{algo_sec}. This is done in Lemma \ref{cslem_cor}.

The main part of the proof is the analysis of the projection-PCA steps (for subspace addition) and the cluster PCA steps (for subspace deletion). We explain its key ideas next. Assume for this approximate analysis that $\wt=0$ and that $\mathrm{dif}(\Phat_*, \P_*)=0$ (previous subspace is perfectly estimated).
In the $k$-th projection-PCA step the goal is to bound $\zeta_{\new,k}:=\mathrm{dif}([\Phat_*, \Phat_{\new,k}], \P_\new)$ conditioned on ``accurate recovery so far". Here ``accurate recovery so far" means $\mathrm{dif}(\Phat_*, \P_*) \approx 0$ and $\zeta_{\new,k-1} \le \zeta_{\new,k-1}^+$. Before $k=1$, there is no estimate of $\P_\new$ and thus we have $\zeta_{\new,0} \le \zeta_{\new,0}^+=1$.

We first use the $\sin \theta$ theorem \cite{davis_kahan} (Theorem \ref{sintheta}) to get a bound on $\zeta_{\new,k}$. This is done in  Lemma \ref{zetakbnd}.
We then bound the terms in this bound using the matrix Azuma inequality from \cite{tail_bound} (Corollaries \ref{azuma_hermitian} and \ref{azuma_norm}). This is done in Lemmas \ref{Ak_cor}, \ref{Akperp_cor} and \ref{calHk_cor}.
%
%
Using the $\sin \theta$ theorem followed by using matrix Azuma for lower bounding $\lambda_{\min}(\frac{1}{\alpha} \sum_t (I - \Phat_* \Phat_*{}') \lt \lt' (I - \Phat_* \Phat_*{}'))$, we can conclude that
\bea
\zeta_{\new,k} \lesssim  \frac{ \|\mathrm{perturbation}\|_2 }{\frac{1}{1-b^2} \lambda_\new^- - \epsilon - \|\mathrm{perturbation}\|_2 }
\lesssim  \frac{ 2 \big\|\frac{1}{\alpha} \sum_t (I - \Phat_* \Phat_*{}') \lt \et' \big\|_2 +  \big\|\frac{1}{\alpha} \sum_t \et \et'\big\|_2}{ \frac{1}{1-b^2}\lambda_\new^-  - \epsilon - (2\|\frac{1}{\alpha} \sum_t (I - \Phat_* \Phat_*{}') \lt \et' \big\|_2 + 2\big\|\frac{1}{\alpha} \sum_t \et \et'\big\|_2)}
\label{zetanew_approx}
\eea
Here $\mathrm{perturbation} = \frac{1}{\alpha} \sum_t (I - \Phat_* \Phat_*{}') \lhatt \lhatt' (I - \Phat_* \Phat_*{}') - \frac{1}{\alpha} \sum_t (I - \Phat_* \Phat_*{}') \lt \lt' (I - \Phat_* \Phat_*{}')$. Since $\sum_t (\lhatt \lhatt' - \lt \lt') = \sum_t (\lt \et' + \et \lt' + \et \et')$, the bound used in the second inequality above follows.
The next task is to bound the two perturbation terms using the matrix Azuma inequality. This is done in Lemma \ref{calHk_cor}.
As explained in Sec \ref{algo_sec}, under ``accurate recovery so far", it can be shown that $\et$ satisfies \eqref{etdef_0} and that $ \big\|\invterm \big\|_2 \le 1.2$. This is proved in Lemma \ref{cslem_cor}. Notice that, when $\wt=0$, $\et$ is exactly supported on $\T_t$. Using the expression for $\et$, expanding $\lt$ in terms of $\bnu_\tau$'s, manipulating the resulting terms carefully (as explained in Sec \ref{general_decomp}), and applying the matrix Azuma inequality, one can show that, whp,
\[
\big\|\frac{1}{\alpha} \sum_{t=t_0}^{t_0+\alpha-1} (\I - \Phat_* \Phat_*{}') \lt \et' \big\|_2 \le 4 \epsilon + \frac{1}{\alpha}\frac{b^2}{1-b^2} (r\gamma^2) + \mathrm{LargeTerm}_k
\]
where $t_0 = \that_j+(k-1)\alpha+1$ is the first time instant of the $k$-the projection-PCA interval and
\[
\mathrm{LargeTerm}_k:=\big\|\frac{1}{\alpha} \sum_{t=t_0}^{t_0+\alpha-1} \sum_{\tau=t_0}^{t}  b^{2t-2\tau}  \P_\new \Lamtnew \P_\new' (\I - \Phat_* \Phat_*{}' - \Phat_{\new,k-1} \Phat_{\new,k-1}{}') \ITt \invterm \ITt' \big\|_2.
\]
In the above, $\epsilon$ is very small (comes from applying Azuma for zero-mean terms). The second term is also very small since $1/\alpha \le (r_\new \zeta)^2$. Thus, $\mathrm{LargeTerm}_k$ is the only significant term. To bound it, for $k=1$, we use the fact that $\Phat_{\new,k-1} = \Phat_{\new,0} = [.]$ and hence $(I - \Phat_* \Phat_*{}' - \Phat_{\new,k-1} \Phat_{\new,k-1}{}') \P_\new \approx \P_\new$ and $\P_\new$ is dense. From Model \ref{dense_model}, $ \big\| \P_\new{}' \ITt \big\|_2 \le 0.02$. Thus, using $ \big\|\invterm \big\|_2 \le 1.2$ and slow subspace change, \eqref{anew_small}, we get that, for $k=1$,
\[
\big\|\frac{1}{\alpha} \sum_{t=t_0}^{t_0+\alpha-1} (I - \Phat_* \Phat_*{}') \lt \et' \big\|_2  \lesssim  \big\|\mathrm{LargeTerm}_1 \big\|_2 \le \frac{1}{1-b^2} 1.2 \cdot 0.02 \cdot \lambda_\new^+ \le \frac{1}{1-b^2} 1.2 \cdot 0.02 \cdot 3 \lambda^-  = 0.072 \frac{1}{1-b^2}\lambda^-.
\]

For $k>1$, we cannot show that $ (\I - \Phat_* \Phat_*{}' - \Phat_{\new,k-1} \Phat_{\new,k-1}{}') \P_\new$ is dense\footnote{The partial result of \cite{rrpcp_perf} assumed that this holds and then used the above approach to get a performance guarantee.}. Thus we use a different approach. We apply the Cauchy-Schwartz inequality (Lemma \ref{CSmat}) with $\bm{X}_t:=\sum_{\tau=t_0}^{t}  b^{2t-2\tau}  \P_\new \Lamtnew \P_\new' (I - \Phat_* \Phat_*{}' - \Phat_{\new,k-1} \Phat_{\new,k-1}{}')$ and $\bm{Y}_t:=\ITt \invterm \ITt'$, followed by using Model \ref{sbyrho} on $\T_t$ to bound  $\lambda_{\max}( \frac{1}{\alpha} \sum_{t=t_0}^{t_0+\alpha-1} \bm{Y}_t \bm{Y}_t')$.

It is easy to see that $\ds \lambda_{\max}( \frac{1}{\alpha} \sum_{t=t_0}^{t_0+\alpha-1} \bm{X}_t \bm{X}_t') \le \max_t \|\bm{X}_t\|_2^2$ and $\ds \|\bm{X}_t\|_2 \le \frac{1}{1-b^2} \lambda_\new^+ \zeta_{\new,k-1}^+ \le  3 \zeta_{\new,k-1}^+  \frac{1}{1-b^2}\lambda^-$. 

We bound $\ds \lambda_{\max}( \frac{1}{\alpha} \sum_{t=t_0}^{t_0+\alpha-1} \bm{Y}_t \bm{Y}_t')$ by using Model \ref{sbyrho} on support change. This is done in Lemma \ref{blockdiag1}. This lemma exploits the fact that $\ds \frac{1}{\alpha} \sum_t \bm{Y}_t \bm{Y}_t' = \frac{1}{\alpha} \sum_t \ITt (\invterm)^2 \ITt'$ is a block-banded matrix and, for each block, the summation is not over $\alpha$ frames but only over $\beta$ frames with $\beta$ being much smaller. For example, if Model \ref{sbyrho} holds with $\rho=1$, this matrix is block diagonal; if it holds with $\rho=2$, then it is block-tridiagonal and so on. Thus, using $\|\invterm\|_2 \le 1.2$, we can show that $\lambda_{\max}( \frac{1}{\alpha} \sum_{t=t_0}^{t_0+\alpha-1} \bm{Y}_t \bm{Y}_t') \le \frac{1}{\alpha} \rho^2 \beta (1.2)^2 \le 0.0001 \cdot (1.2)^2$.

By Cauchy-Schwartz and the above bounds, we can conclude that, for $k>1$,
\[
\big\|\frac{1}{\alpha} \sum_{t=t_0}^{t_0+\alpha-1} (I - \Phat_* \Phat_*{}') \lt \et' \big\|_2  \lesssim \big\| \mathrm{LargeTerm}_k \big\|_2 \le
 \sqrt{0.0001 \cdot (1.2)^2} \cdot  3 \cdot \zeta_{\new,k-1}^+   \frac{1}{1-b^2}\lambda^- = 0.036\cdot \zeta_{\new,k-1}^+  \frac{1}{1-b^2}  \lambda^-
\]

Using an approach similar to the one outlined above one can also bound the $\et \et'$ term. This is actually easier to bound because one does not need Cauchy-Schwartz.
For $k=1$, we get
\[
\big\|\frac{1}{\alpha} \sum_{t=t_0}^{t_0+\alpha-1} \et \et' \big\|_2 \lesssim \rho^2 \beta (1.2)^2 \cdot 0.02^2 \frac{1}{1-b^2} 3 \lambda^- \le 0.0001 \cdot 1.44 \cdot 0.02^2 \cdot 3 \frac{1}{1-b^2}\lambda^- <0.00002 \frac{1}{1-b^2}\lambda^-.
\]
and for $k>1$,
\[
\big\|\frac{1}{\alpha} \sum_{t=t_0}^{t_0+\alpha-1} \et \et' \big\|_2 \lesssim \rho^2 \beta (1.2)^2 \cdot (\zeta_{\new,k-1}^+ )^2  \frac{1}{1-b^2} 3 \lambda^- \le 0.0001 \cdot 1.44 \cdot 3 \cdot (\zeta_{\new,k-1}^+)^2  \frac{1}{1-b^2}\lambda^- < 0.075 (\zeta_{\new,k-1}^+)^2  \frac{1}{1-b^2}\lambda^-
\]

Using the above bounds in \eqref{zetanew_approx} and using $\lambda_\new^- \ge \lambda^-$, we can conclude that,
\[
\zeta_{\new,1}^+ \lesssim 0.19, \ \ \zeta_{\new,k}^+ \lesssim \frac{2 \cdot 0.036 \cdot \zeta_{\new,k-1}^+  +  0.075 (\zeta_{\new,k-1}^+)^2  }{1 - \mathrm{NumeratorTerm}}
\]
Here $\mathrm{NumeratorTerm}$ refers to the expression from the numerator. From the above, it is easy to see that $\zeta_{\new,2}^+ \lesssim 0.19$ and, proceeding similarly, $\zeta_{\new,k}^+ \lesssim 0.19$. Using this to get a loose bound on $\mathrm{NumeratorTerm}$, we can conclude that $\zeta_{\new,k}^+ \lesssim 0.1 \zeta_{\new,k-1}^+ \le 0.19 \cdot 0.1^{k-1}$.

The above approximate analysis ignores the fact that $\Span(\Phat_*) \neq \Span(\P_*)$. It also ignores the unstructured noise term $\wt$ and the other small terms that come with each application of matrix Azuma. With incorporating all this, and with using $\mathrm{dif}(\Phat_*,\P_*) \le r \zeta$ (instead of zero), we can conclude that $\zeta_{\new,k} \le \zeta_{\new,k}^+ \le 0.19 \cdot 0.1^{k-1} + 0.15 r_\new \zeta$. By picking $K$ carefully, we get that $\zeta_{\new,K} \le r_\new \zeta$ and thus $\SE_t \le (r+r_\new)\zeta$ after the $K$-the projection PCA step.

The analysis of cluster PCA is a significant generalization of the above ideas. The slow subspace change assumption is replaced by the clustering assumption at various places in its proof.

%
%
%

\subsection{Novelty in proof techniques} \label{novelty}
This work has two key contributions - it analyzes ReProCS with the deletion step (done via cluster PCA), and it obtains a complete result for ReProCS and ReProCS-cPCA for the case when the $\lt$'s are correlated over time.

While the overall proof structure described above is similar to that used in \cite{rrpcp_isit15}, the proof approach for proving the ``main lemmas" is quite different for the correlated $\lt$'s case. The first such difference is seen in Fact \ref{ltbnds} which shows how to bound $\|(I - \Phat_{t-1} \Phat_{t-1}{}') \lt\|_2$ for when $\lt$ is correlated over time. This is used to prove Lemma \ref{cslem_cor}.
The second and most significant difference is in proving the matrix-Azuma-based lemmas for projection-PCA and for cluster PCA. These are proved in Sec \ref{3_pfs} and \ref{3_pfs_del}. The matrix Azuma inequality \cite[Theorem 7.1]{tail_bound} is significantly harder to apply than the matrix Hoeffding \cite{tail_bound}. There are two reasons for this. First we need to get the sums of conditional expectations of quantities needed to apply this result in a form that can be bounded easily. The simplest way of doing this can lead to loose bounds. To get the desired bounds, we need to rewrite $\lt$ in terms of past $\bnu_t$'s and use the fact that $b^\alpha < (r_\new \zeta)$ (is very small) and that $\sum_{\tau=t-\alpha+1}^t b^t \le 1/(1-b) \le 1/(1-b_0) < 1.12$. In words, the contribution of very old $\bnu_t$'s is negligible and the contribution due to the last $\alpha$ $\bnu_t$'s is only slightly larger than that of one $\bnu_t$.%

The third main difference is the analysis of the automatic cluster estimation step and of the cluster PCA algorithm for deleting the subspace. The fact that the former is correct whp is shown in Lemma \ref{cluster_lem}. This uses Lemma \ref{cluster_lem_azuma} and the separation condition from Model \ref{clust_model} to show that, whp, the clusters obtained by using a threshold of $\ghatp$ on the condition numbers of the eigenvalues of the empirical covariance matrix computed with the $\lhatt$'s are exactly the same as the true clusters defined in Model \ref{clust_model}. The analysis of cluster PCA (Lemma \ref{cPCA}) relies on matrix-Azuma-based Lemmas \ref{Atil_cor}, \ref{Atilperp_cor}, and \ref{Htil_cor}. These are new too and are proved using a significant generalization of the approach used for analyzing the projection-PCA step.

%
%
%

\section{Proof of Theorem \ref{thm1_cor} and Corollary \ref{thm1_cor_corol}} \label{pf_thm1_cor}
We first give the most general denseness assumption and the most general model on $\T_t$ in Sec. \ref{gens} below. Next, we define quantities that will be used in the proofs in Sec. \ref{defs}.
The basic lemmas that are used several times in the proof are stated next in Sec. \ref{basiclems}. The five main lemmas leading to the proof and the proof itself are given in Sec. \ref{pf_thm}. We then give the seven key lemmas that are used to prove the main lemmas in Sec. \ref{keylems}, followed by the proofs of the main lemmas in Sec. \ref{pf_mainlems}. The proofs of the key lemmas are the long ones and these are given in Sec. \ref{3_pfs} and \ref{3_pfs_del}.

\subsection{Generalizations} \label{gens}
Consider the denseness assumption in Model \ref{dense_model}. This can be generalized as follows.
\begin{sigmodel} \label{dense_model_general}
For a basis matrix $\bm{P}$, define the (un)denseness coefficient
\[
\kappa_s(\bm{P}) := \max_{|\mathcal{T}| \leq s} \| {\I_{\mathcal{T}}}' \bm{P} \|_2
\]
Assume that
\beq \label{kappa_dense}
\kappa_{2s,*} :=\max_j \kappa_{2s}(\bm{P}_{t_j}) \leq 0.3 \quad \text{ and } \quad \kappa_{2s,\rmnew}:= \max_j \kappa_{2s}(\bm{P}_{t_j,\new}) \leq 0.02.
\eeq
\end{sigmodel}

\begin{lem} Model \ref{dense_model} is a special case of Model \ref{dense_model_general}. \end{lem}

\begin{proof}
Recall Model \ref{dense_model}. For any basis matrix $\P$, $\left[\kappa_1(\P)\right]^2 = \max_{i} \|\P' \I_i\|_2^2$.  Using the triangle inequality, it is easy to show that $\kappa_s(\bm{P}) \leq \sqrt{s}\kappa_1(\bm{P})$ \cite{rrpcp_perf}. Using this, the claim follows. 
\end{proof}

The proof of Theorem \ref{thm1_cor} only uses \eqref{kappa_dense} for the denseness assumption.

The reason for defining  the (un)denseness coefficient $\kappa_{s}(\P)$ as above is the following lemma from \cite{rrpcp_perf}.
\begin{lem}[\cite{rrpcp_perf}]\label{kappadelta}
For a basis matrix $\bm{P}$, $\delta_s(\I - \bm{P}\bm{P}') = \left(\kappa_s(\bm{P})\right)^2$.
\end{lem}

Next consider the support change model given in Model \ref{sbyrho}. This is one special case of the most general model that works for our result. This model was introduced in \cite{rrpcp_isit15}. We explain it here. What we need to prevent is $\T_t$ occupying the same indices for too many time instants in a given interval. If $\T_t$ does not change `enough' in a time interval of length $\alpha$, we will be unable to see enough entries of a given index of $\lt$ in order to be able to accurately fill in the missing ones.
The following model quantifies `enough' for our purposes.  The number of time instants for which an index is part of $\T_t$ is determined both by how often this set changes, and by how much it moves when it changes. The latter is parameterized by $\rho$ which controls how much the set moves when it changes.  For example $\rho=1$ would require that distinct sets be disjoint, and $\rho=2$ would mean that at least half of the set is displaced whenever it changes.  The parameter $h^+ \in (0,1)$ represents the maximum fraction of time for which the set remains in a given area in a time interval of length $\alpha$.  The smaller $h^+$, the more frequently the set will need to change in order to satisfy the model. Our result requires a bound on the product $\rho^2 h^+$.  

\begin{sigmodel}\label{general_model}
Let $\rho$ be a positive integer. Split $[1, t_{\max}]$ into intervals of length $\alpha$. Use $\J_u:=[(u-1)\alpha+1,u\alpha]$ to denote the $u$-th interval.
For a given interval, $\J_u$, let $\mathcal{T}_{(i),u}$ for $i = 1,\dots, l_u$ be mutually disjoint subsets of $\{ 1, \dots, n\}$ and let
 $\J_{(i),u}, i=1,2, \dots, l_u$ be {\em a} partition\footnote{i.e. the $\J_{(i),u}$'s are mutually disjoint intervals and their union equals $\J_u$} of the interval $\J_u$ so that
\begin{equation}\label{union}
\text{  for all $t \in \J_{(i),u}$,   }
\T_t \subseteq \mathcal{T}_{(i),u} \cup \mathcal{T}_{(i+1),u} \cup \dots \cup \mathcal{T}_{(i+\rho-1),u}
\end{equation}
Define
\begin{align}\label{alphabyp}  
h_u\left(\alpha;\{\mathcal{T}_{(i),u}\}_{\substack{\\i=1,\dots,l_u}}, \{\J_{(i),u} \}_{\substack{\\i=1,\dots,l_u}} \right)
&:= \max_{i=1,2,\dots l_u}\big| \J_{(i),u} \big|
\end{align}
and define $h_u^*(\alpha)$ as the minimum over all choices of $\mathcal{T}_{(i),u}$ and over all choices of the partition $\J_{(i),u}$.
%
\begin{align}\label{hstar}
h_u^*(\alpha) &:=
\min_{\substack{\text{\em all choices of mutually disjoint $\mathcal{T}_{(i),u}, i=1,2, \dots l_u$}  \\  \text{\em and all choices of mutually disjoint $\J_{(i),u},i=1,2,\dots l_u$} \\ \text{\em so that $\cup_{i=1}^{l_u} \J_{(i),u} = \J_u$ and \eqref{union} holds}  } }
h_u\left(\alpha;\{\mathcal{T}_{(i),u}\}_{\substack{\\i=1,\dots,l_u}}, \{\J_{(i),u} \}_{\substack{\\i=1,\dots,l_u}} \right)
\end{align}
Assume that $|\T_t|\le s$ and that for all $u = 1,\dots,\llceil\frac{t_{\max}}{\alpha}\rrceil$,
\[
h_u^*(\alpha)\leq h^+\alpha \ \text{ with }  h^+ \le \frac{0.0001}{\rho^2}.
\]
\end{sigmodel}
In the above model, $h_u^*(\alpha)$ provides a bound on how long $\T_t$ remains in a given ``area", $\mathcal{T}_{(i),u} \cup \mathcal{T}_{(i+1),u} \cup \dots \cup \mathcal{T}_{(i+\rho-1),u}$ during the interval $\J_u$, for the best allocation of $\T_t$'s to a given ``area" and the best choice of the ``areas."

Notice that \eqref{union} can always be trivially satisfied by choosing $l_u =1$, $\mathcal{T}_{(1),u} = \{  1, \dots, n \}$ and $\J_{(1),u}=\J_u$, but this will give $h_u(\alpha;.) = \alpha$ and hence is not a good choice. This is why we take a minimum over all choices.

\begin{lem}\label{spc_case}[{\cite{rrpcp_isit15}}]
Model \ref{sbyrho} is a special case of Model \ref{general_model} above with $h^+= \frac{\beta}{\alpha}$.
\end{lem}

\subsection{Definitions} \label{defs}
\begin{remark}
Recall that $\vartheta$ is the maximum number of clusters from Model \ref{clust_model}. For ease of notation, henceforth, we will assume that there are $\vartheta$ clusters for all $j$. If $\vartheta_j < \vartheta$, it will just mean that the last $(\vartheta - \vartheta_j+1)$ clusters are empty. 
\end{remark}

\begin{definition}
Define $\b_t:= \bm\Phi_t \mt - \bm\Phi_t \xt = \bm\Phi_t (\lt + \wt)$. This is the ``noise" seen by the projected sparse recovery step of the algorithm.

Define $\bm{e}_t$ to be the error made in estimating $\lt$.  That is
$
\et := \lt-\lhatt. 
$
Thus, from the algorithm, $\et = (\xhatt - \xt) - \wt$
\end{definition}

\begin{definition}
Define the intervals
\[
\mathcal{J}_u := [(u-1)\alpha+1,u\alpha].
\]
Define $u_j$ to be the $u$ such that $t_j \in \mathcal{J}_u$.  That is
$u_j := \llceil \frac{t_j}{\alpha}\rrceil.$
For the purposes of describing events before the first subspace change, let $u_0 := 0$.

Define $\hat{u}_j := \frac{\hat{t}_j}{\alpha}.$ Notice from the algorithm that this will be an integer, because detection only occurs when $t \mod{\alpha} = 0 $. We will show that, under appropriate conditioning, whp, $\hat{u}_j = u_j$ or $\hat{u}_j = u_j + 1$.

For the cluster-PCA step, define the following intervals for $k=0,1,2, \dots \vartheta$.
\[
\Ijkt := [\that_j + (K+1) \alpha + (k-1)\alpha + 1, \that_j + (K+1) \alpha + k \alpha]
\]
Notice that $\tilde{\mathcal{I}}_{j,0}$ is where the clusters are determined, and $\tilde{\mathcal{I}}_{j,k}$ is where cluster $k$ is recovered.
\end{definition}

\begin{definition} \label{def_P_starnew}
Define $ \bm{P}_{(j)}:= \bm{P}_{t_j}$,
\begin{align*}
\bm{P}_{(j),*} &:= \bm{P}_{(j-1)} = \bm{P}_{t_{j}-1}  \text{ and }  \bm{P}_{(j),\new}:= \bm{P}_{t_j,\new} \text{ for } j = 1,\dots,J \\
\bm{a}_{t,*} &:= {\Pjs}'\nut   \text{ and }  \bm{a}_{t,\rmnew} := {\Pjnew}'\nut  \text{ for }  t \in [t_j, t_{j+1}).
\end{align*}
Notice that $\bm{a}_{t,*}$ is a vector of length $r_{j-1}$, whose last $(r_{j-1}-r_{j,\old})$ entries are zeroes.
Also define
\[
\bm{P}_{(j),\add} := [\bm{P}_{(j),*} \ \bm{P}_{(j),\new}]
\]
\end{definition}

Thus, for $t\in[t_{j}, t_{j} + d]$, $\nut$ can be written as
\[
\nut = \Pj \at =  [ \Pjs \ \Pjnew ]\vect{\bm{a}_{t,*}}{\atnew}  
\]
 and $\cov(\nut)= \bm{\Sigma}_t$ can be rewritten as
\[
\bm{\Sigma}_t = \Pj \Lamt \Pj{}' =
\left[\bm{P}_{(j),*} \ \bm{P}_{(j),\new}\right]
\left[\begin{array}{cc}\bm{\Lambda}_{t,*} & \bm{0} \\ \bm{0} & \bm{\Lambda}_{t,\rmnew}\end{array}\right]
\left[\begin{array}{c} {\bm{P}_{(j),*}}' \\ {\bm{P}_{(j),\new}}' \end{array} \right]
\]
Notice that the last $(r_{j-1}-r_{j,\old})$ diagonal entries of $\bm{\Lambda}_{t,*}$ are zeroes.

\begin{remark}
From Model \ref{cor_model}, $\bm{P}_{(j),*}$ is orthogonal to $\bm{P}_{(j),\new}$. 
\end{remark}

\begin{definition}\label{def_Phat_starnew}
For $j=1,2,\dots,J$ and $k=1,2,\dots,K$ define
\begin{enumerate}
\item $\ds\Phat_{(j),*} := \Phat_{\that_{j-1} + K\alpha + (\vartheta+1)\alpha}$. If all subspace changes are correctly detected, this is the final estimate of $\bm{P}_{(j),*} = \bm{P}_{(j-1)}$ and $\Phat_{(j),*}= \Phat_{t_j-1}$.  Let $\ds \Phat_{(1),*}:=\Phat_{t_\train}$ (the initial estimate).

\item $\ds\Phat_{(j),\new,0} := [.]$ and $\ds\Phat_{(j),\new,k} := \Phat_{\that_j + k \alpha,\new}$. This is the $k^{\text{th}}$ estimate of $\bm{P}_{(j),\new}$ (again, conditioned on correct change time detection).

\item $\ds\Phat_{(j),\add} := [\Phat_{(j),*} \ \Phat_{(j),\new,K}]$ is the final estimate of $\bm{P}_{(j),\add}$.
\end{enumerate}
\end{definition}

Notice from the algorithm that, 
\begin{enumerate}
\item  $\Phat_{t,*}  = \Phat_{(j),*}$ for all $t \in [\that_{j-1} + K\alpha + (\vartheta+1)\alpha, \that_{j} + K\alpha + (\vartheta+1)\alpha-1]$
\item $\Phat_{t,\new} =  \Phat_{(j),\new,k-1}$ for all $t \in \J_{\uhat_j+k}$ for $k=1,2, \dots K$,  $\Phat_{t,\new} =  \Phat_{(j),\new,K}$ for $t \in [\that_j+K\alpha,  \that_j + K\alpha + (\vartheta+1)\alpha-1]$, and $\Phat_{t,\new} = [.]$ at all other times.

\item At all times, $\Phat_t = [\Phat_{t,*} \  \Phat_{t,\new}]$. 
\item $\Phat_{t-1,*}\perp\Phat_{t,\new}$ at $t=\that_j+k \alpha$ and so $\Phat_{(j),*} \perp \Phat_{(j),\new,k}$
%
\end{enumerate}


\begin{definition}
Define $\bm{G}_{j,k}:= (\P_{t_j})_{\mathcal{G}_{j,k}}$ for $k=1,2,\dots, \vartheta$. The clusters $\mathcal{G}_{j,k}$ were defined in Model \ref{clust_model}. Thus $\P_{(j+1),*} = \P_{(j)}=\P_{t_j} = [{G}_{j,1},{G}_{j,2}, \dots {G}_{j,\vartheta}]$.

Recall that $\hat{\bmG}_{j,k}$ is obtained in the cluster-PCA routine of Algorithm \ref{reprocsdet}. From the definition of $\Phat_{(j),*}$, $\Phat_{(j+1),*} = [\hat{G}_{j,1},\hat{G}_{j,2}, \dots \hat{G}_{j,\vartheta}]$.
\end{definition}

\begin{definition}\label{difzeta} \label{def_zetaknew_etc}
Define
\begin{enumerate}
\item $\ds \zeta_{j,*} := \mathrm{dif}(\hat{\bm{P}}_{(j),*},\bm{P}_{(j),*}) $
\item $\ds \zeta_{j,\new,k} := \mathrm{dif}([\Phat_{(j),*} \ \Phat_{(j),\new,k}],\bm{P}_{(j),\new}) $
\item $\ds \zeta_{j,\add} := \mathrm{dif}(\Phat_{(j),\add},\bm{P}_{(j),\add})$
\item $\ds \tilde{\zeta}_{j,k} := \mathrm{dif} ([\hat{\bmG}_{j,1} \dots \hat{\bmG}_{j,k}], \bmG_{j,k})$.
\end{enumerate}
Using the previous definition, clearly $\zeta_{j+1,*} \le \sum_{k=1}^\vartheta \tilde{\zeta}_{j,k}$.
%
\end{definition}

\begin{definition}
Define
\begin{enumerate}
\item $\ds \zeta_{j,*}^+ :=  r \zeta $ 
\item $\ds \zeta_{j,\new,0}^+ :=1$,  $\ds \zeta_{j,\new,k}^+ := \frac{b_{\bm{\mathcal{H}},k}}{b_{\bm{A}} - b_{\bm{A},\perp} - b_{\bm{\mathcal{H}},k}} $ for $k=1,2, \dots, K$ where $b_{\bm{A}}$, $b_{\bm{A},\perp}$, and $b_{\bm{\mathcal{H}},k}$ are defined in Lemmas \ref{Ak_cor}, \ref{Akperp_cor}, and \ref{calHk_cor} respectively. Their expressions use $\epsilon$ given by (\ref{def_eps}).
\item $\zeta_{j,\add}^+ :=  (r+r_\new) \zeta$. 
\item $\ds \tilde{\zeta}_k^+ := \frac{b_{\tilde{\bm{\mathcal{H}}},k}}{b_{\tilde{\bm{A}},k}-b_{\tilde{\bm{A}},k,\perp}-b_{\tilde{\bm{\mathcal{H}}},k}}$
where $b_{\tilde{\bm{\mathcal{H}}},k}$, $b_{\tilde{\bm{A}},k}$, and $b_{\tilde{\bm{A}},k,\perp}$ are defined in Lemmas \ref{Atil_cor}, \ref{Atilperp_cor}, and \ref{Htil_cor} respectively.
\end{enumerate}
We will show that these are high probability upper bounds on $\zeta_{j,*}$, $\zeta_{j,\new,k}$, $\zeta_{j,\add}$, and $\tilde{\zeta}_{j,k}$ under appropriate conditioning.  We should point out that $\zeta_{j,*}^+$,  $\zeta_{j,\add}^+$, and $\zeta_{j,\new,k}^+$ do not actually depend on $j$.  However, when analyzing Algorithm \ref{reprocsdet} without the c-PCA step, they do depend on $j$. 
\end{definition}

\begin{definition} \label{def_X}
Define the random variable
\[
X_u:= \{ \{\bnu_1, \bnu_2, \dots \bnu_{u \alpha}\}, \{\T_t\}_{t=1,2,\dots t_{\max}} \}.
\]
This is the random variable that we condition on (with appropriate choice of $u$) when analyzing the subspace update steps - detection or projection-PCA or cluster-PCA.
\end{definition}

\begin{definition}
Recall from Algorithm \ref{reprocsdet} that
\[
\thresh = \frac{\lamtrain}{2}.
\]
Also, recall the definition of ${\bm{\mathcal{D}}_u}$ from Algorithm \ref{reprocsdet}.
For $j = 1,\dots,J$, and for $a = u_j$ or $a = u_j + 1$, define the following events
\begin{itemize}
\item $\ds \mathrm{DET}_j^a := \left\{ \hat{u}_j = a \right\} $
\item $\ds \mathrm{PPCA}_{j,k}^a := \left\{ \hat{u}_j = a \ \text{\em and}\ \rank(\Phat_{(j),\new,k}) = r_{j,\new} \ \text{\em and}\ \zeta_{j,\new,k}\leq \zeta_{j,\new,k}^+  \right\}$ for  $k = 1,\dots,K$,
\item $\mathrm{CLUSTER}_{j}^a :=  \left\{ \hat{u}_j = a \text{ and } \hat{\mathcal{G}}_{j,k} = \mathcal{G}_{j,k} \text{ for } k =1,\dots,\vartheta  \right\}$
\item $\ds \mathrm{CPCA}_{j,k}^a := \left\{ \hat{u}_j = a \text{ and } \tilde{\zeta}_{j,k} \leq \tilde{\zeta}_{k}^+   \right\} $ for  $k = 1,\dots,\vartheta$,

\item $ \mathrm{NODETS}_j^a :=
\left\{ \hat{u}_j = a \ \text{\em and}\ \lambda_{\max}\left(\frac{1}{\alpha}\bm{\mathcal{D}}_u{\bm{\mathcal{D}}_u}'\right)<\thresh \ \text{\em for all}\ u \in [\hat{u}_j +K+ (\vartheta+1) + 1, u_{j+1}-1] \right\}
$
\\

\item $\Gamma_{0,\rmend} := \left\{ \zeta_{1,*}\leq r_0\zeta  \right\} \cap \left\{ \lambda_{\max}\left(\frac{1}{\alpha}\bm{\mathcal{D}}_u{\bm{\mathcal{D}}_u}'\right)<\thresh \ \text{\em for all}\ u \in [1, u_1-1] \right\} $

\item $\Gamma_{j,0}^{a} := \Gamma_{j-1,\rmend} \cap \mathrm{DET}_j^a$
\item $\ds \Gamma_{j,k}^a := \Gamma_{j,k-1}^a \cap \mathrm{PPCA}_{j,k}^a$ for $k=1,2, \dots K$

\item $\ds \tilde\Gamma_{j,0}^a := \Gamma_{j,K}^a \cap \mathrm{CLUSTER}_{j}^a$
\item $\ds \tilde\Gamma_{j,k}^a := \tilde\Gamma_{j,k-1}^a \cap  \mathrm{CPCA}_{j,k}^a$ for $k=1,2 \dots \vartheta$
\item $\ds \Gamma_{j,\rmend} := \Big( \tilde\Gamma_{j,\vartheta}^{u_j} \cap \mathrm{NODETS}_j^{u_j} \Big) \cup \left(  \tilde\Gamma_{j,\vartheta}^{u_j +1} \cap \mathrm{NODETS}_j^{u_j +1} \right)$
\end{itemize}
We misuse notation as follows.
Suppose that a set $\Gamma$ is a subset of all possible values that a r.v. $X$ can take. For two r.v.s' $\{X,Y\}$, when we need to say ``$X \in \Gamma$ and $Y$ can be anything" we will sometimes misuse notation and just say ``$\{X,Y\} \in \Gamma$."  For example, we sometimes say $X_{u_j}\in\Gamma_{j,\rmend}$.  This means $X_{u_j-1}\in\Gamma_{j,\rmend}$ and $\bm{a}_{t}$ for $t\in\J_{u_j}$ are unconstrained.
\end{definition}

\begin{definition}\label{defHk}
Define
\begin{enumerate}
\item Let $\bm{D}_{j,\new}:= (\bm{I} - \Phat_{(j),*} \Phat_{(j),*}{}')\bm{P}_{(j),\new} \overset{QR}{=} \bm{E}_{j,\rmnew} \bm{R}_{j,\rmnew}$ denote its reduced QR decomposition, i.e. let $\bm{E}_{j,\rmnew}$ be a basis matrix for $\Span\left(\bm{D}_{j,\new}\right)$ and let $\bm{R}_{j,\rmnew} = {\bm{E}_{j,\rmnew}}'\bm{D}_{j,\new}$.

\item Let $\bm{E}_{j,\rmnew,\perp}$ be a basis matrix for the orthogonal complement of $\Span(\bm{E}_{j,\rmnew})$. To be precise, $\bm{E}_{j,\rmnew,\perp}$ is an $n\times(n-r_j)$ basis matrix so that $[\bm{E}_{j,\rmnew} \ {\bm{E}_{j,\rmnew,\perp}}]$ is unitary.

\item
For $u = \hat{u}_{j}+k$ for $k = 1,\dots,K$, define $\bm{A}_{u}$, $\bm{A}_{u,\perp}$, $\bm{\mathcal{A}}_{u}$ as
\begin{align*}
\bm{A}_{u} &:= \frac{1}{\alpha} \sum_{t \in \mathcal{J}_{u}} {\bm{E}_{j,\rmnew}}' (\I - \hat{\bm{P}}_{(j),*}\hat{\bm{P}}_{(j),*}{}') \bm{\ell}_t {\bm{\ell}_t}' (\I - \hat{\bm{P}}_{(j),*}\hat{\bm{P}}_{(j),*}{}')\bm{E}_{j,\rmnew} \\
\bm{A}_{u,\perp} &:= \frac{1}{\alpha} \sum_{t \in \mathcal{J}_{u}} {\bm{E}_{j,\rmnew,\perp}}' (\I - \hat{\bm{P}}_{(j),*}\hat{\bm{P}}_{(j),*}{}') \bm{\ell}_t {\bm{\ell}_t}' (\I - \hat{\bm{P}}_{(j),*}\hat{\bm{P}}_{(j),*}{}') \bm{E}_{j,\rmnew,\perp}
\end{align*}
and let
\[
\bm{\mathcal{A}}_{u} := \left[ \begin{array}{cc} \bm{E}_{j,\rmnew} & \bm{E}_{j,\rmnew,\perp} \\ \end{array} \right]
\left[\begin{array}{cc} \bm{A}_{u} \ & \bm{0} \ \\ \bm{0} \ & \bm{A}_{u,\perp}  \\ \end{array} \right]
\left[ \begin{array}{c} {\bm{E}_{j,\rmnew}}' \\ {\bm{E}_{j,\rmnew,\perp}}' \\ \end{array} \right]
\]

\item For $u = \hat{u}_{j}+k$ for $k = 1,\dots,K$, define $\bm{\mathcal{M}}_{u}$ and $\bm{\mathcal{H}}_{u}$ as
\[
\M_{u} = (\I - \hat{\bm{P}}_{(j),*}\hat{\bm{P}}_{(j),*}{}')  \left(\frac{1}{\alpha} \sum_{t\in\J_u}\lhat_t\lhat_t{}'\right)(\I - \hat{\bm{P}}_{(j),*}\hat{\bm{P}}_{(j),*}{}')
\]
and
\[
\bm{\mathcal{H}}_{u} := \bm{\mathcal{M}}_u - \bm{\mathcal{A}}_u
\]
\end{enumerate}
\end{definition}

\begin{remark}[]
Recall the definition of ${\bm{\mathcal{D}}_u}$ from Algorithm \ref{reprocsdet}.
Conditioned on $\Gamma_{j,0}^{\uhat_j}$, for  $u = \uhat_j + k$, $k=1,2, \dots, K$, $\Phat_{u\alpha-1,*} = \Phat_{(j),*}$ and thus, for these values of $u$
\[
\frac{1}{\alpha}\bm{\mathcal{D}}_u {\bm{\mathcal{D}}_u}' = \M_u.
\]
For these $u$'s $\M_{u}$ is the matrix whose eigenvectors with eigenvalue above $\mathrm{thresh}$ form $\Phat_{(j),\new,k}$ (see step \ref{PCA} of Algorithm \ref{reprocsdet}). In other words, $\M_{u}$ has eigendecomposition
\begin{align*}
\M_u \overset{\mathrm{EVD}}{=} \left[ \begin{array}{cc} \Phat_{(j),\new,k} & \Phat_{(j),\new,k,\perp} \\ \end{array} \right]
\left[\begin{array}{cc}\hat{\bm{\Lambda}}_u \ & \bm{0} \ \\ \bm{0} \ & \ \hat{\bm{\Lambda}}_{u,\perp} \\ \end{array} \right]
\left[ \begin{array}{c}\Phat_{(j),\new,k}{}' \\ \Phat_{(j),\new,k,\perp}{}' \\ \end{array} \right].
\end{align*}
\end{remark}

\begin{definition}
Define
\ben
\item $\kappa_{s,*} := \max_{j}\kappa_s(\bm{P}_{(j),*})$ and $\kappa_{s,\rmnew} := \max_{j} \kappa_s(\bm{P}_{(j),\rmnew})$.
\item $\kappa_{s,*}^+:=0.3$ and $\kappa_{s,\new}^+:=0.0215$. As we will show later in Lemma \ref{Dnew0_lem}, $\kappa_{s,\new}^+$ upper bounds $\| {\I_{\T_t}}'\bm{D}_{j,\new}\|_2$ under appropriate conditioning.
\item $\phi^+:=1.2$. As we will show later in Lemma \ref{cslem_cor}, this upper bounds $\phi_t := \| [ ({\bm{\Phi}_{t})_{\mathcal{T}_t}}'(\bm{\Phi}_{t})_{\mathcal{T}_t}]^{-1} \|_2$ under appropriate conditioning.
\een
\end{definition}

\begin{definition}
Define $\bm{\Phi}_{(j),0}:= (I - \Phat_{(j),*} \Phat_{(j),*}{}')$ and $\bm{\Phi}_{(j),k}:= (I - \Phat_{(j),*} \Phat_{(j),*}' - \Phat_{(j),\new,k} \Phat_{(j),\new,k}{}')$ for $k=1,2, \dots K$.

Thus for $t \in [t_j, \that_j+\alpha]$ (before the first proj-PCA step), $\bm\Phi_t = \bm{\Phi}_{(j),0}$, for $t \in \J_{\uhat_j+k}$ (during interval used for $k$-th proj-PCA step), $\bm\Phi_t = \bm{\Phi}_{(j),k-1}$, for $t \in [\that_j+K\alpha, \that_j+K\alpha+ (\vartheta+1)\alpha]$ (after $K$-th proj-PCA step), $\bm\Phi_t = \bm{\Phi}_{(j),K}$ and for $t \in [\that_j+K\alpha+ (\vartheta+1)\alpha, t_{j+1}-1]$ (after cluster-PCA step), $\bm\Phi_t = \bm{\Phi}_{(j+1),0}$.
\end{definition}

\begin{remark}
The proof uses Model \ref{general_model} on $\T_t$. By Lemma \ref{spc_case}, Model \ref{sbyrho} is a special case of it. In particular, this means that (a) Model \ref{sbyrho} also implies $\rho^2 h^+ \le 0.01$ and (b) Model \ref{sbyrho} also allows us to use the support change lemma, Lemma \ref{blockdiag1}. This lemma and the sparse recovery lemma, Lemma \ref{cslem_cor}, are used to get bounds on quantities containing $\et$ in the proof of Lemma \ref{calHk_cor}.
\end{remark}

\subsection{Basic Lemmas} \label{basiclems}

\begin{lem} \label{initsub_cor}
Consider Algorithm \ref{reprocsdet}. Under Theorem \ref{thm1_cor} assumptions,
\begin{align*}
& \mathrm{dif}(\Phat_{t_\train}, \P_{t_\train}) \le r_0 \zeta \ \ \text{and} \\
& 0.8 \lambda^- \le \lamtrain \le 1.2 \lambda^-
\end{align*}
with probability at least $1-  n^{-10}$.
\end{lem}
This lemma follows in a fashion analogous to the proof of the p-PCA lemma, Lemma \ref{pPCA} (or actually just the proof of Lemma \ref{Ak_cor} which is one of the lemmas used to prove Lemma \ref{pPCA}). Its proof is in Appendix \ref{proof_initsub_cor}.

\begin{lem}\label{zetadecay}[{Bounds on $b_{\bm{A}},b_{\bm{A},\perp},b_{\bm{\mathcal{H}},k}$, $\zeta_{j,\new,k}$ and $\tilde\zeta_{k}^+$}]  \label{bounds_b_A}
Consider the quantities defined in Defnition \ref{def_zetaknew_etc}. Under the conditions of Theorem \ref{thm1_cor},
\begin{enumerate}
\item 
$b_{\bm{A}} - b_{\bm{\mathcal{H}},1} \geq 0.8\lambda^- > 0.5\lamtrain = \thresh $ and $b_{\bm{A},\perp} + b_{\bm{\mathcal{H}},1} \leq 0.2\lambda^- < 0.35 \lamtrain < \thresh$.

\item $\ds\zeta_{\new,0}^+=1$, $\ds\zeta_{\new,1}^+ \le 0.19$, $\ds\zeta_{\new,k}^+ \leq 0.19\cdot0.1^{k-1} + 0.15 r_{\new}\zeta$ for all $k \ge 1$. 
\item $\ds\tilde{\zeta}_{k}^+ \leq r_{j,k}\zeta$ where $r_{j,k}=|\mathcal{G}_{j,k}|$. 

\end{enumerate}
\end{lem}

This lemma essentially follows using simple algebra. We provide the proof in Appendix \ref{zetatil_bnd}. The proof of the second part is similar to that of Lemma 6.14 of \cite{rrpcp_isit15}.

\begin{lem}\label{blockdiag1}[{Support change lemma \cite[Lemma 5.3]{rrpcp_isit15}}]
Let $s_t = |\mathcal{T}_t|$. Consider a sequence of $s_t \times s_t$ symmetric positive-semidefinite matrices $\bm{A}_t$ such that
$\| \bm{A}_t\|_2 \leq \sigma^+$ for all $t$.  Assume that the $\mathcal{T}_t$ obey Model \ref{general_model}. 
Let $\ds\bm{M} = \sum_{t\in\J_u} \bm{I}_{\mathcal{T}_t} \bm{A}_t {\bm{I}_{\mathcal{T}_t}}'$ be an $n \times n$ matrix ($\I$ is an $n\times n$ identity matrix).
Then
\begin{align*}
\|\bm{M}\|_2 &\leq    \rho^2h^+ \alpha\sigma^+  \leq 0.0001\sigma^+  \alpha
\end{align*}
\end{lem}

\begin{lem}\label{Dnew0_lem}[{\cite{rrpcp_isit15}}]
Assume that the assumptions of Theorem \ref{thm1_cor} hold.
Conditioned on $X_{\uhat_j+k-1}$, for $X_{\uhat_j+k-1} \in \Gamma_{j,k-1}^{\uhat_j}$, for $\hat{u}_j = u_j$ or $\hat{u}_j = u_j+1$,
\begin{equation}\label{kappasplus}
\| {\I_{\mathcal{T}}}' \bm{D}_{j,\rmnew} \|_2 \le \kappa_{s,\new}^+ := .0215
\end{equation}
for all $\mathcal{T}$ such that $|\mathcal{T}|\leq s$.
\end{lem}

The following summarizes many simple facts.
\begin{fact} \label{d_large}\label{tilda_zeta_k_bnd} \label{1byalpha} \label{alphabdn} \
\begin{enumerate}

\item Observe that $\Gamma_{j,0}^{a}$ both for $a = u_j$ and $a=u_j+1$ implies that $u_j\leq\hat{u}_j \leq u_j + 1$. Thus, since $\that_j = \uhat_j \alpha$, in both cases, $t_j\leq \hat{t}_j \leq t_j + 2\alpha$.  So with the model assumption that $d\geq (K+2)\alpha$, we have that $\J_{\hat{u}_j+k} \subseteq [t_j, t_j+d]$ for $k=1,2,\dots, K$, i.e., for all the projection-PCA intervals, (\ref{anew_small}) holds and we can bound $\|\atnew\|_\infty$ by $\gamma_\new$.

\item Since, $\Gamma_{j,K}^{a} \subseteq \Gamma_{j,0}^{a}$, $\Gamma_{j,K}^{a}$ also implies  that $t_j \le \that_j \le t_j+ 2\alpha$. This along with $d_2 > (\vartheta+3)\alpha$ implies that all the intervals used for the cluster-estimation or the cluster-PCA steps are subsets of the interval in which the clustering assumption holds, i.e., $[\that_j+K \alpha+1, \that_j+K \alpha+(\vartheta+1)\alpha] \subseteq [t_j+K \alpha+1, t_j+K \alpha+d_2]$.
\\

\item Lemma \ref{zetadecay}, item 3, implies that, if  $\tilde{\zeta}_{j,k} \leq \tilde{\zeta}_{k}^+$ for $k=1,\dots,\vartheta$, then $\zeta_{j+1,*}:= \mathrm{dif}(\Phat_{(j+1),*},\bm{P}_{(j+1),*})  \leq \sum_{k=1}^\vartheta \tilde{\zeta}_{j,k} \le \sum_{k=1}^\vartheta r_{j,k} \zeta = r_j \zeta \le \zeta_{j+1,*}^+$. This follows by triangle inequality and the fact that $\Phat_{(j+1),*} = [\hat{G}_{j,1},\hat{G}_{j,2}, \dots \hat{G}_{j,\vartheta}]$ and $\P_{(j+1),*} = \P_{(j)} = [{G}_{j,1},{G}_{j,2}, \dots {G}_{j,\vartheta}]$.

\item Thus the event $\Gamma_{j,\rmend}$ implies $\zeta_{j+1,*} \leq \zeta_{j+1,*}^+$. Equivalently, $\Gamma_{j-1,\rmend}$ implies $\zeta_{j,*} \leq \zeta_{j,*}^+$

\item Thus, the event $\Gamma_{j,0}^a$ implies $\zeta_{j,*}\leq  \zeta_{j,*}^+ = r \zeta$ for $a = u_j$ or $a = u_{j+1}$.

\item Thus the event $\Gamma_{j,k-1}^a$ also implies this.

\item Lemma \ref{zetadecay}, item 2, and the choice of $K$ in the theorem imply that $\zeta_{j,\new,K}^+ \leq r_{\new}\zeta$.

\item Using the previous two items, the event $\Gamma_{j,K}^{\hat{u}_j}$,  both for $\uhat_j = u_j$ and $\uhat_j = u_j + 1$, implies that $\mathrm{dif}(\Phat_{(j),\add},\bm{P}_{(j),\add}) \leq \zeta_{j,*}^+ + r_{\new}\zeta = \zeta_{j,\add}^+$.
\\


\item  $\frac{1}{\alpha}\leq (r_{\new}\zeta)^2$.  To see this, observe that the lower bound for $\alpha$ has $(r_{\new}\zeta)^2$ in the denominator, and everything else in the expression is greater than or equal to 1. (Notice that $\frac{{\gamma_{\new}}^2}{\lambda^-}\geq1$)

\item  $b^\alpha\leq(r_{\new}\zeta)$. This follows because $b \le b_0=0.1$ and so $\frac{-\log (r_{\new}\zeta)}{-\log b} \le \frac{-\log (r_{\new}\zeta)}{-\log b_0} = \frac{\log \frac{1}{r_{\new}\zeta}}{2.3} \le \frac{1}{2.3} \frac{1}{r_{\new}\zeta} \le \frac{1}{(r_{\new}\zeta)^2} \le \alpha$. 


\end{enumerate}
\end{fact}

\begin{lem}[{Sparse Recovery Lemma (similar to \cite[Lemma 6.4]{rrpcp_perf} and \cite{rrpcp_isit15})}]\label{cslem_cor}
Assume that all of the conditions of Theorem \ref{thm1_cor} hold.
Recall that $\mathrm{SE}_t = \operatorname{dif}(\Phat_t,\bm{P}_t)$.
\begin{enumerate}
\item  Conditioned on $\Gamma_{j-1,\rmend}$, for $t \in [t_j, (\uhat_j+1)\alpha]$
\begin{enumerate}
\item $\phi_t := \| [ ({\bm{\Phi}_{t})_{\mathcal{T}_t}}'(\bm{\Phi}_{t})_{\mathcal{T}_t}]^{-1} \|_2 \leq \phi^+ := 1.2$.
\item  the support of $\xt$ is recovered exactly i.e. $\hat{\mathcal{T}}_t = \mathcal{T}_t$ and  $\et$ satisfies:
\begin{align}\label{etdef0_cor}
\et := \lt - \lhatt  = (\hat{\bm{x}}_t - \xt) - \wt  = \bm{I}_{\mathcal{T}_t} [ ({\bm{\Phi}_{t})_{\mathcal{T}_t}}'(\bm{\Phi}_{t})_{\mathcal{T}_t}]^{-1}  {\bm{I}_{\mathcal{T}_t}}' \bm{\Phi}_{t} (\bm{\ell}_t + \wt) - \wt
\end{align} 
\item Furthermore,
\begin{align*}
\mathrm{SE}_t &\leq 1 \ \text{, and} \\
\|\bm{e}_t\|_2 &\leq \frac{\phi^+}{1-b} (2\zeta_{,*}^+ \sqrt{r}\gamma +  \sqrt{r_{\new}}\gamma_{\rmnew} + 2\epsilon_w) \leq
1.34 \left(2\sqrt{\zeta} +  \sqrt{r_{\new}}\gamma_{\rmnew} + 2\epsilon_w  \right)
\end{align*}
\end{enumerate}

\item For $k=2,3, \dots, K$ and $\hat{u}_j = u_j$ or $\hat{u}_j = u_j +1$, conditioned on $\Gamma_{j,k-1}^{\hat{u}_j}$, for $t \in \J_{\hat{u}_j+k} = \left[(\uhat_j+ k-1)\alpha+1, (\uhat_j+ k)\alpha \right]$, the first two conclusions above hold. That is, $\phi_t \leq\phi^+$ and $\et$ satisfies (\ref{etdef0_cor}).  Furthermore,
\begin{align*}
\mathrm{SE}_t &\leq \zeta_{j,*}^+ + \zeta_{j,\new,k-1}^+  \ \text{, and} \\
\|\bm{e}_t\|_2 &\leq \frac{\phi^+}{1-b} (2\zeta_{j,*}^+ \sqrt{r}\gamma +  \zeta_{j,\new,k-1}^+ \sqrt{r_{\new}}\gamma_{\rmnew} + 2\epsilon_w)
\leq 1.34 \left(2.15\sqrt{\zeta} +  0.19\cdot(0.1)^{k-1}\sqrt{r_{\new}}\gamma_{\rmnew} + 2\epsilon_w  \right)
\end{align*}

\item For $\hat{u}_j = u_j$ or $\hat{u}_j = u_j +1$, conditioned on $\Gamma_{j,K}^{\uhat_j}$, for $t \in \left[\that_j+ K\alpha+1, \that_j+ K\alpha + (\vartheta+1)\alpha\right]$, the first two conclusions above hold ($\phi_t\leq\phi^+$ and $\et$ satisfies \eqref{etdef0_cor}).
Furthermore,
\begin{align*}
\SE_t &\leq \zeta_{j,\add}^+  \ \text{, and} \\
\|\bm{e}_t\|_2 &\leq  \frac{\phi^+}{1-b} (2\zeta_{j,\add}^+ \sqrt{r}\gamma + 2\epsilon_w) \leq 2.67( \sqrt{\zeta} + \epsilon_w)
\end{align*}

\item For $\hat{u}_j = u_j$ or $\hat{u}_j = u_j +1$, conditioned on $\tilde\Gamma_{j,\vartheta}^{\uhat_j}$, for $t \in \left[\that_j+ K\alpha + (\vartheta+1)\alpha+1, t_{j+1}-1\right]$, the first two conclusions above hold ($\phi_t\leq\phi^+$ and $\et$ satisfies \eqref{etdef0_cor}).
Furthermore,
\begin{align*}
\SE_t &\leq \zeta_{j+1,*}^+  \ \text{, and} \\
\|\bm{e}_t\|_2 &\leq  \frac{\phi^+}{1-b} (2\zeta_{j+1,*}^+ \sqrt{r}\gamma + 2\epsilon_w) \leq 2.67( \sqrt{\zeta} + \epsilon_w)
\end{align*}

\end{enumerate}

\end{lem}

Notice that cases 1) and 4) of the above lemma occur when the algorithm is in the detection phase;  during the intervals for case 2) the algorithm is performing projection-PCA; during the interval for case 3), the algorithm is performing cluster-PCA.
In case 1) new directions have been added but not estimated, so the error, $\et$, is the largest. In case 2), the error is decaying exponentially with each estimation step.  Case 3) occurs after the new directions have been successfully estimated but the old directions are not deleted yet. Case 4) occurs after the latter has been done too (after cluster-PCA is done). Case 4) contains the smallest error bound, with case 3) bounds being only slightly larger. The proof of this lemma is similar to the proof of Lemma 6.15 of \cite{rrpcp_isit15}. It is given in Appendix \ref{pf_cslem}. The main extra fact that we need to use now because the $\lt$'s follow an AR model is the following.

\begin{fact} \label{ltbnds}
From Model \ref{cor_model}, clearly $\|\lt\|_2 \le \frac{\sqrt{r}\gamma}{1-b}$. Moreover, $\lt$ can be expanded as follows.
\[
\lt = \l_{t,small} + \sum_{\tau=t-\alpha+1}^{t} b^{t-\tau} \P_\tau \bm{a}_\tau \ \text{where} \ \l_{t,small}:=\sum_{\tau=0}^{t-\alpha} b^{t-\tau} \bm{\nu}_{\tau}
\]
Using the geometric series sum formula, $b^\alpha \le r_\new \zeta$, and the bound on $\zeta$ from the theorem,
\[
\|\l_{t,small}\|_2 \le  \frac{b^\alpha \sqrt{r} \gamma}{1-b}   \le \frac{r_\new \zeta \sqrt{r} \gamma}{1-b}  \le \frac{\sqrt{\zeta}}{1-b}
\]

For $t \in [t_j, (\uhat_j+1)\alpha)$, conditioned on $\Gamma_{j-1,\rmend}$,
\[
\|\bm\Phi_t \lt\|_2 = \|\bm\Phi_{(j),0} \lt\|_2 \le  \frac{r_\new \zeta \sqrt{r} \gamma}{1-b}   + \frac{1}{1-b} \max_{\tau \in [t-\alpha+1, t]} \|\bm\Phi_{0} \P_\tau \bm{a}_\tau\|_2
\le  \frac{2r\zeta \sqrt{r}\gamma + \sqrt{r_\new}\gamma_\new}{1-b}
\le \frac{2\sqrt{\zeta} + \sqrt{r_\new}\gamma_\new}{1-b}
\]
For a $t \in \J_{\uhat_j + k}$ for $k=2,3,\dots K$, conditioned on $\Gamma_{j,k-1}^{\hat{u}_j}$,   for $\hat{u}_j = u_j$ or $\hat{u}_j = u_j +1$,
\[
\|\bm\Phi_t \lt\|_2 = \|\bm\Phi_{(j),k-1} \lt\|_2 \le  \frac{r_\new \zeta \sqrt{r} \gamma}{1-b}   + \frac{1}{1-b} \max_{\tau \in [t-\alpha+1, t]} \|\bm\Phi_{k-1} \P_\tau \bm{a}_\tau\|_2
\le  \frac{2r\zeta \sqrt{r}\gamma + \zeta_{\new,k-1}^+ \sqrt{r_\new}\gamma_\new}{1-b}
\]
and the above can further be bounded by $\frac{2\sqrt{\zeta}  + \zeta_{\new,k-1}^+ \sqrt{r_\new}\gamma_\new}{1-b}$.

Using $\zeta_{\new,K}^+ \le r_\new \zeta$ (follows using Lemma \ref{zetadecay} and expression for $K$) and the bound on $\zeta$,
for $t \in [\that_j+K\alpha+1, \that_j+K\alpha + (\vartheta+1)\alpha]$, conditioned on $\Gamma_{j,K}^{\uhat_j}$,
\[
\|\bm\Phi_t \lt\|_2 = \|\bm\Phi_{(j),K} \lt\|_2 \le \frac{(2r\zeta+r_\new\zeta) \sqrt{r} \gamma}{1-b} \le  \frac{2\sqrt{\zeta}}{1-b}
\]
Using Fact \ref{d_large}, item 3, for $t \in [\that_j+K\alpha + (\vartheta+1)\alpha+1, t_{j+1}-1]$, conditioned on $\tilde\Gamma_{j,\vartheta}^{\uhat_j}$, $\zeta_{j+1,*} \le \zeta_{j+1,*}^+ = r\zeta$ and so
\[
\|\bm\Phi_t \lt\|_2 = \|\bm\Phi_{(j+1),0} \lt\|_2 \le \frac{2r\zeta \sqrt{r} \gamma}{1-b} \le  \frac{2\sqrt{\zeta}}{1-b}
\]

Recall that $\b_t:=\bm{\Phi}_t (\lt + \wt)$. Thus, using the above, we get that $\|\b_t\|_2 \le \|\bm\Phi_{t} \lt\|_2 + \|\wt\|_2 \le \xi$ ($ \xi$ is set in Theorem \ref{thm1_cor}).
\end{fact}


\subsection{Main lemmas for proving Theorem \ref{thm1_cor} and proof of Theorem \ref{thm1_cor}} \label{pf_thm}

The first three lemmas below deal with analyzing the addition step. They have statements which are exactly the same as the corresponding lemmas in \cite{rrpcp_isit15}. But the proofs of the key lemmas needed for proving them are very different since the $\lt$'s are now correlated over time. We thus relegate the proofs of these lemmas to the appendix. The proofs of the key lemmas needed for these are given in the main text though. The fourth and the fifth lemma below deal with the deletion step (cluster-PCA) and these are new. These are proved in this section itself.

\begin{lem}[No false detection of subspace changes] \label{falsedet_del}
$\ds  \Pr\left(\mathrm{NODETS}_j^a \ | \ \tilde\Gamma_{j,\vartheta}^a \right) = 1$ for $a = u_j$ or $a = u_j + 1$.
\end{lem}

\begin{lem}[Subspace change detected within $2\alpha$ frames] \label{det}
For $j = 1,\dots,J$,
\[
\Pr\left(\mathrm{DET}_j^{u_j+1} \ | \ \Gamma_{j-1,\rmend}, \overline{\mathrm{DET}^{u_j}} \right) \geq p_{\det,1} := 1 - p_{\bm{A}} - p_{\bm{\mathcal{H}}}.
\]
The definitions of $p_{\bm{A}}$ and $p_{\bm{\mathcal{H}}}$ can be found in the proofs of Lemmas \ref{Ak_cor} and \ref{calHk_cor} respectively.
\end{lem}

\begin{lem}[$k$-th iteration of pPCA works well] \label{pPCA}
\[
\Pr\left(\Gamma_{j,k}^a \ | \ \Gamma_{j,k-1}^a \right) = \Pr\left(\mathrm{PPCA}_{j,k}^a \ | \ \Gamma_{j,k-1}^{a} \right)  \geq p_{\mathrm{ppca}} := 1 - p_{\bm{A}} - p_{\bm{A},{\perp}} - p_{\bm{\mathcal{H}}}
\]
for $a = u_j$ or $a = u_{j}+1$.
The definitions of $p_{\bm{A}}$, $p_{\bm{A},{\perp}}$, and $p_{\bm{\mathcal{H}}}$ can be found in the proofs of Lemmas \ref{Ak_cor}, \ref{Akperp_cor}, and \ref{calHk_cor} respectively.
\end{lem}


\begin{lem}[Clusters are correctly estimated]\label{cluster_lem}
\[
\Pr\left(  \mathrm{CLUSTER}_{j}^a \ \big| \ \Gamma_{j,K}^a \right) \geq p_{\mathrm{cluster}} = 1 -  p_{\mathrm{cl}} - p_{\tilde{\bm{le}}} - p_{\tilde{\bm{ee}}} 
\]
for $a=u_j$ or $a=u_{j}+1$. The definition of $p_{\mathrm{cl}}$ can be found in the proof of Lemma \ref{cluster_lem_azuma} and definition of $p_{\tilde{\bm{le}}}$, $p_{\tilde{\bm{ee}}}$ can be found in the proof of Lemma \ref{Htil_cor}.
\end{lem}

\begin{lem}[Subspaces corresponding to each cluster are correctly estimated] \label{cPCA}
\[
\Pr\left(\mathrm{CPCA}_{j,k}^a \ | \ \tilde\Gamma_{j,k-1}^a  \right) \geq p_{\mathrm{cpca}} := 1 - p_{\tilde{\bm{A}}} - p_{\tilde{\bm{A}},\perp} - p_{\tilde{\bm{\mathcal{H}}}}
\]
for $a=u_{j+1}$ or $a=u_{j+1}+1$.  The probabilities $p_{\tilde{\bm{A}}}$, $p_{\tilde{\bm{A}},\perp}$, $p_{\tilde{\bm{\mathcal{H}}}}$ are defined in the proofs of Lemmas \ref{Atil_cor}, \ref{Atilperp_cor}, and \ref{Htil_cor} respectively.

Using Fact \ref{d_large}, $\bigcap_{k=1}^{\vartheta} \mathrm{CPCA}_{j,k}^a$ implies that $\zeta_{(j+1),*} \le \zeta_{(j+1),*}^+ = r \zeta$. Thus, $\Gamma_{j,\rmend}^a$ also implies this.
\end{lem}

\begin{corollary}\label{GammaCorDel} \label{main_cor}
Let  $p_{\det,0}: = \Pr\left(\mathrm{DET}_j^{u_j} \ | \ \Gamma_{j-1,\rmend} \right).$
Combining Lemmas \ref{falsedet_del}, \ref{det}, \ref{pPCA}, \ref{cluster_lem}, and \ref{cPCA} gives
\begin{align*}
\Pr\left( \Gamma_{j,\rmend} \ | \ \Gamma_{j-1,\rmend} \right) &=
\Pr\bigg( \Big( \mathrm{DET}_j^{u_j} \bigcap_{k=1}^K\mathrm{PPCA}_{j,k}^{u_j}  \bigcap \mathrm{CLUSTER}_{j}^{u_j}\bigcap_{k=1}^\vartheta \mathrm{CPCA}_{j,k}^{u_j}  \Big) \bigcup \\
&\hspace{.5in} \Big(\overline{\mathrm{DET}_j^{u_j}}\cap\mathrm{DET}_j^{u_j+1}  \bigcap_{k=1}^K\mathrm{PPCA}_{j,k}^{u_j+1} \bigcap \mathrm{CLUSTER}_{j}^{u_j+1}\bigcap_{k=1}^\vartheta \mathrm{CPCA}_{j,k}^{u_j+1}   \Big) \ \big| \ \Gamma_{j-1,\rmend}  \bigg)\\%
%
&\geq p_{\det,0} \cdot  (p_{\mathrm{ppca}})^K \cdot (p_{\mathrm{cluster}})\cdot(p_{\mathrm{cpca}})^{\vartheta}
+ (1-p_{\det,0})\cdot  p_{\det,1} \cdot (p_{\mathrm{ppca}})^K \cdot(p_{\mathrm{cluster}})\cdot(p_{\mathrm{cpca}})^{\vartheta} \\
&\geq  p_{\det,1} (p_{\mathrm{ppca}})^K \cdot p_{\mathrm{cluster}} (p_{\mathrm{cpca}})^{\vartheta}
\end{align*}
\end{corollary}

\begin{proof}[{Proof of Theorem \ref{thm1_cor} and Corollary \ref{thm1_cor_corol}}]
Using the fact that $\Gamma_{j-1,\rmend}\subseteq\Gamma_{j-2,\rmend}\subseteq\dots\subseteq \Gamma_{1,\rmend}\subseteq\Gamma_{0,\rmend}$,
$\Pr(\Gamma_{J,\rmend})  = \Pr(\Gamma_{0,\rmend}) \prod_{j=1}^{J} \Pr(\Gamma_{j,\rmend} \ | \ \Gamma_{j-1,\rmend})$.

By Lemma \ref{initsub_cor} and the argument used to prove Lemmas \ref{cslem_cor} and \ref{falsedet_del}, we get that $\Pr(\Gamma_{0,\rmend}) \ge 1 - n^{-10}$.
Thus, using Corollary \ref{GammaCorDel}, and the lower bound on $\alpha$,
\[
\Pr(\Gamma_{J,\rmend}) \ge (1 - n^{-10}) \left( p_{\det,1} (p_{\mathrm{ppca}})^K \cdot p_{\mathrm{cluster}} (p_{\mathrm{cpca}})^{\vartheta} \right)^J
\ge (1 - n^{-10})(p_{\mathrm{ppca}})^{(K+1)J} (p_{\mathrm{cluster}} (p_{\mathrm{cpca}})^{\vartheta})^J \ge (1 - n^{-10})^3 \geq 1 - 3n^{-10}.
\]
By Fact \ref{d_large}, Lemma \ref{cslem_cor}, and Lemma \ref{zetadecay}, $\Gamma_{J,\rmend}$ implies that $\That_t = \T_t$ for all times $t$; and that all the bounds on the subspace error $\SE_t$ and on $\et$ hold.
\end{proof}

\subsection{Key lemmas needed for proving the main lemmas} \label{keylems} 
The following lemma follows from the $\sin\theta$ theorem \cite{davis_kahan} (Theorem \ref{sintheta} in Appendix \ref{prelim}) and Weyl's inequality.  It is taken from \cite{rrpcp_perf}.
\begin{lem}[\cite{rrpcp_perf}, Lemma 6.9]\label{zetakbnd}
At $u = \hat{u}_j + k$,
if $\rank(\Phat_{(j),\new,k}) = r_{j,\new}$, and if $\lambda_{\min}(\bm{A}_u) - \|\bm{A}_{u,\perp}\|_2 - \|\bm{\mathcal{H}}_u\|_2 >0$, then
\begin{align} \label{zetakbound}
\zeta_{j,\new,k} \leq  \frac{\|\bm{\mathcal{H}}_u\|_2}{\lambda_{\min} (\bm{A}_u) - \|\bm{A}_{u,\perp}\|_2 - \|\bm{\mathcal{H}}_u\|_2}.
\end{align}
Similarly, if $\hat{\mathcal{G}}_{j,k} = \mathcal{G}_{j,k}$ and $\lambda_{\min}(\tilde{\bm{A}}_{j,k}) - \|\tilde{\bm{A}}_{j,k,\perp}\|_2 - \| \tilde{\bm{\mathcal{H}}}_{j,k} \|_2 > 0$, then
\begin{equation}\label{sintheta_del}
\tilde{\zeta}_{j,k} \leq \frac{\|\tilde{\bm{\mathcal{H}}}_{j,k} \|_2}{\lambda_{\min}(\tilde{\bm{A}}_{j,k}) - \|\tilde{\bm{A}}_{j,k,\perp}\|_2 - \| \tilde{\bm{\mathcal{H}}}_{j,k} \|_2}
\end{equation}
\end{lem}

The next three lemmas (\ref{Ak_cor}, \ref{Akperp_cor}, and \ref{calHk_cor}) each assert a high probability bound for one of the terms in \eqref{zetakbound}. These, along with Lemma \ref{zetakbnd}, are used to prove Lemmas \ref{det} and \ref{pPCA}.  The proofs of these lemmas use the matrix Azuma inequalities (Lemmas \ref{azuma}, \ref{azuma_hermitian} or \ref{azuma_norm} in the Appendix) and hence we refer to them as the ``addition Azuma" lemmas.
Let
\beq
\epsilon = \frac{1}{1-b^2} 0.001 r_\new \zeta \lambda^-
\label{def_eps}
\eeq
\begin{lem}\label{Ak_cor} %
Define
\[
b_{\bm{A}}:=  \frac{1}{1-b^2}\left( (1 - (\zeta_*^+)^2)\lambda_\new^- -  (r_\new \zeta)^2 \frac{b^2}{1-b^2}(1-\zeta_*^+)^2\lambda_\new^- \right)  - 4\epsilon 
\]
For $k = 1,\dots,K$, for all $X_{\hat{u}_j+k-1}\in\Gamma_{j,k-1}^{\hat{u}_j}$ with $\hat{u}_j=u_j$ or $\hat{u}_j=u_j+1$,
\begin{align*}
\Pr\left( \lambda_{\min} \left(\bm{A}_{\hat{u}_j + k} \right)\geq b_{\bm{A}} \ \big| \ X_{\hat{u}_j+k-1} \right) \geq 1 -  p_{\bm{A}}
\end{align*}
where $p_{\bm{A}}$ is defined in the proof.
\end{lem}

\begin{lem}\label{Akperp_cor} 
Define
\[
b_{\bm{A},\perp}:= \frac{1}{1-b^2} (\zeta_*^+)^2 \lambda^+ +  \frac{0.05(r_\new \zeta)^2b^2\lambda^-}{(1-b^2)(1-b)^2}  + 4 \epsilon 
\]
For $k = 1,\dots,K$, for all $X_{\hat{u}_j+k-1}\in\Gamma_{j,k-1}^{\hat{u}_j}$ with $\hat{u}_j=u_j$ or $\hat{u}_j=u_j+1$,
\begin{align*}
\Pr\left( \lambda_{\max}\left(\bm{A}_{\hat{u}_j + k,\perp} \right) \leq b_{\bm{A},\perp}  \ \big| \ X_{\hat{u}_j+k-1} \right)\geq 1 - p_{\bm{A},{\perp}}
\end{align*}
where $p_{\bm{A},{\perp}}$ is defined in the proof.
\end{lem}

\begin{lem}\label{calHk_cor} 
Define
\[
b_{\bm{\mathcal{H}},k} := 2b_{\l \e,k} + b_{\e \e,k} + 2b_{\bm{F_k}}  
\]
where for $k \ge 2$,
\begin{align*}
b_{\l \e, k}& :=  \frac{1}{1-b^2}( \sqrt{\rho^2 h^+ } \phi^+  (\zeta_*^+)^2 \lambda^+ + \sqrt{\rho^2 h^+ } \phi^+  \zeta_{\new,k-1}^+ \lambda_\new^+ ) +  \frac{0.05(r_\new \zeta) b^2 \lambda^-}{(1-b^2)(1-b)^2} + 6 \epsilon \\
%
b_{\e \e, k} &:=  \frac{1}{1-b^2}  (\rho^2 h^+  \ (\phi^+)^2  (\zeta_*^+)^2 \lambda^+ + \rho^2 h^+  \ (\phi^+)^2  (\zeta_{\new,k-1}^+)^2 \lambda_\new^+)
+  \frac{0.05(r_\new \zeta) b^2 \lambda^-}{(1-b^2)(1-b)^2}
+ (\phi^+)^2 (0.06 r_\new \zeta \lambda^-)
+ 8 \epsilon \\
%
b_{\bm{F},k} &:= \frac{1}{1-b^2}  (\zeta_*^+)^2 \lambda^+ +  \frac{0.05(r_\new \zeta) b^2 \lambda^-}{(1-b^2)(1-b)^2} + 4 \epsilon
\end{align*}
and for $k=1$,
\begin{align*}
b_{\l \e, 1}& :=  \frac{1}{1-b^2}( \sqrt{\rho^2 h^+ } \phi^+  (\zeta_*^+)^2 \lambda^+ + \phi^+\kappa_{s,\new}^+ \lambda_\new^+ ) +  \frac{0.05(r_\new \zeta) b^2 \lambda^-}{(1-b^2)(1-b)^2} + 6 \epsilon \\
%
b_{\e \e, 1} &:=  \frac{1}{1-b^2}  (\rho^2 h^+  \ (\phi^+)^2 (\zeta_*^+)^2 \lambda^+ + \rho^2 h^+  \ (\phi^+)^2 (\kappa_{s,\new}^+)^2 \lambda_\new^+)
+  \frac{0.05(r_\new \zeta) b^2 \lambda^-}{(1-b^2)(1-b)^2} \rho^2 h^+
+ (\phi^+)^2 (0.06 r_\new \zeta \lambda^-)
+ 8 \epsilon \\
b_{\bm{F},1} &:= \frac{1}{1-b^2} (\zeta_*^+)^2 \lambda^+ +  \frac{0.05(r_\new \zeta) b^2 \lambda^-}{(1-b^2)(1-b)^2} + 4 \epsilon
\end{align*}
For $k = 1,\dots,K$, for all $X_{\hat{u}_j+k-1}\in\Gamma_{j,k-1}^{\hat{u}_j}$ with $\hat{u}_j=u_j$ or $\hat{u}_j=u_j+1$,
\begin{align}
&\Pr \left(\|\bm{\mathcal{H}}_{\hat{u}_j + k}\|_2 \leq b_{\bm{\mathcal{H}},k}  \ \big| \ X_{\hat{u}_j+k-1}\right)\geq 1 - p_{\bm{\mathcal{H}}}
\label{normHk_cor}
\end{align}
where $p_{\bm{\mathcal{H}}}:=  p_{\l \e} + p_{\e \e} + p_{\bm{F}}$ and $p_{\l \e}$, $p_{\e \e}$ and $p_{\bm{F}}$ are defined in the proof.
\end{lem}

\begin{fact}
\label{ppca_bound}
Using $\rho^2h^+ \leq 10^{-4}$, $\frac{\lambda_\new^+}{\lambda^-} \leq  3$, $\lambda_\new^- \ge \lambda^-$, $\phi^+=1.2$, $\kappa_{s,\new}^+ = 0.0215$, $b\leq 0.1$,  $\zeta \leq \min\{\frac{10^{-4}}{(r+r_\new)^2},\frac{0.003\lambda^-}{(r+r_\new)^2\lambda^+}\}$, $\zeta_*^+ = r\zeta$, $\epsilon = \frac{1}{1-b^2}0.001r_\new\zeta$,
\bea
b_{\bm{A}} & \geq & \frac{\lambda^-}{1-b^2} \left(0.9999 - 0.005r_\new\zeta \right) \nn \\
b_{\bm{A},\perp} &  \leq & \frac{0.008r_\new\zeta\lambda^-}{1-b^2} \nn \\ 
b_{\bm{\mathcal{H}},1}  &  \leq & \frac{\lambda^-}{1-b^2}(0.156 + 0.1 r_\new\zeta) \nn \\
b_{\bm{\mathcal{H}},k}  &  \leq & \frac{\lambda^-}{1-b^2}(0.073\zeta_{\new,k-1}^+ + 0.1r_\new\zeta)\nn
\eea
\end{fact}

The following lemma is needed for the proof of Lemma \ref{cluster_lem}.
\begin{lem} \label{cluster_lem_azuma}
Let $\that_{\mathrm{cl}}:=\that_{j}+ K\alpha+1$. Let $\bbb:= 0.05 \lammin$.
\begin{align*}
\Pr\left( \| \frac{1}{\alpha}\sum_{t=\that_{cl}}^{\that_{cl} + \alpha-1} \lt {\lt}' - \frac{1}{1-b^2} \Sigma_{(j)} \| \le \bbb  \ \big| \ X_{\uhat_j+K} \right) \ge 1-p_{\mathrm{cl}}.
\end{align*}
for all $X_{\uhat_j+K} \in \Gamma_{j,K}^{\uhat_{j}}$ for ${\uhat_{j}}=u_{j}$ or $u_{j}+1$
In the above, $p_{\mathrm{cl}}$ is defined in the proof.
\end{lem}

The next three lemmas are needed for the proof of Lemma \ref{cPCA}. The third one below is also used in the proof of Lemma \ref{cluster_lem}.

\begin{lem}\label{Atil_cor}
Define
\[
b_{\tilde{\bm{A}},k} := (1 - r^2\zeta^2 ) (1 - \frac{(r_\new \zeta)^2 b^2}{1-b^2}) \frac{1}{1-b^2}\lambda_{j,k}^- - 4\epsilon 
\]
For $j=1,\dots,J$ and $k=1,\dots,\vartheta$, for $a=u_{j}$ or $a=u_{j}+1$, for all $X_{(\uhat_j+K+1)+k-1}\in \tilde{\Gamma}_{j,k-1}^a$,
\[
\mathbb{P} \left(\lambda_{\min} (\tilde{\bm{A}}_{j,k}) \geq b_{\tilde{\bm{A}},k} \ \big| \ X_{(\uhat_j+K+1)+k-1}  \right) > 1- p_{\tilde{\bm{A}}}
\]
where $p_{\tilde{\bm{A}}}$ is defined in the proof.
\end{lem}

\begin{lem}\label{Atilperp_cor}
Define
\[
b_{\tilde{\bm{A}},\perp,k} := \frac{1}{1-b^2}(2(r\zeta)^2 \lambda^+ + \lambda_{k+1}^+) + \frac{0.05(r_\new \zeta)  b^2}{(1-b^2)(1-b)^2}  + 4 \epsilon 
\]
For $j=1,\dots,J$ and $k=1,\dots,\vartheta$, for $a=u_{j}$ or $a=u_{j}+1$, for all $X_{(\uhat_j+K+1)+k-1}\in \tilde{\Gamma}_{j,k-1}^a$,
\[
\Pr \left(\lambda_{\max}(\tilde{\bm{A}}_{j,k,\perp}) \leq b_{\tilde{\bm{A}},\perp,k}  \ \big| \ X_{(\uhat_j+K+1)+k-1} \right) > 1- p_{\tilde{\bm{A}},\perp,k}
\]
where $p_{\tilde{\bm{A}},\perp,k}$ is defined in the proof.
\end{lem}

\begin{lem}\label{Htil_cor}
Define
\[
b_{\tilde{\bm{\mathcal{H}}},k} := 2b_{\tilde{\bm{\ell e}},k} + b_{\tilde{\bm{ee}},k} + 2b_{\tilde{\bm{F}},k}
\]
where
$b_{\tilde{\bm{\ell e}},k}: = \sqrt{\rho^2 h^+ (\phi^+)^2}   ( \frac{1}{1-b^2} (r + r_\new) \zeta ( (r\zeta) \lambda^+ + \lambda_k^+) ) + \frac{0.05(r_\new \zeta)  b^2\lambda^-}{(1-b^2)(1-b)^2} + 6 \epsilon$ 
\\
%
$b_{\tilde{\bm{e e}},k}: = \rho^2 h^+  \ (\phi^+)^2 \frac{1}{1-b^2}  (r \zeta) (r+r_\new)\zeta \lambda^+ +  \frac{0.05(r_\new \zeta) b^2 \lambda^-}{(1-b^2)(1-b)^2} + (\phi^+)^2 2(0.03 \zeta \lambda^-) + 8 \epsilon$ 
\\
%
$b_{\tilde{\bm{F}},k}:= \frac{1}{1-b^2} \left( (r\zeta)^2\lambda^+ + \frac{(r\zeta)^2}{\sqrt{1 - (r\zeta)^2}}\lambda_{k+1}^+ \right) + \frac{0.05(r_\new \zeta)  b^2\lambda^-}{(1-b^2)(1-b)^2} + 4\epsilon$ 

For $k=1,\dots,k$, for $a=u_{j}$ or $a=u_{j}+1$, for all $X_{(\uhat_j+K+1)+k-1}\in \tilde{\Gamma}_{j,k-1}^a$,
\begin{align*}
&\Pr \left( \| \tilde{\bm{\mathcal{H}}}_k \|_2 \leq b_{\tilde{\bm{\mathcal{H}}},k} \ \big| \ X_{(\uhat_j+K+1)+k-1} \right) \geq 1 - p_{\tilde{\bm{\mathcal{H}}}}
\end{align*}
where $p_{\tilde{\bm{\mathcal{H}}}}:= p_{\tilde{\bm{\ell e}}} + p_{\tilde{\bm{e e}}} + p_{\tilde{\bm{F}}}$ and
$p_{\tilde{\bm{\ell e}}}$, $p_{\tilde{\bm{e e}}}$, $p_{\tilde{\bm{F}}}$ are defined in the proof.

Also,  for $a=u_{j}$ or $a=u_{j}+1$, for all $X_{\uhat_j+K}\in \Gamma_{j,K}^a$,
\[
\Pr \left(2\|\frac{1}{\alpha} \sum_{t=\hat{t}_j + K\alpha +1}^{\hat{t}_j + (K+1)\alpha +1} \lt \et' \|_2 + \| \frac{1}{\alpha}\sum_t \e_t\e_t' \|_2  \leq  b_{\tilde{\bm{\mathcal{H}}},1}  \big| \ X_{\uhat_j+K} \right) \geq 1 - p_{\tilde{\bm{\ell e}}} - p_{\tilde{\bm{e e}}}
\]
(This is used in the proof of Lemma \ref{cluster_lem_azuma}. It follows using the exact same approach as that used to bound $\| \tilde{\bm{\mathcal{H}}}_1 \|_2$.)
\end{lem}

\begin{fact}
\label{del_bounds}
Using $\rho^2h^+ \leq 10^{-4}$, $\frac{\lambda_k^+}{\lambda_k^-} \leq g^+ = 3$, $\phi^+=1.2$, $\kappa_{s,\new}^+ = 0.0215$, $b\leq 0.1$, $\frac{\lambda_{k+1}^+}{\lambda_k^-} \leq \chi^+ = 0.2$, $\zeta \leq \min\{\frac{10^{-4}}{(r+r_\new)^2},\frac{0.003\lambda^-}{(r+r_\new)^2\lambda^+}\}$, $\zeta_*^+ = r\zeta$, $\epsilon = \frac{1}{1-b^2}0.001r_\new\zeta$, we have 
\bea
b_{\tilde{\bm{A}},k} & \geq & \frac{\lambda_k^-}{1-b^2}(0.9999 - 0.005r_\new\zeta) \\
b_{\tilde{\bm{A}},\perp,k} & \leq & \frac{\lambda_k^-}{1-b^2}(0.2 + 0.07r_\new\zeta) \\
b_{\tilde{\bm{\mathcal{H}}},k} & \leq & \frac{\lambda_k^-}{1-b^2}(0.072(r+r_\new)\zeta + 0.095r_\new\zeta)
\eea
\end{fact}

\subsection{Proofs of the main lemmas} \label{pf_mainlems}  
Lemmas \ref{falsedet_del}, \ref{det}, and \ref{pPCA} are proved in Appendix \ref{old_main_lemmas}. These use the first three lemmas from the above subsection.

\begin{proof}[Proof of Lemma \ref{cluster_lem}]
In this proof, all of the probabilistic statements are conditioned on $X_{\uhat_j+K} \in \Gamma_{j,K}^{\uhat_{j}}$ for ${\uhat_{j}}=u_{j}$ or $u_{j}+1$.


Let $\that_{\mathrm{cl}}:=\that_{j}+ K\alpha+1$. Recall from Algorithm \ref{reprocsdet} that $\ds\hat{\bm{\Sigma}}_{\mathrm{sample}} := \frac{1}{\alpha}\sum_{t=\that_{cl}}^{\that_{cl} + \alpha-1} \lhatt \lhatt{}'$. Define $\ds\bm{\Sigma}_{\mathrm{sample}}:= \frac{1}{\alpha}\sum_{t=\that_{cl}}^{\that_{cl} + \alpha-1} \lt {\lt}'$.

By Lemma \ref{cluster_lem_azuma} and Lemma \ref{event}, under the given conditioning, with probability (w.p.) at least $1 - p_{\mathrm{cl}}$,
\begin{equation}
\lambda_{\max}(\bm{\Sigma}_{\mathrm{sample}} - \frac{1}{1-b^2} \bm\Sigma_{(j)}) \le \bbb := 0.05 \lambda^-
\label{cluster_lem_azuma_bnd}
\end{equation}

Let $k_0=0$. Let $k_i$ denote the last index of cluster $i$. Thus true cluster 1, $\mathcal{G}_{j,1} = \{1,2, \dots k_1\}$, true cluster 2, $\mathcal{G}_{j,2} = \{k_1+1, k_1+2, \dots k_2\}$ and so on
 for all $i=1,2, \dots \vartheta_j$. Recall that $\bm{\Sigma}_{(j)}$ has rank $r_{j}$ and so $k_{\vartheta_j} = r_{j}$.


Consider ``true cluster" 1. We need to show that ``estimated cluster" 1, $\hat{\mathcal{G}}_{j,1} = \{1,2, \dots k_1\}$. Let $\hat{\lambda}_i := \lambda_{i}\left(\hat{\bm{\Sigma}}_{\mathrm{sample}}\right)$. We will be done if we can show that
\ben
\item  $\ds\frac{\hat\lambda_1}{\hat\lambda_{k_1}} \le \ghatp$ and
\item $\ds\frac{\hat\lambda_1}{\hat\lambda_{k_1+1}} > \ghatp$
\een
Define
\[
\bb: = \left\|\bm{\Sigma}_{\mathrm{sample}} - \hat{\bm{\Sigma}}_{\mathrm{sample}}\right\|_2
\]
Using the fact that $\lhatt = \lt - \et$ we get that
\[
\bb \le 2 \bigg\|\frac{1}{\alpha} \sum_t \lt {\et}'\bigg\| + \bigg\| \frac{1}{\alpha} \sum_t \et {\et}'\bigg\|
\]
Using Lemma \ref{Htil_cor}, Fact \ref{del_bounds} and $(r+r_\new) \zeta \lambda_k^- \le (r+r_\new) \zeta \lambda^+ \le 0.0003 \lambda^-$ (from the bound on $\zeta$), under the given conditioning,
\[
\begin{array}{ll}
\bb &\le \frac{\lambda_k^-}{1-b^2}(0.072(r+r_\new)\zeta + 0.095r_\new\zeta) \\
&< 0.01 \lambda^- 
\end{array}
\]
with probability at least $1 - p_{\tilde{\bm{le}}} - p_{\tilde{\bm{ee}}}$ where $p_{\tilde{\bm{le}}}$, $p_{\tilde{\bm{ee}}}$ are defined in Lemma \ref{Htil_cor}.
%

Using Weyl's inequality and \eqref{cluster_lem_azuma_bnd}, for $i=1,\dots,n$
\begin{align*}
\hat{\lambda}_i:= \lambda_i(\hat{\bm{\Sigma}}_{\mathrm{sample}}) 
&\leq \lambda_i(\bm{\Sigma}_{\mathrm{sample}}) + \lambda_{\max}(\hat{\bm{\Sigma}}_{\mathrm{sample}} - \bm{\Sigma}_{\mathrm{sample}}) \\
&\leq \lambda_i(\frac{1}{1-b^2}\bm{\Sigma}_{(j)}) + \lambda_{\max}(\bm{\Sigma}_{\mathrm{sample}} -\frac{1}{1-b^2}\bm{\Sigma}_{(j)} ) + \bb \\
& \le \lambda_i(\frac{1}{1-b^2}\bm{\Sigma}_{(j)}) + \bbb + \bb
\end{align*}
and
\begin{align*}
\hat{\lambda}_i:= \lambda_i(\hat{\bm{\Sigma}}_{\mathrm{sample}}) 
&\geq \lambda_i(\bm{\Sigma}_{\mathrm{sample}}) - \lambda_{\max}(\hat{\bm{\Sigma}}_{\mathrm{sample}} - \bm{\Sigma}_{\mathrm{sample}}) \\
&\geq \lambda_i(\frac{1}{1-b^2}\bm{\Sigma}_{(j)}) - \lambda_{\max}(\bm{\Sigma}_{\mathrm{sample}} -\frac{1}{1-b^2}\bm{\Sigma}_{(j)} ) - \bb \\
& \ge \lambda_i(\frac{1}{1-b^2}\bm{\Sigma}_{(j)}) - \bbb - \bb
\end{align*}
%
The above strategy to bound $ \lambda_i(\hat{\bm{\Sigma}}_{\mathrm{sample}})$ was suggested in \cite{tropp_personal_comm}.

Thus, using the the fact that $\frac{\lambda_{1}(\Lamjex) }{ \lambda_{k_1}(\Lamjex) } \le g^+=3$ and $\lambda_{k_1}(\Lamjex) \ge \lammin$, we have that
\[
\frac{\hat\lambda_1}{\hat\lambda_{k_1}}
\le \frac{\frac{1}{1-b^2}\lambda_{1}(\Lamjex) + \bbb  + \bb}{\frac{1}{1-b^2}\lambda_{k_1}(\Lamjex)  - \bbb -\bb}
\le \frac{g^+ + \frac{(\bbb + \bb)(1-b^2)}{\lambda_{k_1}(\Lamjex)}}{1 - \frac{(\bbb + \bb)(1-b^2)}{\lambda_{k_1}(\Lamjex)}}
\le \frac{g^+ + \frac{(\bbb + \bb)}{\lammin}}{1 - \frac{(\bbb + \bb)}{\lammin}}
\le \frac{3 + 0.06}{1-0.06} = \ghatp.
\]
Similarly, using the lower bound $\frac{\lambda_{1}(\Lambda_{(j)})}{\lambda_{k_1+1}(\Lambda_{(j)})} \ge \frac{1}{\chi^+}=5$ from Model \ref{clust_model}, 
\[
\frac{\hat\lambda_1}{\hat\lambda_{k_1+1}}
\ge \frac{\frac{1}{1-b^2}\lambda_{1}(\Lamjex) - \bbb  - \bb}{\frac{1}{1-b^2}\lambda_{k_1+1}(\Lamjex)  + \bbb  + \bb}
\ge \frac{\frac{\lambda_{1}(\Lamjex)}{\lambda_{k_1+1}(\Lamjex)} - \frac{(\bbb + \bb)(1-b^2)}{\lambda_{k_1+1}(\Lamjex)}}{1  + \frac{(\bbb + \bb)(1-b^2)}{\lambda_{k_1+1}(\Lamjex)}}
\ge \frac{\frac{\lambda_{k_1}(\Lamjex)}{\lambda_{k_1+1}(\Lamjex)} - \frac{(\bbb + \bb)}{\lammin}}{1  + \frac{(\bbb + \bb)}{\lammin}}
\ge \frac{\frac{1}{\chi^+} - \frac{(\bbb + \bb)}{\lammin}}{1  + \frac{(\bbb + \bb)}{\lammin}}
\ge \frac{5 - 0.06}{1+0.06} = 4.67  > \ghatp
\]
This shows that the first cluster is correctly recovered. 
Proceeding in the same manner, 
\[
\frac{\hat\lambda_{k_{i-1}+1}}{\hat\lambda_{k_{i}}}
\le \frac{g^+ + \frac{(\bbb + \bb)}{\lammin}}{1 - \frac{(\bbb + \bb)}{\lammin}}
\le \frac{3 + 0.06}{1-0.06} = \ghatp.
\]
and
\[
\frac{\hat\lambda_{k_{i-1}+1}}{\hat\lambda_{k_{i}+1}}
\ge \frac{\frac{1}{\chi^+} - \frac{(\bbb + \bb)}{\lammin}}{1  + \frac{(\bbb + \bb)}{\lammin}}
= \frac{5 - 0.06}{1+0.06} = 4.67  > \ghatp.
\]
Recall that the clustering algorithm excludes all eigenvalues below $0.25\lamtrain$.
Recall also that $\bm{\Sigma}_{(j)}$ has rank rank $r_{j} = k_{\vartheta_j}$. Thus from the upper and lower bounds on $\hat\lambda_i$ given above and using Lemma \ref{initsub_cor}, we can also conclude that,
\[
\hat{\lambda}_{k_{\vartheta_j}} \geq \lambda_{k_{\vartheta_j}} - \bbb - \bb \geq \lammin - 0.06\lammin > 0.75 \lambda^- > 0.25 \lamtrain
\]
and
\[
\hat{\lambda}_{k_{\vartheta_j}+1} \leq \lambda_{k_{\vartheta_j}+1} + \bbb + \bb \leq 0 + 0.06\lammin < 0.25\lamtrain
\]
Thus, the algorithm also stops at the correct place. 
We have shown that all of the clusters will be recovered exactly and no extra clusters will be formed (algorithm stops at the correct place). Thus, 
\[
\Pr\left(  \mathrm{CLUSTER}_{j}^a \ \big| \ \Gamma_{j,K}^a\right) \geq p_{\mathrm{cluster}} := 1 - p_{\mathrm{cl}} - p_{\tilde{\bm{le}}} - p_{\tilde{\bm{ee}}}
\]
for $a=u_{j}$ or $a=u_{j}+1$.  This proves the lemma.
\end{proof}

\begin{proof}[Proof of Lemma \ref{cPCA}]
Since we condition on the event $\Gamma_{j,k-1}^{\uhat_j}$ and $\Gamma_{j,k-1}^{\uhat_j} \subseteq \mathrm{CLUSTER}_j^{\uhat_j}$, the clusters are correctly recovered, i.e. $\hat{\mathcal{G}}_{j,k} = \mathcal{G}_{j,k}$.
This lemma then follows by combining Lemma \ref{zetakbnd} with the bounds from Lemmas \ref{Atil_cor}, \ref{Atil_cor}, \ref{Htil_cor} and finally using Lemma \ref{event}.
\end{proof}

\section{Proof of the Addition Azuma Lemmas} \label{3_pfs}

\subsection{A general decomposition used in all the proofs} \label{general_decomp}
\newcommand{\bM}{\bm{M}}
\newcommand{\bN}{\bm{N}}


A general decomposition will be developed here. We will use this in all the proofs that follow. Consider an interval $\J_u$ and let $t_0$ denote the first time instant of this interval.  Let $X \equiv X_{(t_0-1)/\alpha} = \{\bnu_0, \bnu_1, \dots \bnu_{t_0-1}, \{\T_t\}_{t=1,2,\dots t_{\max}}\}$. Let $\bM_t$ and $\bN_t$ be matrices that are deterministic given $X$. Consider bounding
\[
\frac{1}{\alpha}\sum_{t=t_0}^{t_0 + \alpha-1} \bN_t \lt {\lt}' \bM_t
\]
conditioned on $X$ for $X \in \Gamma$. 


From our model, notice that
\[
\lt = b^{t-t_0+1} \l_{t_0-1} + \sum_{\tau=t_0}^{t} b^{t-\tau} \bnu_\tau
\]
Thus,
\begin{align*}
\frac{1}{\alpha}\sum_{t=t_0}^{t_0 + \alpha-1} \bN_t \lt {\lt}' \bM_t
&= \frac{1}{\alpha} \sum_{t=t_0}^{t_0 + \alpha-1}   \bN_t \left( b^{t-t_0+1} \l_{t_0-1} + \sum_{\tau=t_0}^{t} b^{t-\tau} \bnu_\tau \right)  \left(  b^{t-t_0+1} \l_{t_0-1} +  \sum_{\ttau=t_0}^{t} b^{t-\ttau} \bnu_\ttau \right)' \bM_t  \\
&:=  \mathrm{term1} + \mathrm{term2} + \mathrm{term3}
\end{align*}
where
\begin{align*}
\mathrm{term1} &= \frac{1}{\alpha} \sum_{t=t_0}^{t_0+\alpha-1}  b^{2(t-t_0)+2} \bN_t ( \l_{t_0-1} \l_{t_0-1}{}' ) \bM_t \\
\mathrm{term3} &= \frac{1}{\alpha} \sum_{t=t_0}^{t_0+\alpha-1}  \sum_{\tau=t_0}^{t} b^{2t-t_0-\tau+1} \bN_t( \bnu_\tau \l_{t_0-1}{}' + \l_{t_0-1} \bnu_\tau' )\bM_t \\
\mathrm{term2} &= \frac{1}{\alpha} \sum_{t=t_0}^{t_0 + \alpha-1}  \bN_t \left( \sum_{\tau=t_0}^{t} b^{t-\tau} \bnu_\tau \right) \left( \sum_{\ttau=t_0}^{t} b^{t-\ttau} \bnu_\ttau \right)' \bM_t    \\
&:= \mathrm{term21} + \mathrm{term22} + \mathrm{term23} \ \text{where}  \\
\mathrm{term21} &= \frac{1}{\alpha} \sum_{t=t_0}^{t_0+\alpha-1}  \sum_{\tau=t_0}^{t}  b^{2t-2\tau} \bN_t (\bnu_\tau \bnu_\tau') \bM_t \\%
\mathrm{term22} &= \frac{1}{\alpha} \sum_{t=t_0}^{t_0+\alpha-1}  \sum_{\tau=t_0}^{t} \sum_{\ttau=t_0}^{\tau-1}  b^{2t-\tau-\ttau} \bN_t (\bnu_\tau \bnu_\ttau') \bM_t \\
\mathrm{term23} &= \frac{1}{\alpha} \sum_{t=t_0}^{t_0+\alpha-1}  \sum_{\tau=t_0}^{t} \sum_{\ttau=\tau+1}^{t}  b^{2t-\tau-\ttau} \bN_t (\bnu_\tau \bnu_\ttau' ) \bM_t
\end{align*}

We will show that $\mathrm{term22},\mathrm{term23},\mathrm{term3}$ are close to zero whp, and that $\mathrm{term1}$ can be bounded by a very small value (proportional to $1/\alpha$). The only non-trivial term is $\mathrm{term21}$ and we will show how to (i) bound its spectral norm whp and, (ii) when $\bN_t = \bM_t'$ (so that this term is symmetric), how to also bound its minimum eigenvalue whp. For all terms, except $\mathrm{term1}$ (which is a constant when conditioning on $X$), we will use the matrix Azuma inequalities (given in Appendix \ref{prelim}).

We first show how to bound the near-zero terms. Consider $\mathrm{term22}$. By Lemma \ref{sumswitch} (exchange order of double sum),
\begin{align*}
\mathrm{term22}
&= \frac{1}{\alpha} \sum_{\tau=t_0}^{t_0+\alpha-1} \sum_{t=\tau}^{t_0+\alpha-1} \sum_{\ttau=t_0}^{\tau-1}  b^{2t-\tau-\ttau} \bN_t (\bnu_\tau \bnu_\ttau') \bM_t \\
&:= \frac{1}{\alpha} \sum_{\tau=t_0}^{t_0+\alpha-1} \bm{Z}_\tau
\end{align*}
To apply matrix Azuma (Lemma \ref{azuma_norm}), we need to bound $\| \frac{1}{\alpha} \sum_{\tau=t_0}^{t_0+\alpha-1} \E[\bm{Z}_\tau| \bm{Z}_{t_0},\bm{Z}_{t_0+1}, \dots, \bm{Z}_{\tau-1},X] \|_2$ and $\|\bm{Z}_\tau\|$ conditioned on  $X$. Now,
\[
\E[\bm{Z}_\tau| \bm{Z}_{t_0},\bm{Z}_{t_0+1}, \dots, \bm{Z}_{\tau-1},X] 
= \sum_{t=\tau}^{t_0+\alpha-1} \sum_{\ttau=t_0}^{\tau-1}  b^{2t-\tau-\ttau} \bN_t \E[\bnu_\tau \bnu_\ttau'| \bm{Z}_{t_0},\bm{Z}_{t_0+1}, \dots, \bm{Z}_{\tau-1},X] \bM_t
\]
Consider $\E[ \bnu_\tau \bnu_\ttau' | \bm{Z}_{t_0},\bm{Z}_{t_0+1}, \dots, \bm{Z}_{\tau-1},X]$. Notice that here $\ttau \le \tau-1$. Thus, this is a case a of $\E[W Y | Z]$ where $W$ is independent of $\{Y,Z\}$ with $W \equiv \bnu_\tau$, $Y\equiv  \bnu_\ttau$ and $Z \equiv \{\bm{Z}_{t_0},\bm{Z}_{t_0+1}, \dots, \bm{Z}_{\tau-1},X\}$. This is true because $\bm{Z}_\tau$ is a function of  $\bnu_{t_0},\bnu_{t_0+1},  \dots, \bnu_\tau$ and thus $\{\bm{Z}_{t_0},\bm{Z}_{t_0+1}, \dots, \bm{Z}_{\tau-1},X\} = f(\bnu_0, \bnu_1, \dots \bnu_{\tau-1},\{\T_{\tilde{t}}\}_{\tilde{t}=1,2,\dots,t_{\max}})$. So $\{Y,Z\} = g(\bnu_0, \bnu_1, \dots \bnu_{\tau-1},\{\T_{\tilde{t}}\}_{\tilde{t}=1,2,\dots,t_{\max}})$ and this is independent of $\bnu_\tau$ (by the independence assumption from the theorem). Thus, by Lemma \ref{ind_exp}, since $\bnu_\tau$ is zero mean,
\[
\E[ \bnu_\tau \bnu_\ttau'| \bm{Z}_{t_0},\bm{Z}_{t_0+1}, \dots, \bm{Z}_{\tau-1},X] = \E[\bnu_\tau] \E [ \bnu_\ttau'| \bm{Z}_{t_0},\bm{Z}_{t_0+1}, \dots, \bm{Z}_{\tau-1},X] = 0
\]
Also,
\[
\|\bm{Z}_\tau\|_2 \le  (\max_\tau \sum_{t=\tau}^{t_0+\alpha-1} \sum_{\ttau=t_0}^{\tau-1}  b^{2t-\tau-\ttau})  \max_{t,\tau,\ttau} \|\bN_t \bnu_\tau \bnu_\ttau' \bM_t \|_2
\le b_{prob,term22}:= \frac{b}{(1-b^2)(1-b)} \max_{t,\tau,\ttau} \|\bN_t \bnu_\tau \bnu_\ttau' \bM_t \|_2
\]
Thus by Azuma, conditioned on $X$, $\|\mathrm{term22}\|_2 \le \epsilon$ w.p. at least $1-(2n) \exp\left(\frac{-\alpha \epsilon^2}{32 (b_{prob,term22})^2} \right)$.

Consider $\mathrm{term23}$. By Lemma \ref{sumswitch} (exchange order of double sum),
\begin{align*}
\mathrm{term23}
&= \frac{1}{\alpha} \sum_{\tau=t_0}^{t_0+\alpha-1} \sum_{t=\tau}^{t_0+\alpha-1}  \sum_{\ttau=\tau+1}^{t}  b^{2t-\tau-\ttau} \bN_t (\bnu_\tau \bnu_\ttau') \bM_t
\end{align*}
This term is not in a form where we can apply the matrix Azuma inequalities to get a useful bound. But we can get it into a nicer form by a simple change of variables. Let $p=(t_0+\alpha-1)-\tau$ and use this to replace $\tau$.
Then,  
\begin{align*}
\mathrm{term23}
& = \frac{1}{\alpha} \sum_{p=0}^{\alpha-1} \ \sum_{t=t_0+\alpha-1-p}^{t_0+\alpha-1} \ \sum_{\ttau=t_0+\alpha-p}^{t}  b^{2t-t0-\alpha+p-\ttau+1} \bN_t (\bnu_{t_0+\alpha-1-p} \bnu_\ttau' ) \bM_t\\  
&:= \frac{1}{\alpha} \sum_{p=0}^{\alpha-1} \ \bm{Z}_p
\end{align*}
To apply Azuma (Lemma \ref{azuma_norm}), we need to bound $\| \frac{1}{\alpha} \sum_{p=0}^{\alpha-1} \E[\bm{Z}_p| \bm{Z}_{0},\bm{Z}_{1}, \dots, \bm{Z}_{p-1},X] \|_2$ and $\|\bm{Z}_p\|$ conditioned on  $X$. Now,
\[
\E[\bm{Z}_p| \bm{Z}_{0},\bm{Z}_{1}, \dots, \bm{Z}_{p-1},X] = \sum_{t=t_0+\alpha-1-p}^{t_0+\alpha-1}  \sum_{\ttau=t_0+\alpha-p}^{t}  b^{2t-t_0-\alpha+1-\ttau} \bN_t \E[\bnu_{t_0+\alpha-1-p} \bnu_\ttau'| \bm{Z}_{0},\bm{Z}_{1}, \dots, \bm{Z}_{p-1},X] \bM_t
\]
Notice that $\bm{Z}_{p}$ is a function of $\bnu_{t_0+\alpha-p-1}, \bnu_{t_0+\alpha-p}, \dots, \bnu_{t_0+\alpha-1}$. Also recall that $X = \{\bnu_0, \bnu_1, \dots \bnu_{t_0-1}\}$.
Thus, $\{ \bm{Z}_0, \bm{Z}_1, \dots, \bm{Z}_{p-1},X \} = f(\bnu_0, \bnu_1, \dots \bnu_{t_0-1}, \ \bnu_{t_0+\alpha-p}, \bnu_{t_0+\alpha-p+1}, \dots \bnu_{t_0+\alpha-1})$.
%
%
Notice also that $\ttau \ge t_0+ \alpha-p$. Thus, the expectation above is again a case of $\E[W Y | Z]$ where $W$ is independent of $\{Y,Z\}$ with $W = \bnu_{t_0+\alpha-p-1}$, $Y = \bnu_\ttau$ (for a $\ttau \ge t_0+\alpha-p$) and $Z = \{\bm{Z}_0, \bm{Z}_1, \dots, \bm{Z}_{p-1},X\} = f(\bnu_0, \bnu_1, \dots \bnu_{t_0-1}, \ \bnu_{t_0+\alpha-p}, \bnu_{t_0+\alpha-p+1}, \dots, \bnu_{t_0+\alpha-1})$. Using, this, by Lemma \ref{ind_exp}, $\E[\bnu_{t_0+\alpha-1-p} \bnu_\ttau'| \bm{Z}_{0},\bm{Z}_{1}, \dots, \bm{Z}_{p-1},X] = 0$.
%
%
%
Also,
\[
\|\bm{Z}_p\| \le  (\max_p \sum_{t=t_0+\alpha-1-p}^{t_0+\alpha-1} \sum_{\ttau=t_0+\alpha-p}^{t}  b^{2t-t_0-\alpha+1-\ttau})  \max_{t,p,\ttau} \|\bN_t \bnu_{t_0+\alpha-1-p} \bnu_\ttau' \bM_t \|_2
\le b_{prob,term23}:= \frac{1}{(1-b)^2} \max_{t,\tau,\ttau} \|\bN_t \bnu_\tau \bnu_\ttau' \bM_t \|_2
\]
Thus by Azuma, conditioned on $X$, $\|\mathrm{term23}\|_2 \le \epsilon$ w.p. at least $1-(2n) \exp\left(\frac{-\alpha \epsilon^2}{32 (b_{prob,term23})^2} \right)$.

Consider $\mathrm{term3}$. By Lemma \ref{sumswitch} (exchange order of double sum),
\begin{align*}
\mathrm{term3} &= \frac{1}{\alpha} \sum_{\tau=t_0}^{t_0+\alpha-1} \sum_{t=\tau}^{t_0+\alpha-1} b^{2t-t_0-\tau+1} \bN_t  ( \bnu_\tau \l_{t_0-1}{}' + \l_{t_0-1} \bnu_\tau' )  \bM_t \\
& := \frac{1}{\alpha} \sum_{\tau=t_0}^{t_0+\alpha-1} \bm{Z}_\tau
\end{align*}
To apply Azuma (Lemma \ref{azuma_norm}), we need to bound $\|\frac{1}{\alpha} \sum_{\tau=t_0}^{t_0+\alpha-1} \E[\bm{Z}_\tau| \bm{Z}_{t_0}, \bm{Z}_{t_0+1}, \dots, \bm{Z}_{\tau-1}, X]\|_2$ and $\|\bm{Z}_\tau\|_2$ conditioned on  $X$. We can show that
\begin{align}
\E[\bm{Z}_\tau| \bm{Z}_{t_0}, \bm{Z}_{t_0+1}, \dots, \bm{Z}_{\tau-1}, X]
&= \sum_{t=\tau}^{t_0+\alpha-1} b^{2t-t_0-\tau+1} \bN_t \E[(\bnu_\tau \l_{t_0-1}{}' + \l_{t_0-1} \bnu_\tau') | \bm{Z}_{t_0}, \bm{Z}_{t_0+1}, \dots, \bm{Z}_{\tau-1}, X] \bM_t = 0.
\nonumber
\end{align}
This follows because $\bm{Z}_\tau = f(\bnu_\tau,X)$ and thus, $\{\bm{Z}_{t_0}, \bm{Z}_{t_0+1}, \dots, \bm{Z}_{\tau-1}, X\} = \tilde{f}(\bnu_0, \bnu_1, \dots, \bnu_{\tau-1})$. Also,
$\l_{t_0-1} = g(X)=\tilde{g}(\bnu_0, \bnu_1, \dots, \bnu_{t_0-1})$. Thus, this is again a case of $\E[W Y | Z]$ with $W =  \bnu_\tau$, $Y = \l_{t_0-1} = \tilde{g}(\bnu_0, \bnu_1, \dots, \bnu_{t_0-1})$ and $Z = \{\bm{Z}_{t_0}, \bm{Z}_{t_0+1}, \dots, \bm{Z}_{\tau-1}, X\} = \tilde{f}(\bnu_0, \bnu_1, \dots, \bnu_{\tau-1})$. 

Also,
\[
\|\bm{Z}_\tau\| 
\le b_{prob,term3}:= \frac{1}{(1-b^2)} \max_{t,\tau} (\|\bN_t \bnu_\tau \l_{t_0-1}{}' \bM_t \|_2 + \|\bN_t \l_{t_0-1}{} \bnu_\tau' \bM_t \|_2)
\]
Thus by Azuma, conditioned on $X$, $\|\mathrm{term3}\|_2 \le \epsilon$ w.p. at least $1-(2n) \exp\left(\frac{-\alpha \epsilon^2}{32 (b_{prob,term3})^2} \right)$.

Consider $\mathrm{term1}$. Since $\l_{t_0-1} = f(X)$ and everything else in this term is also a function of $X$, this term is a constant given $X$. Thus we can bound it directly. We have
\begin{align}
\|\mathrm{term1}\|_2 & \le \frac{1}{\alpha} \sum_{t=t_0}^{t_0+\alpha-1}  b^{2(t-t_0)+2} \max_{t \in [t_0, t_0+\alpha-1]} \|\bN_t  ( \l_{t_0-1} \l_{t_0-1}{}' ) \bM_t \|_2  \\
& \le \frac{b^2}{\alpha(1-b^2)}  \max_{t \in [t_0, t_0+\alpha-1]} \|\bN_t  ( \l_{t_0-1} \l_{t_0-1}{}' ) \bM_t \|_2 \le  \frac{(r_\new \zeta)^2 b^2}{(1-b^2)}  \max_{t \in [t_0, t_0+\alpha-1]} \|\bN_t  ( \l_{t_0-1} \l_{t_0-1}{}' ) \bM_t \|_2:= b_{term1}
\label{term1_upper_bnd}
\end{align}

Consider $\mathrm{term21}$. By Lemma \ref{sumswitch} (exchange summation order),
\begin{align*}
\mathrm{term21}
&= \frac{1}{\alpha} \sum_{\tau=t_0}^{t_0+\alpha-1} \sum_{t=\tau}^{t_0+\alpha-1}  b^{2t-2\tau} \bN_t (\bnu_\tau \bnu_\tau') \bM_t \\
&:= \frac{1}{\alpha} \sum_{\tau=t_0}^{t_0+\alpha-1}  \bm{Z}_\tau
\end{align*}
To obtain an upper bound on its spectral norm using Azuma, we need to upper bound $\|\frac{1}{\alpha} \sum_{\tau=t_0}^{t_0+\alpha-1} \E[\bm{Z}_\tau| \bm{Z}_{t_0}, \bm{Z}_{t_0+1}, \dots, \bm{Z}_{\tau-1}, X] \|_2$ and $\|\bm{Z}_\tau\|_2$. To get a lower bound on its minimum eigenvalue we need to lower bound $\lambda_{\min}(\frac{1}{\alpha} \sum_{\tau=t_0}^{t_0+\alpha-1} \E[\bm{Z}_\tau| \bm{Z}_{t_0}, \bm{Z}_{t_0+1}, \dots, \bm{Z}_{\tau-1}, X])$ as well.
We have
\begin{align*}
\E[\bm{Z}_\tau| \bm{Z}_{t_0}, \bm{Z}_{t_0+1}, \dots, \bm{Z}_{\tau-1}, X]
&=  \sum_{t=\tau}^{t_0+\alpha-1}  b^{2t-2\tau} \bN_t \E[\bnu_\tau \bnu_\tau'| \bm{Z}_{t_0}, \bm{Z}_{t_0+1}, \dots, \bm{Z}_{\tau-1}, X]  \bM_t \\
&= \sum_{t=\tau}^{t_0+\alpha-1}  b^{2t-2\tau} \bN_t \bm{\Sigma}_\tau \bM_t
\end{align*}
The last row follows because we condition on a function of $\{\bnu_0, \bnu_1, \dots, \bnu_{\tau-1}\}$, $\bnu_\tau$ is independent of all these and $\E[\bnu_\tau \bnu_\tau'] = \bm{\Sigma}_\tau$. Then by applying Lemma \ref{sumswitch} in reverse order, we get
\begin{align*}
\frac{1}{\alpha} \sum_{\tau=t_0}^{t_0+\alpha-1} \E [\bm{Z}_\tau| \bm{Z}_{t_0}, \bm{Z}_{t_0+1}, \dots, \bm{Z}_{\tau-1}, X]
&:= \frac{1}{\alpha} \sum_{\tau=t_0}^{t_0+\alpha-1} \sum_{t=\tau}^{t_0+\alpha-1}  b^{2t-2\tau} \bN_t  \bm{\Sigma}_\tau \bM_t  \\
&=  \frac{1}{\alpha} \sum_{t=t_0}^{t_0+\alpha-1} \sum_{\tau=t_0}^{t}  b^{2t-2\tau} \bN_t \bm{\Sigma}_\tau \bM_t
\end{align*}
Also,
\[
\|\bm{Z}_\tau\|_2 \le  (\max_\tau \sum_{t=\tau}^{t_0+\alpha-1}  b^{2t-2\tau})  \max_{t,\tau} \|\bN_t \bnu_\tau \bnu_\tau' \bM_t \|_2
\le b_{prob,term21}:= \frac{1}{(1-b^2)} \max_{t,\tau} \|\bN_t \bnu_\tau \bnu_\tau' \bM_t \|_2
\]
Thus by Azuma (Lemma \ref{azuma_norm}), conditioned on $X$,
\begin{align}
\|\mathrm{term21}\|_2 \le \|\frac{1}{\alpha} \sum_{t=t_0}^{t_0+\alpha-1} \sum_{\tau=t_0}^{t}  b^{2t-2\tau} \bN_t \bm{\Sigma}_\tau \bM_t\|_2 + \epsilon
\label{term21_upper_bnd}
\end{align}
w.p. at least $1-(2n) \exp\left(\frac{-\alpha \epsilon^2}{32 (b_{prob,term21})^2} \right)$.

Let $b_{term21}$ denote the upper bound on the first term in the RHS of \eqref{term21_upper_bnd}. Then, conditioned on $X$,
\begin{align}
\|\frac{1}{\alpha}\sum_{t=t_0}^{t_0 + \alpha-1} \bN_t \lt {\lt}' \bM_t\|_2 \le b_{term1} + b_{term21} + 4\epsilon
\label{upper_bnd}
\end{align}
with probability obtained from a union bound.

Consider the special case when $\bN_t' = \bM_t$. In this case, $\frac{1}{\alpha}\sum_{t=t_0}^{t_0 + \alpha-1} \bN_t \lt {\lt}' \bM_t$ is a symmetric matrix and so is
$\mathrm{term21}$.
We can lower bound its minimum eigenvalue using Azuma Lemma \ref{azuma_hermitian} to get that, conditioned on $X$,
\begin{align}
\lambda_{\min}(\mathrm{term21}) \ge \lambda_{\min}(\frac{1}{\alpha} \sum_{t=t_0}^{t_0+\alpha-1} \sum_{\tau=t_0}^{t}  b^{2t-2\tau} \bN_t \bm{\Sigma}_\tau \bN_t') - \epsilon
\label{term21_lower_bnd}
\end{align}
w.p. at least $1-(2n) \exp\left(\frac{-\alpha \epsilon^2}{32 (b_{prob,term21})^2} \right)$.
Let $b_{lower,term21}$ denote the lower bound on the first term in the RHS of \eqref{term21_lower_bnd}.
Then, conditioned on $X$, we can conclude that
\begin{align}
\lambda_{\min}(\frac{1}{\alpha}\sum_{t=t_0}^{t_0 + \alpha-1} \bN_t \lt {\lt}' \bM_t) \ge b_{lower,term21} - 4 \epsilon
\label{lower_bnd}
\end{align}
with probability obtained from a union bound. We get the above because of the following reason. Since $\mathrm{term1}$ is symmetric positive semi-definite, $\lambda_{\min}(\mathrm{term1}) \ge 0$.
Since $\mathrm{term3}$ is symmetric, $\lambda_{\min}(\mathrm{term3}) \ge - \|\mathrm{term3}\| \ge -\epsilon$. Since $\mathrm{term2}$ is also a symmetric matrix in this case, it follows that $\mathrm{term22}+\mathrm{term23} = \mathrm{term2} - \mathrm{term21}$ is a symmetric matrix. Thus $\lambda_{\min}(\mathrm{term22}+\mathrm{term23}) \ge - \|\mathrm{term22}+\mathrm{term23}\| \ge -\|\mathrm{term22}\| - \|\mathrm{term23}\| \ge - 2\epsilon$.

In the special case when $\bN_t' = \bM_t= \bM_0$, using Lemma \ref{sum_often}, the RHS in \eqref{term21_lower_bnd} can be lower bounded by
$\frac{1}{1-b^2}(1 - \frac{b^2}{\alpha (1-b^2)}) \min_{\tau \in [t_0, t_0+\alpha-1]} \lambda_{\min} (\bM_0'  \bm{\Sigma}_\tau \bM_0 ) - \epsilon$.


In the special case when $\bN_t = \bN_0$ and $\bM_t = \bM_0$, the RHS in \eqref{term21_upper_bnd} can be upper bounded by $\frac{1}{1-b^2} \max_{\tau \in [t_0, t_0+\alpha-1]} \| \bN_0 \bm{\Sigma}_\tau \bM_0\|_2 + \epsilon$.

In the special case when $\bN_t = \bm{\Phi}_0$ and $\bM_t = \bm{\Phi}_{k-1} \ITt \invterm \ITt'$, we can apply Cauchy-Schwartz for matrices  followed by Lemma \ref{blockdiag1} (support change lemma) to the RHS of \eqref{term21_upper_bnd} to get the final upper bound.

In the special case when $\bN_t' = \bM_t = \bm{\Phi}_{k-1} \ITt \invterm \ITt'$, we can directly apply Lemma \ref{blockdiag1} (support change lemma) to the RHS of \eqref{term21_upper_bnd} to get the upper bound.

\subsection{A general decomposition for terms containing $w_t$ }  \label{general_decomp_w}
Consider bounding $\frac{1}{\alpha}\sum_{t=t_0}^{t_0 + \alpha-1} \bN_t \lt {\wt}' \bM_t$ conditioned on $X$. Here $X$ contains $\bnu_t$'s for all $t \le t_0-1$ and contains all the $\T_t$'s.
Using the independence assumption from the theorem, 
\[
\E[\bm{Z}_t | \bm{Z}_{t-1}, \bm{Z}_{t-2}, \dots, \bm{Z}_{t_0}, X] = \E[\lt|\bm{Z}_{t-1}, \bm{Z}_{t-2}, \dots, \bm{Z}_{t_0}, X] \E[\wt'] = 0
\]
This follows by Lemma \ref{ind_exp} with $W \equiv \wt$, $Y \equiv \lt = g(\bnu_0, \bnu_1, \dots, \bnu_t)$ and $Z \equiv \{\bm{Z}_{t-1}, \bm{Z}_{t-2}, \dots, \bm{Z}_{t_0}, X\} = f(\bm{w}_{t_0},\bm{w}_{t_0+1}, \dots, \bm{w}_{t-1}, \bnu_0, \bnu_1, \dots, \bnu_{t-1}, \T_{\ttau},\ttau=1,2,\dots,t_{\max}\}$ and using the fact that $\wt$ is zero mean.
Also,
\[
\|\bN_t \lt {\wt}' \bM_t\|_2 \le b_{prob,\lt \wt}:= \max_t \|\bN_t \lt {\wt}' \bM_t\|_2
\]
Thus we can conclude by Azuma Lemma \ref{azuma_norm} that
\[
\|\frac{1}{\alpha}\sum_{t=t_0}^{t_0 + \alpha-1} \bN_t \lt {\wt}' \bM_t\|_2 \le \epsilon
\]
w.p. at least $1-(2n) \exp\left(\frac{-\alpha \epsilon^2}{32 (b_{prob, \lt \wt})^2} \right)$. 


\begin{fact} \label{wt_model_weaker_proof}
In situations where it is not practical to assume that $\wt$ is independent of $\T_t$, the assumption of Remark \ref{wt_model_weaker} can be used. With this, we can proceed as in Sec. \ref{general_decomp} above. There will be only two terms, $term1 = \frac{1}{\alpha}\sum_{t=t_0}^{t_0 + \alpha-1} \bN_t b^{t-t_0+1} \l_{t_0-1} {\wt}' \bM_t$ and $term2 = \frac{1}{\alpha}\sum_{t=t_0}^{t_0 + \alpha-1} \sum_{\tau=t_0}^{t} \bN_t b^{t-\tau} \bnu_\tau \wt'$. We can bound term1 as before by $\frac{(r_\new \zeta)^2 b}{1-b} \sqrt{r}\gamma \epsilon_w \|\bN_t\|_2 \|\bM_t\|_2$. Everywhere where we use this, $\|\bN_t\|_2 \|\bM_t\|_2 \le 1.2^2 =1.44$. With this and with using the bounds on $\epsilon_w$ and $\zeta$, this is smaller than $0.001 r_\new \zeta \lambda^- = \epsilon$.
By Lemma \ref{sumswitch}, $term2= \frac{1}{\alpha} \sum_{\tau=t_0}^{t_0 + \alpha-1} \sum_{t=\tau}^{t_0+\alpha-1} \bN_t b^{t-\tau} \bnu_\tau \wt' := \frac{1}{\alpha}\sum_{\tau=t_0}^{t_0 + \alpha-1} \bm{Z}_\tau$. Notice that $\{ \bm{Z}_{t_0},  \bm{Z}_{t_0+1}, \dots,  \bm{Z}_{\tau-1},X\} =f(\bnu_0,\bnu_1, \dots, \bnu_{\tau-1}, \bm{w}_{t_0}, \bm{w}_{t_0+1}, \dots, \bm{w}_{t_0+\alpha-1}, \T_{\ttau}, \ttau=1,2,\dots,t_{\max} )$ and so $\E[ \bm{Z}_\tau| \bm{Z}_{t_0},  \bm{Z}_{t_0+1}, \dots,  \bm{Z}_{\tau-1},X]$ is an example of of $\E[WY|Z]$ with $W$ independent of $\{Y,Z\}$ if we let $W = \bnu_\tau$, $Y = \wt$ and $Z = \{\bm{Z}_{t_0}, \bm{Z}_{t_0+1}, \dots, \bm{Z}_{\tau-1},X\}$. Hence it is equal to zero. Thus, using Azuma Lemma \ref{azuma_norm} we can bound $term2$ by $\epsilon$ whp.
With this, whenever $\|\bN_t\|_2 \|\bM_t\|_2 \le 1.2^2 =1.44$, $\|\frac{1}{\alpha}\sum_{t=t_0}^{t_0 + \alpha-1} \bN_t \lt {\wt}' \bM_t\|_2 < 2\epsilon$ (instead of $\epsilon$). 
\end{fact}

\subsection{Proofs of the Addition Azuma  Bounds: Lemmas \ref{Ak_cor}, \ref{Akperp_cor}, and \ref{calHk_cor}} \label{add_hoeffding}

We remove the subscript $j$ at various places in this and later sections. Thus, for example, $\bm\Phi_{(j),k-1}$ is replaced by $\bm\Phi_{k-1}$ for $k=1,2,\dots K$.

\begin{definition}
Let $X \equiv X_{k-1} \equiv X_{\uhat_j + k-1}$.
\end{definition}

\begin{fact}\label{matbnds}
Let $\bm{D}_{\new,k-1} := \bm{\Phi}_{k-1} P_\new$ and  $\bm{D}_{*,k-1} := \bm{\Phi}_{k-1} P_*$. Recall that $\bm{D}_\new =\bm{D}_{\new,0} = \bm{\Phi}_0 P_\new$.
When $X_{\uhat_j+k-1}\in \Gamma_{j,k-1}^a$ for $a=u_j$ or $a=u_j+1$,
\begin{enumerate}
\item $\|\bm{D}_{*,k-1}\|_2 \leq \zeta_{j,*}^+$ for $k = 1,\dots,K$ (this follows using Fact \ref{d_large}).
\item $\|\bm{D}_{\new,k-1}\|_2 \leq \zeta_{\new,k-1}^+$ for $k = 1,\dots,K+1$ (by definition of $\Gamma_{j,k-1}^{\uhat_j}$).
\item Recall that $\zeta_{\new,0}^+ = 1$.
\item $\lambda_{\min}(\bm{R}_{\new}{\bm{R}_{\new}}') \geq 1 -(\zeta_{*}^+)^2$ (this follows because $\|\Phat_*{}'\P_\new\|_2 = \|\Phat_*{}'(\I - \P_*{\P_*}')\P_\new\|_2  \le \zeta_*$)
\item $\bm{E}_\new{}' \bm{D}_\new = \bm{E}_\new{}' \bm{E}_\new \bm{R}_\new = \bm{R}_\new$ and $\bm{E}_{\new,\perp}{}' \bm{D}_\new = \bm{0}$.
\item $\|[ ({\bm{\Phi}_{t})_{\mathcal{T}_t}}'(\bm{\Phi}_{t})_{\mathcal{T}_t}]^{-1}\|_2 \leq \phi^+$ (using Lemma \ref{cslem_cor})
\item $\et$ satisfies (\ref{etdef0_cor}) with probability one (using Lemma \ref{cslem_cor}).
\end{enumerate}
\end{fact}

\begin{proof}[Proof of Lemma \ref{Ak_cor}]
In this proof all probabilistic statements are conditioned on $X_{\uhat_j+k-1}$ for $X_{\uhat_j+k-1} \in \Gamma_{j,k-1}^{\uhat_j}$ for $\uhat_j=u_{j}$ or $u_{j}+1$.
We need a lower bound on the minimum eigenvalue of $\bm{A}_u$ for $u = \uhat_j+k$ for $k=1,2,\dots, K$ and $\uhat_j = u_j$ or $u_j+1$.
For $u = \uhat_j+k$, recall that
\[
\bm{A}_u :=  \frac{1}{\alpha} \sum_{t \in \mathcal{J}_{u}} {\bm{E}_{\rmnew}}' \bm{\Phi}_0 \bm{\ell}_t {\bm{\ell}_t}' \bm{\Phi}_0 \bm{E}_{\rmnew}
\]
Let $t_0$ be the first time instant of $\J_{\uhat_j+k}$. We proceed as in Section \ref{general_decomp} with $\bN_t' = \bM_t = \bm{\Phi}_0 \bm{E}_{j,\rmnew}$. Thus,
\begin{align*}
b_{prob,term2} & := \max(b_{prob,term21},b_{prob,term22},b_{prob,term23}) \le \frac{1}{(1-b)^2} (r\zeta \sqrt{r} \gamma + \sqrt{r_\new} \gamma_\new)^2 \\
b_{prob,term3}
& \le \frac{1}{(1-b)^2} \frac{(2 r\zeta \sqrt{r} \gamma + \sqrt{r_\new} \gamma_\new)}{1-b}(r\zeta \sqrt{r} \gamma + \sqrt{r_\new} \gamma_\new)
\le \frac{1}{(1-b)^3} (2 r\zeta \sqrt{r} \gamma + \sqrt{r_\new} \gamma_\new)^2
\end{align*}

Use $b_{prob}$ to denote an upper bound on $\max(b_{prob,term2},b_{prob,term3})$. Then
\[
b_{prob} = \frac{1}{(1-b)^3} (2 r\zeta \sqrt{r} \gamma + \sqrt{r_\new} \gamma_\new)^2
\]
Using \eqref{lower_bnd}, \eqref{term21_lower_bnd} and Lemma \ref{sum_often},
\begin{align*}
\lambda_{\min}(\bm{A}_u)
& \ge \lambda_{\min}(  \frac{1}{\alpha} \sum_{t=t_0}^{t_0+\alpha-1} \sum_{\tau=t_0}^{t}  b^{2t-2\tau}  {\bm{E}_{\rmnew}}' \bm{\Phi}_0 \bm{\Sigma}_\tau \bm{\Phi}_0 \bm{E}_{\rmnew} ) - 4 \epsilon \\
& \ge \frac{1}{1-b^2}(1 - \frac{b^2}{\alpha (1-b^2)}) \min_{\tau \in [t_0, t_0+\alpha-1]} \lambda_{\min}( {\bm{E}_{\rmnew}}' \bm{\Phi}_0 \bm{\Sigma}_\tau \bm{\Phi}_0 \bm{E}_{\rmnew} ) - 4 \epsilon
\end{align*}
w.p. at least $1-  4 \cdot (2n) \exp\left(\frac{-\alpha \epsilon^2}{32 (b_{prob})^2} \right) $.
Using Fact \ref{matbnds}, and Ostrowski's theorem, we get
\[
\lambda_{\min}( {\bm{E}_{\rmnew}}' \bm{\Phi}_0 \bm{\Sigma}_\tau \bm{\Phi}_{0}  {\bm{E}_{\rmnew}} ) \ge \lambda_{\min}(\bm{R}_\new \bm{\Lambda}_{\tau,\new} \bm{R}_\new')  \ge  \lambda_{\min}(\bm{R}_\new \bm{R}_\new')  \lambda_{\min}(\bm{\Lambda}_{\tau,\new}) \ge (1 - (\zeta_*^+)^2) \lambda_\new^-
\]
Thus, using $1/\alpha \le (r_\new \zeta)^2$,
\[
\lambda_{\min}(\bm{A}_u) \ge (1 - \frac{b^2}{\alpha (1-b^2)}) (1 - (\zeta_*^+)^2) \frac{\lambda_\new^-}{1-b^2} - 4 \epsilon
\ge b_{\bm{A}}:= \frac{1}{1-b^2}\left( (1 - (\zeta_*^+)^2)\lambda_\new^- - (r_\new \zeta)^2 \frac{b^2}{1-b^2} (1-(\zeta_*^+)^2)\lambda_\new^- \right)  - 4\epsilon
\]
w.p. at least $1- p_{\bm{A}}$ with
$p_{\bm{A}}:=4 \cdot (2n) \exp\left(\frac{-\alpha \epsilon^2}{32 (b_{prob})^2} \right)$. 
\end{proof}

\begin{proof}[Proof of Lemma \ref{Akperp_cor}]
In this proof all probabilistic statements are conditioned on $X_{\uhat_j+k-1}$ for $X_{\uhat_j+k-1} \in \Gamma_{j,k-1}^{\uhat_j}$ for $\uhat_j=u_{j}$ or $u_{j}+1$.
We need to upper bound the maximum eigenvalue of
\[
\bm{A}_{u,\perp} := \frac{1}{\alpha} \sum_{t \in \mathcal{J}_{u}} {\bm{E}_{\rmnew,\perp}}' \bm{\Phi}_0 \bm{\ell}_t {\bm{\ell}_t}' \bm{\Phi}_0  \bm{E}_{\rmnew,\perp}.
\]
Let $t_0$ be the first time instant of $\J_{\uhat_j+k}$. We proceed as in Section \ref{general_decomp} with $\bN_t' = \bM_t = \bm{\Phi}_0  \bm{E}_{\rmnew,\perp}$.
Thus,
\begin{align*}
b_{prob,term2} & = \max(b_{prob,term21},b_{prob,term22},b_{prob,term23}) = \frac{1}{(1-b)^2} (r\zeta)^2 r\gamma^2 \\
b_{prob,term3} & \le \frac{1}{(1-b)^3} (2 r\zeta \sqrt{r} \gamma)^2.
\end{align*}
Use $b_{prob}$ to denote the upper bound on $\max(b_{prob,term2},b_{prob,term3})$. Then
\[
b_{prob} = \frac{1}{(1-b)^3} (2 r\zeta \sqrt{r} \gamma)^2
\]

Using \eqref{upper_bnd}, \eqref{term21_upper_bnd} and \eqref{term1_upper_bnd}
\[
b_{term1} =
\frac{(r_\new \zeta)^2 b^2}{(1-b^2)}  \max_{t \in [t_0, t_0+\alpha-1]} \lambda_{\max}( {\bm{E}_{\rmnew,\perp}}' \bm{\Phi}_0  ( \l_{t_0-1} \l_{t_0-1}{}' )\bm{\Phi}_0  \bm{E}_{\rmnew,\perp}  )
\le  \frac{(r_\new \zeta)^2 b^2}{(1-b^2)} \frac{(r \gamma^2)}{(1-b)^2} \le \frac{0.05(r_\new \zeta) b^2 \lambda^-}{(1-b^2)(1-b)^2}
\]
(we can get a tighter bound for the above, but do not need it and hence do not pursue it)
and
\begin{align*}
\lambda_{\max}(\bm{A}_{u,\perp}) & \le  \lambda_{\max}(\frac{1}{\alpha} \sum_{t=t_0}^{t_0+\alpha-1} \sum_{\tau=t_0}^{t}  b^{2t-2\tau}  {\bm{E}_{\rmnew,\perp}}' \bm{\Phi}_0 \bm{\Sigma}_\tau \bm{\Phi}_0  \bm{E}_{\rmnew,\perp}) + b_{term1} + 4 \epsilon \\
& \le \frac{1}{1-b^2} \max_{\tau \in [t_0, t_0+\alpha-1]} \lambda_{\max}({\bm{E}_{\rmnew,\perp}}' \bm{\Phi}_0 \bm{\Sigma}_\tau \bm{\Phi}_0  \bm{E}_{\rmnew,\perp})) + b_{term1} + 4 \epsilon
\end{align*}
w.p. at least $1- 4 \cdot (2n) \exp\left(\frac{-\alpha \epsilon^2}{32 (b_{prob})^2} \right)$

Using Fact \ref{matbnds}, $\lambda_{\max}({\bm{E}_{\rmnew,\perp}}' \bm{\Phi}_0 \bm{\Sigma}_\tau \bm{\Phi}_0  \bm{E}_{\rmnew,\perp})) \le (r\zeta)^2 \lambda^+ $. Thus,
\[
\lambda_{\max}(\bm{A}_{u,\perp}) \le \frac{1}{1-b^2} (r\zeta)^2 \lambda^+ +  \frac{0.05(r_\new \zeta) b^2 \lambda^-}{(1-b^2)(1-b)^2} + 4 \epsilon
\le  b_{\bm{A},\perp}
\]
w.p. at least $1- p_{\bm{A},\perp}$ with
$p_{\bm{A},\perp}:= 4 \cdot (2n) \exp\left(\frac{-\alpha \epsilon^2}{32 (b_{prob})^2} \right)$
\end{proof}

\begin{proof}[Proof of Lemma \ref{calHk_cor}]
In this proof all probabilistic statements are conditioned on $X_{\uhat_j+k-1}$ for $X_{\uhat_j+k-1} \in \Gamma_{j,k-1}^{\uhat_j}$ for $\uhat_j=u_{j}$ or $u_{j}+1$.
Using the expression for $\bm{\mathcal{H}}_u$ given in Definition \ref{defHk}, and noting that for a basis matrix $\bm{E}$, $\bm{EE}' + \bm{E}_{\perp}{\bm{E}_{\perp}}' = \I$ we get that
\[
\bm{\mathcal{H}}_u = \frac{1}{\alpha} \sum_{t\in\J_u} \Big( \bm{\Phi}_0 \bm{e}_t {\bm{e}_t}' \bm{\Phi}_0 - (\bm{\Phi}_0\lt {\bm{e}_t}' \bm{\Phi}_0 + \bm{\Phi}_0 \bm{e}_t \lt{}'\bm{\Phi}_0)  + (\bm{F}_t + {\bm{F}_t}') \Big)
\]
where
\[
\bm{F}_t = \bm{E}_{\rmnew,\perp}{\bm{E}_{\rmnew,\perp}}' \bm{\Phi}_0 \lt  \lt' \bm{\Phi}_0 \bm{E}_{\rmnew}{\bm{E}_{\rmnew}}'.
\]
Thus,
\begin{equation}\label{add_calH1_cor}
\| \bm{\mathcal{H}}_u \|_2 \leq  2\bigg\| \frac{1}{\alpha} \sum_t  \bm{\Phi}_0 \lt {\bm{e}_t}'  \bigg\|_2 +  \bigg\| \frac{1}{\alpha} \sum_t \bm{e}_t {\bm{e}_t}' \bigg\|_2 +  2\bigg\| \frac{1}{\alpha} \sum_t \bm{F}_t \bigg\|_2
\end{equation}
Next we obtain high probability bounds on each of the three terms on the right hand side of (\ref{add_calH1_cor}) using the Azuma corollaries.

{\bf The $\lt \et'$ term. }
Consider the first term.
Using Fact \ref{matbnds} and the expression for $\et$ from \eqref{etdef0_cor},
\begin{align*}
\frac{1}{\alpha} \sum_t  \bm{\Phi}_0 \lt {\bm{e}_t}'
& = \frac{1}{\alpha} \sum_t  \bm{\Phi}_0 \lt (\lt+ \wt)' \bm{\Phi}_{k-1} \ITt \invterm \ITt' - \frac{1}{\alpha} \sum_t  \bm{\Phi}_0 \lt \wt' \\
& := \mathrm{term} + \mathrm{termw}, \ \ \text{where} \\
\mathrm{term} & := \frac{1}{\alpha} \sum_t  \bm{\Phi}_0 \lt \lt' \bm{\Phi}_{k-1} \ITt \invterm \ITt',  \\
\mathrm{termw} & := \frac{1}{\alpha} \sum_t  \bm{\Phi}_0 \lt \wt' \bm{\Phi}_{k-1} \ITt \invterm \ITt'  - \frac{1}{\alpha} \sum_t  \bm{\Phi}_0 \lt \wt'
\end{align*}
Here we use $\mathrm{termw}$ to refer to the sum of all terms containing $\wt$. 

By following the approach of Section \ref{general_decomp_w}, under the given conditioning,
\[
\|\mathrm{termw}\|_2 \le 2\epsilon
\]
w.p. at least $1-2\cdot(2n) \exp\left(\frac{-\alpha \epsilon^2}{32 (b_{prob, termw})^2} \right)$ where
\[
 b_{prob, termw} = (\phi^+)\frac{(2 r\gamma \sqrt{r}\gamma + \sqrt{r_\new} \gamma_\new) \epsilon_w}{1-b}
\]
We proceed as in Section \ref{general_decomp} for $\mathrm{term}$. In this case, $\bN_t = \bm{\Phi}_0$ and $\bM_t = \bm{\Phi}_{k-1} \ITt \invterm \ITt'$.
Thus
\begin{align*}
b_{prob,term2} &= \max(b_{prob,term21},b_{prob,term22},b_{prob,term23}) \le \frac{1}{(1-b)^2} \phi^+(\zeta_*^+ \sqrt{r}\gamma + \sqrt{r_\new} \gamma_\new)^2 \\
b_{prob,term3} & \le \frac{1}{(1-b)^3} \phi^+ (2 r\zeta \sqrt{r} \gamma + \sqrt{r_\new} \gamma_\new)^2
\end{align*}
Use $b_{prob}$ to denote the upper bound on $\max(b_{prob,term2},b_{prob,term3})$. Then
\[
b_{prob} = \frac{1}{(1-b)^3} \phi^+ (2 r\zeta \sqrt{r} \gamma + \sqrt{r_\new} \gamma_\new)^2
\]
Using \eqref{term1_upper_bnd}, \eqref{upper_bnd} and \eqref{term21_upper_bnd},
\[
b_{term1} = \frac{(r_\new \zeta)^2 b^2}{(1-b^2)} \max_t \|\bm{\Phi}_0 \l_{t_0-1} \l_{t_0-1}' \bm{\Phi}_{k-1} \ITt \invterm \ITt'\|_2
\le  \frac{(r_\new \zeta)^2 b^2}{(1-b^2)} \frac{(r \gamma^2)}{(1-b)^2} \le \frac{0.05(r_\new \zeta) b^2 \lambda^-}{(1-b^2)(1-b)^2}
\]
(we can get a tighter bound for the above, but do not need it and hence do not pursue it)
and
\[
\|\mathrm{term}\|_2
\le \|\frac{1}{\alpha} \sum_{t=t_0}^{t_0+\alpha-1} \sum_{\tau=t_0}^{t}  b^{2t-2\tau}  \bm{\Phi}_0  \bm{\Sigma}_\tau \bm{\Phi}_{k-1} \ITt \invterm \ITt' \|_2 +
b_{term1} + 4 \epsilon
\]
w.p. at least $1 -4 \cdot (2n) \exp\left(\frac{-\alpha \epsilon^2}{32 (b_{prob})^2} \right)$. 

First consider the $k=1$ case. In this case, $\bm\Phi_{k-1} = \bm\Phi_0$. By Lemma \ref{Dnew0_lem}, under the given conditioning, $\|\P_\new{}'\bm\Phi_0 \I_{\T_t}\|_2 = \|\I_{\T_t}{}' \bm\Phi_0 \P_\new\|_2 \le \kappa_{s,\new}^+ = 0.0215$. Using this and Fact \ref{matbnds},
\[
\|\frac{1}{\alpha} \sum_{t=t_0}^{t_0+\alpha-1} \sum_{\tau=t_0}^{t}  b^{2t-2\tau}  \bm{\Phi}_0  \bm{\Sigma}_\tau \bm{\Phi}_{k-1} \ITt \invterm \ITt' \|_2
\le \frac{1}{1-b^2} \phi^+ ((\zeta_*^+)^2 \lambda^+  +  \kappa_{s,\new}^+ \lambda_\new^+ )
\]
and so for $k=1$,
\[
\|\frac{1}{\alpha} \sum_t  \bm{\Phi}_0 \lt {\bm{e}_t}' \|_2 \le \frac{1}{1-b^2}  \phi^+((\zeta_*^+)^2 \lambda^+  +  \kappa_{s,\new}^+ \lambda_\new^+ ) + \frac{0.05(r_\new \zeta) b^2 \lambda^-}{(1-b^2)(1-b)^2} + 6 \epsilon
\]
w.p. at least $1 - p_{\l \e}$ with $p_{\l \e}: = 4 \cdot (2n) \exp\left(\frac{-\alpha \epsilon^2}{32 (b_{prob})^2} \right) + 2 \cdot (2n) \exp\left(\frac{-\alpha \epsilon^2}{32 (b_{prob, termw})^2} \right)$.

For $k>1$, we cannot use Lemma \ref{Dnew0_lem}. Thus, we follow a different approach - we use Lemma \ref{CSmat} (Cauchy-Schwartz for sums of matrices) followed by Lemma \ref{blockdiag1} (support change lemma).
Let $\bm{X}_t:=\sum_{\tau=t_0}^{t}  b^{2t-2\tau} \bm{\Phi}_0 \bm{\Sigma}_\tau \bm{\Phi}_{k-1}$ and $\bm{Y}_t:=\ITt \invterm \ITt'$. Then by Lemma \ref{CSmat} (Cauchy-Schwartz),
\[
\|\frac{1}{\alpha} \sum_{t=t_0}^{t_0+\alpha-1} \sum_{\tau=t_0}^{t}  b^{2t-2\tau} \bm{\Phi}_0 \bm{\Sigma}_\tau \bm{\Phi}_{k-1} \ITt \invterm \ITt'\|_2
\le \sqrt{\lambda_{\max}( \frac{1}{\alpha} \sum_{t=t_0}^{t_0+\alpha-1} \bm{X}_t \bm{X}_t')  \lambda_{\max}( \frac{1}{\alpha} \sum_{t=t_0}^{t_0+\alpha-1} \bm{Y}_t \bm{Y}_t' )}
\] 
Now,
\begin{align*}
\lambda_{\max}( \frac{1}{\alpha} \sum_{t=t_0}^{t_0+\alpha-1} \bm{X}_t \bm{X}_t')
& \le \max_t \|\bm{X}_t\|^2
 \le \left( \sum_{\tau=t_0}^{t}  b^{2t-2\tau} \max_\tau \|\bm{\Phi}_0 \bm{\Sigma}_\tau \bm{\Phi}_{k-1}\|_2  \right)^2 \le \left( \frac{1}{1-b^2} ((\zeta_*^+)^2 \lambda^+ + \zeta_{\new,k-1}^+ \lambda_\new^+ ) \right)^2
\end{align*}
By Lemma \ref{blockdiag1} (support change lemma)
\[
\lambda_{\max}( \frac{1}{\alpha} \sum_{t=t_0}^{t_0+\alpha-1} \bm{Y}_t \bm{Y}_t' ) \le \rho^2 h^+ (\phi^+)^2
\]
Thus,
\[
\|\frac{1}{\alpha} \sum_{t=t_0}^{t_0+\alpha-1} \sum_{\tau=t_0}^{t}  b^{2t-2\tau} \bm{\Phi}_0 \bm{\Sigma}_\tau \bm{\Phi}_{k-1} \ITt \invterm \ITt'\|_2
\le \sqrt{\rho^2 h^+ (\phi^+)^2} \left( \frac{1}{1-b^2} ( (\zeta_*^+)^2 \lambda^+ + \zeta_{\new,k-1}^+ \lambda_\new^+ ) \right)
\]
and so for $k>1$,
\[
\|\frac{1}{\alpha} \sum_t  \bm{\Phi}_0 \lt {\bm{e}_t}' \|_2
\le \sqrt{\rho^2 h^+ (\phi^+)^2} \frac{1}{1-b^2} \left(  (\zeta_*^+)^2 \lambda^+ + \zeta_{\new,k-1}^+ \lambda_\new^+ \right)
+ \frac{0.05(r_\new \zeta) b^2 \lambda^-}{(1-b^2)(1-b)^2} + 6 \epsilon
\]
w.p. at least $1 - p_{\l \e}$ with $p_{\l \e}: = 4 \cdot (2n) \exp\left(\frac{-\alpha \epsilon^2}{32 (b_{prob})^2} \right) + 2 \cdot (2n) \exp\left(\frac{-\alpha \epsilon^2}{32 (b_{prob, termw})^2} \right)$.


{\bf The $\et \et'$ term. }
Consider the second term. Using Fact \ref{matbnds} and the expression for $\et$ from \eqref{etdef0_cor}, 
\begin{align*}
\frac{1}{{\alpha}} \sum_t  \et \et'
& = \mathrm{term} + \mathrm{termw}, \ \text{where} \\
\mathrm{term}&:=  \frac{1}{{\alpha}} \sum_t   \ITt \invterm \ITt' \bm{\Phi}_{k-1}  (\lt \lt') \bm{\Phi}_{k-1} \ITt \invterm \ITt' \\
\mathrm{termw}&:=  \frac{1}{{\alpha}} \sum_t   \ITt \invterm \ITt' \bm{\Phi}_{k-1}  (-\wt \wt' - \lt \wt') + \wt \wt' + \\
& \frac{1}{{\alpha}} \sum_t (-\wt\wt' - \wt\lt')\bm{\Phi}_{k-1} \ITt \invterm \ITt' + \\
& \frac{1}{{\alpha}} \sum_t   \ITt \invterm \ITt' \bm{\Phi}_{k-1}  (\lt \wt' + \wt \wt' + \wt \lt') \bm{\Phi}_{k-1} \ITt \invterm \ITt'
\end{align*}
Consider the $\wt \wt'$ part of $\mathrm{termw}$. Let $\bN_t = \ITt \invterm \ITt' \bm{\Phi}_{k-1}$. Using Lemma \ref{blockdiag1} (support change lemma), the bound on $\epsilon_w^2$ and Lemma \ref{CSmat} (Cauchy-Schwartz),
\[
\|\frac{1}{{\alpha}} \sum_t \bN_t\wt\wt'\|_2 \leq \sqrt{\|\frac{1}{{\alpha}} \sum_t \bN_t\bN_t'\|_2\|\frac{1}{{\alpha}} \sum_t \wt\wt'\wt\wt'\|_2} \leq \sqrt{\rho^2 h^+ (\phi^+)^2}\epsilon_w^2
\]
Using Lemma \ref{blockdiag1} (support change lemma), we have
\[
\|\frac{1}{{\alpha}} \sum_t \bN_t \wt\wt'\bN_t'\|_2 \leq \rho^2 h^+ (\phi^+)^2\epsilon_w^2
\]
The $\lt \wt'$ in $\mathrm{termw}$ can be bounded by $\epsilon$ using the approach of Section \ref{general_decomp_w}. Thus, using the bound on $\epsilon_w^2$ from the theorem,
\[
\|\mathrm{termw}\|_2 \le (1+2\sqrt{\rho^2 h^+}\phi^+ + 2\rho^2 h^+ (\phi^+)^2) (0.03 \zeta \lambda^-) + 4\epsilon \le 2(\phi^+)^2 (0.03 \zeta \lambda^-) + 4\epsilon
\]
w.p. at least $1-4\cdot(2n) \exp\left(\frac{-\alpha \epsilon^2}{32 (b_{prob, termw})^2} \right)$. Here
\[
b_{prob, termw} = (\phi^+)^2  \frac{(2 r\zeta \sqrt{r}\gamma + \sqrt{r_\new} \gamma_\new) \epsilon_w}{1-b}.
\]

%

For $\mathrm{term}$, we proceed as in Section \ref{general_decomp} with $\bN_t' = \bM_t = \bm{\Phi}_{k-1} \ITt \invterm \ITt'$. Thus,
\begin{align*}
b_{prob,term2} &= \max(b_{prob,term21},b_{prob,term22},b_{prob,term23}) \le \frac{1}{(1-b)^2} (\phi^+)^2 (\zeta_*^+ \sqrt{r}\gamma + \sqrt{r_\new} \gamma_\new)^2 \\
b_{prob,term3} & \le \frac{1}{(1-b)^3} (\phi^+)^2 (2 r\zeta \sqrt{r} \gamma + \sqrt{r_\new} \gamma_\new)^2
\end{align*}
Use $b_{prob}$ to denote the upper bound on $\max(b_{prob,term2},b_{prob,term3})$. Then
\[
b_{prob} = \frac{1}{(1-b)^3} (\phi^+)^2 (2r \zeta \sqrt{r} \gamma + \sqrt{r_\new} \gamma_\new)^2
\]
Using \eqref{term1_upper_bnd}, \eqref{upper_bnd}, \eqref{term21_upper_bnd}, we get
\[
b_{term1}
 \le  \frac{(r_\new \zeta)^2 b^2}{(1-b^2)} \frac{r \gamma^2}{(1-b)^2} \le \frac{0.05(r_\new \zeta) b^2 \lambda^-}{(1-b^2)(1-b)^2}
%
\]
(we can get a tighter bound for the above, but do not need it and hence do not pursue it)
and
\[
\|\mathrm{term}\|_2 \le \| \frac{1}{\alpha} \sum_{t=t_0}^{t_0+\alpha-1} \ITt \left( \sum_{\tau=t_0}^{t}  b^{2t-2\tau}  \invterm \ITt' \bm{\Phi}_{k-1}  \bm{\Sigma}_\tau \bm{\Phi}_{k-1} \ITt \invterm \right) \ITt'\|_2
+ b_{term1} + 4\epsilon
\]
w.p. at least $1 - 4 \cdot (2n) \exp\left(\frac{-\alpha \epsilon^2}{32 (b_{prob})^2} \right)$.
By using Lemma \ref{blockdiag1} (support change lemma) and Lemma \ref{Dnew0_lem} for $k=1$ and by using only Lemma \ref{blockdiag1} (support change lemma) for $k>1$, we get
\[
\|\frac{1}{{\alpha}} \sum_t  \et \et'\|_2 \le b_{\bm{\e \e}}
\]
w.p. at least $1 - p_{\e \e}$ with
$ p_{\e \e}:= 4 \cdot (2n) \exp\left(\frac{-\alpha \epsilon^2}{32 (b_{prob})^2} \right) + 4\cdot(2n) \exp\left(\frac{-\alpha \epsilon^2}{32 (b_{prob, termw})^2} \right)$.


{\bf The $\bm{F}_t$ term. }
Consider the smallest term,  $\big\| \frac{1}{\alpha} \sum_t \bm{F}_t \big\|_2 = \|\bm{E}_{\rmnew,\perp} {\bm{E}_{\rmnew,\perp}}' \bm{\Phi}_{0} \lt \lt' \bm{\Phi}_{0} \bm{E}_{\rmnew}{\bm{E}_{\rmnew}}'\|_2$.
We again proceed as in Section \ref{general_decomp}. In this case $\bN_t = \bm{E}_{\rmnew,\perp} {\bm{E}_{\rmnew,\perp}}' \bm{\Phi}_{0}$ and $\bM_t = \bm{\Phi}_{0} \bm{E}_{\rmnew}{\bm{E}_{\rmnew}}'$.
Thus,
\begin{align*}
b_{prob,term2} & \le \frac{1}{(1-b)^2} (\zeta_*^+ \sqrt{r}\gamma) (2\zeta_*^+ \sqrt{r}\gamma + \sqrt{r_\new} \gamma_\new) \\
b_{prob,term3} & \le \frac{1}{(1-b)^3} (\zeta_*^+ \sqrt{r}\gamma) (2\zeta_*^+ \sqrt{r}\gamma + \sqrt{r_\new} \gamma_\new)
\end{align*}
and so
\[
b_{prob} = \frac{1}{(1-b)^3} (\zeta_*^+ \sqrt{r}\gamma) (2\zeta_*^+ \sqrt{r}\gamma + \sqrt{r_\new} \gamma_\new)
\]
\[
\| \frac{1}{\alpha} \sum_t \bm{F}_t\|_2
\le \frac{1}{1-b^2} (\zeta_*^+)^2 \lambda^+ +  \frac{0.05(r_\new \zeta) b^2 \lambda^-}{(1-b^2)(1-b)^2} + 4 \epsilon
\]
w.p. at least $1- p_{\bm{F}}$ with
$p_{\bm{F}}:= 4\cdot (2n) \exp\left(\frac{-\alpha \epsilon^2}{32 (b_{prob})^2} \right)$.
%
Combining the bounds on the three terms on the RHS of \eqref{add_calH1_cor} we get the final result of this lemma.
\end{proof}

\section{Proof of Deletion Azuma Lemmas - Lemma \ref{cluster_lem_azuma} and Lemmas \ref{Atil_cor}, \ref{Atilperp_cor}, \ref{Htil_cor}} \label{3_pfs_del}

\subsection{Proof of Lemma \ref{cluster_lem_azuma}}

\begin{proof}[Proof of Lemma \ref{cluster_lem_azuma}]
In this proof, all of the probabilistic statements are conditioned on $X_{\uhat_j+K}$ for $X_{\uhat_j+K} \in \Gamma_{j,K}^{\uhat_{j}}$ for ${\uhat_{j}}=u_{j}$ or $u_{j}+1$.

Let $t_0:= \that_{cl}$.
Using Fact \ref{d_large}, under the given conditioning, for all $t \in [t_0, t_0 + \alpha-1]$,
\begin{align} \label{eq_Sigma_j}
\E[\nut \nut'] = \bm{\Sigma}_t = \bm{\Sigma}_{(j)} :=  \P_{(j)} \Lamj \P_{(j)}{}' 
\end{align}
We need to bound
$f = \|\frac{1}{\alpha}\sum_{t=t_0}^{t_0 + \alpha-1} \lt {\lt}' - \frac{1}{1-b^2} \bm{\Sigma}_{(j)}\|_2$.
Let
\[
\epsilon: = \frac{1}{1-b^2} 0.001r_\new \zeta \lambda^-
\]

To do this we can proceed as in Section \ref{general_decomp} with $\bN_t = \bM_t' = \I$ but with one change. We include the constant term $- \frac{1}{1-b^2} \bm\Sigma_{(j)}$ in $\mathrm{term21}$.
Thus,
\begin{align*}
b_{prob,term2} & \le \frac{2}{(1-b)^2} (\sqrt{r}\gamma)^2  \\
b_{prob,term3} & \le \frac{1}{(1-b)^3} (\sqrt{r} \gamma)^2
\end{align*}

Let
\[
f_{term21}:=\|\frac{1}{\alpha} \sum_{t=t_0}^{t_0+\alpha-1} \sum_{\tau=t_0}^{t} b^{2t-2\tau} \bm{\Sigma}_{\tau}  -  \frac{1}{1-b^2} \bm{\Sigma}_{(j)}\|_2.
\]
Using \eqref{term1_upper_bnd}, \eqref{upper_bnd}, \eqref{term21_upper_bnd}, we get
\[
b_{term1} = \frac{(r_\new \zeta)^2 b^2}{(1-b^2)} \max_t \|\l_{t_0-1} \l_{t_0-1}{}'\|_2
\le  \frac{(r_\new \zeta)^2 b^2}{(1-b^2)} \frac{(r \gamma^2)}{(1-b)^2} \le \frac{0.05(r_\new \zeta) b^2 \lambda^-}{(1-b^2)(1-b)^2}
\]
and
\begin{align*}
f& \le f_{term21} +  b_{term1} + 4 \epsilon
\end{align*}
w.p. at least $1 - 3 \cdot (2n) \exp\left(\frac{-\alpha \epsilon^2}{32 (b_{prob,term2})^2} \right)  - (2n) \exp\left(\frac{-\alpha \epsilon^2}{32 (b_{prob,term3})^2} \right)$.

Since $\bm\Sigma_\tau = \bm\Sigma_{(j)}: = \P_{(j)} \Lamj \P_{(j)}{}'$ for this interval, using Lemma \ref{sum_often} and using the bound on $1/\alpha$ from Fact \ref{d_large},
\begin{align*}
f_{term21} &= \|\frac{1}{\alpha} \sum_{t=t_0}^{t_0+\alpha-1} \sum_{\tau=t_0}^{t} b^{2t-2\tau}  \bm{\Sigma}_{(j)}    -  \frac{1}{1-b^2} \bm{\Sigma}_{(j)}\|_2 \\
& = \|\frac{1}{1-b^2}(1 - \frac{1}{\alpha} \frac{b^2(1 - b^{2 \alpha})}{1-b^2}) \bm{\Sigma}_{(j)}  - \frac{1}{1-b^2} \bm{\Sigma}_{(j)}\|_2 \\
& \le  \frac{1}{\alpha} \frac{b^2}{(1-b^2)^2} \|\bm{\Sigma}_{(j)}\|_2  \le  (r_\new \zeta)^2\frac{b^2}{(1-b^2)^2}  \lambda^+  \le (r_\new \zeta) \frac{b^2}{(1-b^2)^2}  0.05\lambda^-
\end{align*}
%
Thus,
\[
f \le (r_\new \zeta)\frac{b^2}{(1-b^2)^2} 0.05 \lambda^-  + \frac{0.05(r_\new \zeta) b^2 \lambda^-}{(1-b^2)(1-b)^2} + 4 \epsilon \le \bbb
\]
w.p. at least $1 - p_{\mathrm{cl}}$ where
$p_{\mathrm{cl}}:= 4\cdot (2n) \exp\left(-\frac{\alpha \epsilon^2 (1-b)^6}{32 \cdot 4r^2 \gamma^4} \right)$.
\end{proof}


\subsection{Definitions and preliminaries for proofs of  Lemmas \ref{Atil_cor}, \ref{Atilperp_cor}, \ref{Htil_cor}}

\newcommand{\cur}{\mathrm{cur}}
\begin{definition}
Define
\begin{enumerate}
\item $\bm{G}_{j,1,\det} := [.]$ and for $k=2,3, \dots \vartheta$, $\bm{G}_{j,k,\det} := [\bm{G}_{j,1} \ \bm{G}_{j,2} \ \dots \ \bm{G}_{j,k-1}]$.
\item Define $\bm{G}_{j,k,\undet} := [\bm{G}_{j,k+1} \ \bm{G}_{j,k+2} \ \dots \ \bm{G}_{j,\vartheta_j}]$, $\bm{G}_{j,k,\mathrm{cur}} := \bm{G}_{j,k}$;

\item Define $\hat{\bm{G}}_{j,1,\det} = [.]$ and $\hat{\bm{G}}_{j,k,\det}:= [\hat{\bm{G}}_{j,1} \ \hat{\bm{G}}_{j,2} \ \dots \ \hat{\bm{G}}_{j,k-1}]$ 

\item $\bm\Psi_{j,k}:= (\I - \hat{\bmG}_{j,k,\det}\hat{\bmG}_{j,k,\det}{}')$; thus $\bm\Psi_{j,1} = \I$

\item ${\bm{D}}_{j,k,\mathrm{cur}} :=  \bm\Psi_{j,k} \bm{G}_{j,k,\cur} $,
${\bm{D}}_{j,k,\det} := \bm\Psi_{j,k} \bm{G}_{j,k,\det} $,
${\bm{D}}_{j,k,\undet} :=  \bm\Psi_{j,k} \bm{G}_{j,k,\undet} $;

\end{enumerate}
\end{definition}

\begin{definition}\label{defAHtil}
\
\begin{enumerate}
\item Let ${\bm{D}}_{j,k,\mathrm{cur}} := \bm\Psi_{j,k} \bm{G}_{j,k,\cur} \overset{\mathrm{QR}}= {\bm{E}}_{j,k,\mathrm{cur}} {\bm{R}}_{j,k,\mathrm{cur}} $ denote its reduced QR decomposition.  So ${\bm{E}}_{j,k,\mathrm{cur}}$ is a basis matrix, and ${\bm{R}}_{j,k,\mathrm{cur}}$ is upper triangular.
 Let ${\bm{E}}_{j,k,\mathrm{cur},\perp}$ be a basis matrix for the orthogonal complement of $\range({\bm{E}}_{j,k,\mathrm{cur}})$.
\item Using ${\bm{E}}_{j,k,\mathrm{cur}}$ and ${\bm{E}}_{j,k,\mathrm{cur},\perp}$ , define
\begin{align*}
\tilde{\bm{A}}_{j,k} &:= \frac{1}{{\alpha}} \sum_{t\in\Ijkt} {\bm{E}}_{j,k,\mathrm{cur}}{}' \bm\Psi_{j,k} \lt {\lt}'\bm\Psi_{j,k} {\bm{E}}_{j,k,\mathrm{cur}} \\
\tilde{\bm{A}}_{j,k,\perp} &:= \frac{1}{{\alpha}} \sum_{t\in\Ijkt} {\bm{E}}_{j,k,\mathrm{cur},\perp}{}' \bm\Psi_{j,k} \lt {\lt}' \bm\Psi_{j,k} {\bm{E}}_{j,k,\mathrm{cur},\perp}
\end{align*}
and let
\[
\tilde{\bm{\mathcal{A}}}_{j,k} := \left[ \begin{array}{cc} {\bm{E}}_{j,k,\mathrm{cur}} & {\bm{E}}_{j,k,\mathrm{cur},\perp} \\ \end{array} \right]
\left[\begin{array}{cc} \tilde{\bm{A}}_{j,k} \ & \bm{0} \ \\ \bm{0} \ & \tilde{\bm{A}}_{j,k,\perp}  \\ \end{array} \right]
\left[ \begin{array}{c} {\bm{E}}_{j,k,\mathrm{cur}}{}' \\ {\bm{E}}_{j,k,\mathrm{cur},\perp}{}' \\ \end{array} \right]
\]

\item Define
\begin{align*}
\tilde{\bm{\mathcal{H}}}_{j,k} = \frac{1}{{\alpha}} \sum_{t \in \Ijkt } \bm\Psi_{j,k} \hat{\bm{\ell}}_t \hat{\bm{\ell}}_t{}' \bm\Psi_{j,k} -  \tilde{\bm{\mathcal{A}}}_{j,k}
\end{align*}
From Algorithm \ref{reprocsdet},
\begin{align*}
\frac{1}{{\alpha}} \sum_{t \in \Ijkt } \bm\Psi_{j,k} \hat{\bm{\ell}}_t \hat{\bm{\ell}}_t{}' \bm\Psi_{j,k}
\overset{\mathrm{EVD}}{=} \left[ \begin{array}{cc} \hat{\bmG}_{j,k} & \hat{\bmG}_{j,k,\perp} \\ \end{array} \right]
\left[\begin{array}{cc} \hat{\bm{\Lambda}}_{t} \ & \bm{0} \ \\ \bm{0} \ & \ \hat{\bm{\Lambda}}_{t,\perp} \\ \end{array} \right]
\left[ \begin{array}{c} \hat{\bmG}_{j,k}{}' \\ \hat{\bmG}_{j,k,\perp}{}' \\ \end{array} \right].
\end{align*}

\end{enumerate}

\end{definition}

\newcommand{\Lamtdet}{\bm\Lambda_{t,\det}}
\newcommand{\Lamtcur}{\bm\Lambda_{t,\cur}}
\newcommand{\Lamtundet}{\bm\Lambda_{t,\undet}}


\begin{lem}\cite{rrpcp_perf} \label{matbnds_del}
When $X_{(\uhat_j+K+1)+k-1}\in \tilde{\Gamma}_{j,k-1}^a$ with  $a=u_{j}$ or $a=u_{j}+1$,
\begin{enumerate}
\item $\|{\bm{D}}_{j,k,\det}\|_2 \leq r\zeta$
\item $\ds\sqrt{1-(r)^2\zeta^2} \leq \sigma_i({\bm{R}}_{j,k,\cur})=\sigma_i({\bm{D}}_{k,\cur}) \leq 1$
\item $\ds\|{\bm{E}}_{j,k,\mathrm{cur}}{}'{\bm{D}}_{j,k,\undet}\|_2 \leq \frac{(r \zeta)^2}{\sqrt{1-(r)^2\zeta^2}}$
\item 
\[\bm\Psi_{j,k} \bm\Sigma_t \bm\Psi_{j,k} = \left[{\bm{D}}_{j,k,\det} \ {\bm{D}}_{j,k,\cur} \ {\bm{D}}_{j,k,\undet} \right] \left[\begin{array}{ccc} \Lamtdet \ & \bm{0} \ & \bm{0} \  \\ \bm{0} \ & \ \Lamtcur & \ \\ \bm{0} \ & \bm{0} \ & \Lamtundet \ \end{array} \right]  \left[ \begin{array}{c}{\bm{D}}_{j,k,\det} \\ {\bm{D}}_{j,k,\cur} \\ {\bm{D}}_{j,k,\undet} \end{array} \right]'
  \]
    with $\lambda_{\max}(\Lamtdet) \le \lambda^+$, $\lambda_{j,k}^- \le \lambda_{\min}(\Lamtcur) \le \lambda_{\max}(\Lamtcur) \le \lambda_{j,k}^+$, $\lambda_{\max}(\Lamtundet) \le \lambda_{j,k+1}^+$. 

\item Using the first four claims, it is easy to see that
\ben
\item $\|{\bm{E}}_{j,k,\cur,\perp}{}'\bm\Psi_{j,k} \bm\Sigma_t \bm\Psi_{j,k}{\bm{E}}_{j,k,\cur,\perp}\|_2 \le (r\zeta)^2 \lambda^+ + \lambda_{k+1}^+$ (when $k=1$, the first term disappears)
\item $\|{\bm{E}}_{j,k,\cur,\perp}{}'\bm\Psi_{j,k} \bm\Sigma_t \bm\Psi_{j,k}{\bm{E}}_{j,k,\cur}\|_2 \le (r\zeta)^2 \lambda^+ +  \frac{(r \zeta)^2}{\sqrt{1-(r)^2\zeta^2}}\lambda_{k+1}^+$ (when $k=1$, the bound equals zero)
\item $\|\bm\Psi_{j,k} \bm\Sigma_t \bm\Phi_{j,K}\|_2 \le ((r\zeta)\lambda^+ + \lambda_k^+) (r+r_\new) \zeta$
\item $\|\bm\Phi_{j,K} \bm\Sigma_t \bm\Phi_{j,K}\|_2 \le ( (r+r_\new) \zeta )^2 \lambda^+$
\een
\end{enumerate}
\end{lem}

\begin{proof}
Consider the first claim. When $k=1$, ${\bm{G}}_{k,\det} = [.]$ and hence ${\bm{D}}_{j,k,\det} = [.]$. Thus $\|{\bm{D}}_{j,k,\det}\|_2 = 0 \leq r\zeta$.
For $k>1$, it follows by applying Lemmas \ref{cPCA} and \ref{zetadecay} applied for $\tilde{k} = 1,2, \dots, k-1$.
The next two claims follow using Lemma \ref{hatswitch}. Notice that $\bm{D}_{k,\cur} = \bm\Psi_{j,k} \bm{G}_{j,k,\cur}$ where $\bm\Psi_{j,k} = (I - \hat{\bm{G}}_{j,k,\det} \hat{\bm{G}}_{j,k,\det}{}')$. Use item 4 of Lemma \ref{hatswitch} and the fact that $\bm{G}_{j,k,\det}{}' \bm{G}_{j,k,\cur} = \bm{0}$ to get the second claim. For the third claim, notice that ${\bm{E}}_{j,k,\mathrm{cur}}{}'{\bm{D}}_{j,k,\undet} = {\bm{R}}_{j,k,\mathrm{cur}}^{-1} \bm{G}_{j,k,\cur}{}' \bm{\Psi}_{j,k} \bm{\Psi}_{j,k} \bm{G}_{j,k,\undet}$. Use the previous claim to bound $\|{\bm{R}}_{j,k,\mathrm{cur}}^{-1}\|_2$. Use item 3 of Lemma \ref{hatswitch} and the facts that $\bm{G}_{j,k,\det}{}' \bm{G}_{j,k,\cur} = 0$ and $\bm{G}_{j,k,\det}{}' \bm{G}_{k,\undet} = 0$ to bound $\|\bm{G}_{j,k,\cur}{}' \bm{\Psi}_{j,k}\|_2$ and $\|\bm{\Psi}_{j,k} \bm{G}_{j,k,\undet}\|_2$ respectively.
When $k=1$, both the above claims follow even more easily: ${\bm{D}}_{k,\cur} = {\bm{G}}_{k,\cur}$ and so $\sigma_i({\bm{D}}_{k,\cur}) = 1$ and thus satisfies the given bounds; also, ${\bm{E}}_{k,\mathrm{cur}} = {\bm{G}}_{k,\cur}$ and ${\bm{D}}_{k,\undet} = {\bm{G}}_{k,\undet}$ and thus, $\|{\bm{E}}_{j,k,\mathrm{cur}}{}'{\bm{D}}_{j,k,\undet}\|_2 =0 \le \frac{(r \zeta)^2}{\sqrt{1-(r)^2\zeta^2}}$.

The fourth claim just uses the definitions and Model \ref{clust_model}.
\end{proof}

\subsection{Proofs of Lemmas \ref{Atil_cor}, \ref{Atilperp_cor}, \ref{Htil_cor}}
We remove the subscript $j$ at various places in this section. 

\begin{proof}[Proof of Lemma \ref{Atil_cor}]
In this proof all probabilistic statements are conditioned on $X_{(\uhat_j+K+1)+k-1}$ for all $X_{(\uhat_j+K+1)+k-1}\in \tilde{\Gamma}_{j,k-1}^a$ with  $a=u_{j}$ or $a=u_{j}+1$.

Recall that $\tilde{\bm{A}}_{k} := \frac{1}{{\alpha}} \sum_{t\in\Ijkt} {\bm{E}_{k,\cur}}' {\bm\Psi}_{k} \lt {\lt}' {\bm\Psi}_{k}\bm{E}_{k,\cur}$.

We proceed as in Section \ref{general_decomp}. In this case $\bN_t' = \bM_t =  {\bm\Psi}_{k}\bm{E}_{k,\cur}$ and $t_0$ is the first time instant of $\Ijkt$. Thus,
\begin{align*}
\lambda_{\min}(\tilde{\bm{A}}_{k})
& \ge \lambda_{\min}(  \frac{1}{\alpha} \sum_{t=t_0}^{t_0+\alpha-1} \sum_{\tau=t_0}^{t}  b^{2t-2\tau} {\bm{E}_{k,\cur}}' {\bm\Psi}_{k} \bm{\Sigma}_\tau {\bm\Psi}_{k}\bm{E}_{k,\cur}  ) - 4 \epsilon \\
& \ge \frac{1}{1-b^2}(1 - \frac{b^2}{\alpha (1-b^2)}) \min_{\tau \in [t_0, t_0+\alpha-1]} \lambda_{\min}({\bm{E}_{k,\cur}}' {\bm\Psi}_{k} \bm{\Sigma}_\tau {\bm\Psi}_{k}\bm{E}_{k,\cur}) - 4 \epsilon
\end{align*}
with probability at least $ 1-4\cdot(2n)\exp\left(-\frac{\alpha\epsilon^2}{32b_{prob}^2}\right)$, where $ b_{prob}=\frac{r\gamma^2}{(1-b)^2}$.

Finally, using Lemma \ref{matbnds_del} and Ostrowski's theorem,
\[
\lambda_{\min}(\tilde{\bm{A}}_{k}) \ge \frac{1}{1-b^2}(1 - \frac{b^2}{\alpha (1-b^2)}) (1 - (r \zeta)^2) \lambda_k^-
\]
\end{proof}

\begin{proof}[Proof of Lemma \ref{Atilperp_cor}]
In this proof all probabilistic statements are conditioned on  $X_{(\uhat_j+K+1)+k-1}$ for all $X_{(\uhat_j+K+1)+k-1}\in \tilde{\Gamma}_{j,k-1}^a$ with  $a=u_{j}$ or $a=u_{j}+1$.

Recall that $\tilde{\bm{A}}_{k,\perp} := \frac{1}{{\alpha}} \sum_{t} {\bm{E}_{k,\cur,\perp}}' {\bm\Psi}_{k} \lt {\lt}' {\bm\Psi}_{k}\bm{E}_{k,\cur,\perp}$.

We proceed as in Section \ref{general_decomp}. In this case $\bN_t' = \bM_t = {\bm\Psi}_{k}\bm{E}_{k,\cur,\perp}$. Thus,
\[
b_{term1} =
\frac{(r_\new \zeta)^2 b^2}{(1-b^2)}  \max_{t \in [t_0, t_0+\alpha-1]} \lambda_{\max}({\bm{E}_{k,\cur,\perp}}' {\bm\Psi}_{k} ( \l_{t_0-1} \l_{t_0-1}{}' ) {\bm\Psi}_{k}\bm{E}_{k,\cur,\perp})
\le  \frac{(r_\new \zeta)^2 b^2}{(1-b^2)} \frac{(r \gamma^2)}{(1-b)^2} \le \frac{0.05(r_\new \zeta) b^2 \lambda^-}{(1-b^2)(1-b)^2}
\]
(we can get a tighter bound for the above, but do not need it and hence do not pursue it)
and
\begin{align*}
\lambda_{\max}(\tilde{\bm{A}}_{k,\perp})
& \le \lambda_{\max}( \frac{1}{\alpha} \sum_{t=t_0}^{t_0+\alpha-1} \sum_{\tau=t_0}^{t}  b^{2t-2\tau} {\bm{E}_{k,\cur,\perp}}' {\bm\Psi}_{k} \bm{\Sigma}_\tau {\bm\Psi}_{k}\bm{E}_{k,\cur,\perp} ) +
b_{term1} + 4 \epsilon \\
& \le \frac{1}{1-b^2} \max_{\tau \in [t_0, t_0+\alpha-1]} \lambda_{\max}({\bm{E}_{k,\cur,\perp}}' {\bm\Psi}_{k} \bm{\Sigma}_\tau {\bm\Psi}_{k}\bm{E}_{k,\cur,\perp} ) + b_{term1} + 4 \epsilon
\end{align*}
with probability at least $ 1-4\cdot(2n)\exp\left(-\frac{\alpha\epsilon^2}{32b_{prob}^2}\right)$, where $ b_{prob}=\frac{r\gamma^2}{(1-b)^2}$.

Thus, using Lemma \ref{matbnds_del},
\[
\lambda_{\max}(\tilde{\bm{A}}_{k,\perp}) \le \frac{1}{1-b^2} ( (r\zeta)^2 \lambda^+ + \lambda_{k+1}^+ ) + \frac{0.05(r_\new \zeta) b^2 \lambda^-}{(1-b^2)(1-b)^2}  + 4 \epsilon
\]
\end{proof}

\begin{proof}[Proof of Lemma \ref{Htil_cor}]
In this proof all probabilistic statements are conditioned on  $X_{(\uhat_j+K+1)+k-1}$ for all $X_{(\uhat_j+K+1)+k-1}\in \tilde{\Gamma}_{j,k-1}^a$ with  $a=u_{j}$ or $a=u_{j}+1$. Recall that $\bm\Psi_{j,1} = \I$.

In a fashion similar to the proof of Lemma \ref{calHk_cor}, we can show that
\begin{equation}\label{calH_cor}
\|\tilde{\bm{\mathcal{H}}}_k\|_2 \leq  2\bigg\| \frac{1}{{\alpha}} \sum_t  {\bm\Psi}_{k} \lt {\bm{e}_t}'  \bigg\|_2 + \bigg\| \frac{1}{{\alpha}} \sum_t \bm{e}_t {\bm{e}_t}' \bigg\|_2 +   2\bigg\| \frac{1}{{\alpha}} \sum_t \bm{F}_t \bigg\|_2
\end{equation}
where
\[
\bm{F}_t = \bm{E}_{k,\cur} {\bm{E}_{k,\cur}}' {\bm\Psi}_{k} \lt {\lt}' {\bm\Psi}_{k} \bm{E}_{k,\cur,\perp} {\bm{E}_{k,\cur,\perp}}'.
\]
We now bound the three terms above.

Consider $\| \frac{1}{{\alpha}} \sum_t  {\bm\Psi}_{k}\lt {\bm{e}_t}'  \|_2$. Using Lemma \ref{cslem_cor}, $\et$ satisfies \eqref{etdef0_cor} with probability one under the given conditioning. Thus,
\begin{align*}
\frac{1}{{\alpha}} \sum_t  {\bm\Psi}_{k}\lt {\bm{e}_t}'
& = \frac{1}{{\alpha}} \sum_t {\bm\Psi}_{k} \lt  (\lt+ \wt)' \bm{\Phi}_{K} \ITt \invterm \ITt'  - \frac{1}{\alpha}\sum_t {\bm\Psi}_{k}\lt\wt'\\
& = \frac{1}{{\alpha}} \sum_t {\bm\Psi}_{k} \lt \lt' \bm{\Phi}_{K} \ITt \invterm \ITt'  + \frac{1}{{\alpha}} \sum_t {\bm\Psi}_{k} \lt  \wt' \bm{\Phi}_{K} \ITt \invterm \ITt' - \frac{1}{\alpha}\sum_t {\bm\Psi}_{k}\lt\wt' \\
& := \mathrm{term} + \mathrm{termw}
\end{align*}
{ Here $\mathrm{termw}$ refers to the terms containing $\wt$. By following the approach of Section \ref{general_decomp_w}, under given conditions,
$$\|\mathrm{termw}\|_2\leq 2\epsilon$$
w.p. at least $1-2\cdot(2n)\exp\left(\frac{-\alpha\epsilon^2}{32(b_{prob,termw})^2}\right)$, where
$$b_{prob,termw} = \frac{\phi^+\sqrt{r}\gamma\epsilon_w}{1-b}.$$
}

We proceed as in Section \ref{general_decomp} for $\mathrm{term}$. In this case $\bN_t = \bN_0 = {\bm\Psi}_{k}$ and $\bM_t = \bm{\Phi}_{K} \ITt \invterm \ITt'$. Thus,
\[
b_{term1}
\le  \frac{(r_\new \zeta)^2 b^2}{(1-b^2)} \frac{\phi^+(r \gamma^2)}{(1-b)^2} \le \frac{0.05(r_\new \zeta) b^2 \lambda^-}{(1-b^2)(1-b)^2}
\]
(we can get a tighter bound for the above, but do not need it and hence do not pursue it)
and
\begin{align*}
\|\mathrm{term}\|_2
& \le \| \frac{1}{\alpha} \sum_{t=t_0}^{t_0+\alpha-1} \sum_{\tau=t_0}^{t}  b^{2t-2\tau} {\bm\Psi}_{k} \bm{\Sigma}_\tau \bm{\Phi}_{K} \ITt \invterm \ITt' \|_2 +
b_{term1} + 4 \epsilon
\end{align*}
{ w.p. at least $1 -4\cdot(2n)\exp\left(\frac{-\alpha\epsilon^2}{32(b_{prob})^2}\right)$, where $$b_{prob} = \frac{\phi^+r\gamma^2(r+r_{\new})\zeta}{(1-b)^2}.$$}
Let
\[
\frac{1}{\alpha} \sum_{t=t_0}^{t_0+\alpha-1} \sum_{\tau=t_0}^{t}  b^{2t-2\tau} {\bm\Psi}_{k} \bm{\Sigma}_\tau \bm{\Phi}_{K} \ITt \invterm \ITt' :=\frac{1}{\alpha} \sum_{t=t_0}^{t_0+\alpha-1} \bm{X}_t \bm{Y}_t'
\]
where $\bm{X}_t:=\sum_{\tau=t_0}^{t}  b^{2t-2\tau} {\bm\Psi}_{k} \bm{\Sigma}_\tau \bm{\Phi}_{K}$ and $\bm{Y}_t:=\ITt \invterm \ITt'$.
By Lemma \ref{blockdiag1} (support change lemma) $\lambda_{\max}( \frac{1}{\alpha} \sum_{t=t_0}^{t_0+\alpha-1} \bm{Y}_t \bm{Y}_t' ) \le \rho^2 h^+ (\phi^+)^2$.

By Lemma \ref{matbnds_del} and the fact that $\|\Phi_K P_* \|_2 \le \zeta_{*}^+ = r\zeta$ and $\|\Phi_K P_\new\|_2 \le r_\new \zeta$,
\begin{align*}
\lambda_{\max}( \frac{1}{\alpha} \sum_{t=t_0}^{t_0+\alpha-1} \bm{X}_t \bm{X}_t')
& \le \max_t \|\bm{X}_t\|^2
 \le ( \frac{1}{1-b^2} (r + r_\new) \zeta ( (r\zeta) \lambda^+ + \lambda_k^+) )^2
\end{align*}
Thus, by Cauchy-Schwartz for matrices,
\[
\| \frac{1}{\alpha} \sum_{t=t_0}^{t_0+\alpha-1} \sum_{\tau=t_0}^{t}  b^{2t-2\tau} {\bm\Psi}_{k} \bm{\Sigma}_\tau \bm{\Phi}_{K} \ITt \invterm \ITt' \|_2
\le \sqrt{\rho^2 h^+ (\phi^+)^2}   ( \frac{1}{1-b^2} (r + r_\new) \zeta ( (r\zeta) \lambda^+ + \lambda_k^+) )
\]

Thus, with probability at least $ 1 - p_{\tilde{\bm{le}}}$,
\[
\bigg\| \frac{1}{{\alpha}} \sum_t  {\bm\Psi}_{k}\lt {\bm{e}_t}'  \bigg\|_2
\le
\sqrt{\rho^2 h^+ (\phi^+)^2}   ( \frac{1}{1-b^2} (r + r_\new) \zeta ( (r\zeta) \lambda^+ + \lambda_k^+) ) + \frac{0.05(r_\new \zeta) b^2 \lambda^-}{(1-b^2)(1-b)^2} +  6 \epsilon
\]
where $p_{\tilde{\bm{le}}} = 2\cdot(2n)\exp\left(\frac{-\alpha\epsilon^2}{32(b_{prob,termw})^2}\right) +4\cdot(2n)\exp\left(\frac{-\alpha\epsilon^2}{32(b_{prob,term})^2}\right)$.

{
Next consider the $\et \et'$ term. Recall that, under the given conditioning, $\et = \ITt \invterm \ITt' \bm{\Phi}_{K}  (\lt+ \wt) - \wt$. Thus,
\begin{align*}
\frac{1}{{\alpha}} \sum_t  \et \et'
& = \mathrm{term} + \mathrm{termw}, \ \text{where} \\
\mathrm{term}&:=  \frac{1}{{\alpha}} \sum_t   \ITt \invterm \ITt' \bm{\Phi}_{K}  (\lt \lt') \bm{\Phi}_{K} \ITt \invterm \ITt' \\
\mathrm{termw}&:=  \frac{1}{{\alpha}} \sum_t   \ITt \invterm \ITt' \bm{\Phi}_{K}  (-\wt \wt' - \lt \wt') + \wt \wt' + \\
& \frac{1}{{\alpha}} \sum_t (-\wt\wt' - \wt\lt')\bm{\Phi}_{K} \ITt \invterm \ITt' + \\
& \frac{1}{{\alpha}} \sum_t   \ITt \invterm \ITt' \bm{\Phi}_{K}  (\lt \wt' + \wt \wt' + \wt \lt') \bm{\Phi}_{K} \ITt \invterm \ITt'
\end{align*}
For the $\wt \wt'$ part of $\mathrm{termw}$, let $\bN_t = \ITt \invterm \ITt' \bm{\Phi}_{K} $. Using Lemma \ref{CSmat} (Cauchy-Schwartz),  Lemma \ref{blockdiag1} (support change lemma) and the bound on $\epsilon_w^2$, we have
\[
\|\frac{1}{{\alpha}} \sum_t \bN_t\wt\wt'\|_2 \leq \sqrt{\|\frac{1}{{\alpha}} \sum_t \bN_t\bN_t'\|_2\|\frac{1}{{\alpha}} \sum_t \wt\wt'\wt\wt'\|_2} \leq \sqrt{\rho^2 h^+ (\phi^+)^2}\epsilon_w^2
\]
Using Lemma \ref{blockdiag1} (support change lemma), we have
\[
\|\frac{1}{{\alpha}} \sum_t \bN_t \wt\wt'\bm{\Phi}_{K} \bN_t'\|_2 \leq \rho^2 h^+ (\phi^+)^2\epsilon_w^2
\]
The $\lt \wt'$ in $\mathrm{termw}$ can be bounded by $\epsilon$ using the approach of Section \ref{general_decomp_w}. Thus,
\[
\|\mathrm{termw}\|_2 \le (1+2\sqrt{\rho^2 h^+}\phi^+ + 2\rho^2 h^+ (\phi^+)^2) (0.03 \zeta \lambda^-) + 4\epsilon \le 2(\phi^+)^2 (0.03 \zeta \lambda^-)
\]
w.p. at least $1-4\cdot(2n) \exp\left(\frac{-\alpha \epsilon^2}{32 (b_{prob, termw})^2} \right)$
where
\[
b_{prob, termw} = (\phi^+)^2  \frac{(2 r\gamma \sqrt{r}\gamma + \sqrt{r_\new} \gamma_\new) \epsilon_w}{1-b}.
\]

For $\mathrm{term}$, we proceed as in Section \ref{general_decomp} with $\bN_t' = \bM_t = \bm{\Phi}_{K} \ITt \invterm \ITt'$. Thus,
\begin{align*}
b_{prob,term2} &= \max(b_{prob,term21},b_{prob,term22},b_{prob,term23}) \le \frac{1}{(1-b)^2} (\phi^+)^2 (\zeta_*^+ \sqrt{r}\gamma + \sqrt{r_\new} \gamma_\new)^2 \\
b_{prob,term3} & \le \frac{1}{(1-b)^3} (\phi^+)^2 (2 r\zeta \sqrt{r} \gamma + \sqrt{r_\new} \gamma_\new)^2
\end{align*}
Use $b_{prob}$ to denote the upper bound on $\max(b_{prob,term2},b_{prob,term3})$. Then
\[
b_{prob} = \frac{1}{(1-b)^3} (\phi^+)^2 (2r \zeta \sqrt{r} \gamma + \sqrt{r_\new} \gamma_\new)^2
\]
Using \eqref{term1_upper_bnd}, \eqref{upper_bnd}, \eqref{term21_upper_bnd}, we get
\[
b_{term1}
\le  \frac{(r_\new \zeta)^2 b^2}{(1-b^2)} \frac{(r \gamma^2) }{(1-b)^2} \le \frac{0.05(r_\new \zeta) b^2 \lambda^-}{(1-b^2)(1-b)^2}
%
\]
(we can get a tighter bound for the above, but do not need it and hence do not pursue it)
and
\[
\|\mathrm{term}\|_2 \le \| \frac{1}{\alpha} \sum_{t=t_0}^{t_0+\alpha-1} \ITt \left( \sum_{\tau=t_0}^{t}  b^{2t-2\tau}  \invterm \ITt' \bm{\Phi}_{K}  \bm{\Sigma}_\tau \bm{\Phi}_{K} \ITt \invterm \right) \ITt'\|_2
+ b_{term1} + 4\epsilon
\]
w.p. at least $1 - 4 \cdot (2n) \exp\left(\frac{-\alpha \epsilon^2}{32 (b_{prob})^2} \right)$.
By Lemma \ref{blockdiag1} (support change lemma), Lemma \ref{matbnds_del}, and the fact that $\|\Phi_K P_* \|_2 \le \zeta_{*}^+ = r\zeta$ and $\|\Phi_K P_\new\|_2 \le r_\new \zeta$,
\begin{align*}
&& \| \frac{1}{\alpha} \sum_{t=t_0}^{t_0+\alpha-1} \ITt \left( \sum_{\tau=t_0}^{t}  b^{2t-2\tau}  \invterm \ITt' \bm{\Phi}_{K}  \bm{\Sigma}_\tau \bm{\Phi}_{K} \ITt \invterm \right) \ITt'\|_2 \\
&& \le \rho^2 h^+  \ (\phi^+)^2 \frac{1}{1-b^2} ((r+r_\new) \zeta )^2 \lambda^+  
\end{align*}
Combining all the bounds from above,
\[
\|\frac{1}{{\alpha}} \sum_t  \et \et'\|_2 \le
\rho^2 h^+  \ (\phi^+)^2 \frac{1}{1-b^2}  ((r+r_\new) \zeta )^2 \lambda^+
+  \frac{0.05(r_\new \zeta) b^2 \lambda^-}{(1-b^2)(1-b)^2}
+ (\phi^+)^2 2(0.03 \zeta \lambda^-)
+ 8 \epsilon
\]
w.p. at least $1 - p_{\tilde{\e \e}}$ with
$p_{\tilde{\e \e}}:= 4 \cdot (2n) \exp\left(\frac{-\alpha \epsilon^2}{32 (b_{prob})^2} \right) + 4\cdot(2n) \exp\left(\frac{-\alpha \epsilon^2}{32 (b_{prob, termw})^2} \right)$.
}

Finally consider $\bigg\| \frac{1}{{\alpha}} \sum_t \bm{F}_t \bigg\|_2 = \|\frac{1}{{\alpha}} \sum_t  \bm{E}_{k,\cur} {\bm{E}_{k,\cur}}' {\bm\Psi}_{k} \lt {\lt}' {\bm\Psi}_{k} \bm{E}_{k,\cur,\perp} {\bm{E}_{k,\cur,\perp}}'\|_2$.
We proceed as in Section \ref{general_decomp}. Here $\bN_t = \bm{E}_{k,\cur} {\bm{E}_{k,\cur}}' {\bm\Psi}_{k}$ and $\bM_t = {\bm\Psi}_{k} \bm{E}_{k,\cur,\perp} {\bm{E}_{k,\cur,\perp}}'$. Thus, we get
\[
b_{term1}
\le  \frac{(r_\new \zeta)^2 b^2}{(1-b^2)} \frac{(r \gamma^2)}{(1-b)^2} \le \frac{0.05(r_\new \zeta) b^2 \lambda^-}{(1-b^2)(1-b)^2}
\]
(we can get a tighter bound for the above, but do not need it and hence do not pursue it)
and
\[
\|\frac{1}{{\alpha}} \sum_t \bm{F}_t \|_2
\le
\frac{1}{1-b^2} \max_{\tau \in [t_0, t_0+\alpha-1]} \| \bm{E}_{k,\cur} {\bm{E}_{k,\cur}}' {\bm\Psi}_{k} \bm{\Sigma}_\tau {\bm\Psi}_{k} \bm{E}_{k,\cur,\perp} {\bm{E}_{k,\cur,\perp}}' \|_2
+ b_{term1} + 4 \epsilon
\]
By Lemma \ref{matbnds_del},
\[
\| \bm{E}_{k,\cur} {\bm{E}_{k,\cur}}' {\bm\Psi}_{k} \bm{\Sigma}_\tau {\bm\Psi}_{k} \bm{E}_{k,\cur,\perp} {\bm{E}_{k,\cur,\perp}}' \|_2
\le (r\zeta)^2\lambda^+ + \frac{(r\zeta)^2}{\sqrt{1 - (r\zeta)^2}}\lambda_{k+1}^+
\]
Thus,
\[
\|\frac{1}{{\alpha}} \sum_t \bm{F}_t \|_2
\le
\frac{1}{1-b^2} \left( (r\zeta)^2\lambda^+ + \frac{(r\zeta)^2}{\sqrt{1 - (r\zeta)^2}}\lambda_{k+1}^+ \right)
+ \frac{0.05(r_\new \zeta) b^2 \lambda^-}{(1-b^2)(1-b)^2} + 4\epsilon
\]
with probability at least $ 1-p_{\tilde{\bm{F}}}$, where $p_{\tilde{\bm{F}}} = 4\cdot(2n)\exp\left(\frac{-\alpha\epsilon^2}{32(b_{prob})^2}\right)$, with $ b_{prob} = \frac{r\gamma^2}{(1-b)^2}$.

Combining the bounds on the three terms above, we get the final result of the lemma.
\end{proof}

\section{Automatically setting algorithm parameters and simulation experiments} \label{expts}

\subsection{Automatically setting algorithm parameters}\label{param_set}
The algorithm has five parameters. As explained in \cite{rrpcp_tsp}, one can set $\xi_t= \| \bm\Phi_t \lhat_{t-1}\|_2$. One can either set $\omega_t = 7 \xi_t$ or one can use the average image pixel intensity to set it. In \cite{rrpcp_tsp}, they used $\omega = q\sqrt{\|\mt\|_2^2/n}$ with $q=1$ when it was known that $\|\xt\|_2$ is of the same order as  $\|\lt\|_2$; and $q=0.25$ when $\|\xt\|_2$ was known to be much smaller (the case of foreground moving objects whose intensity is very similar to that of background objects).
There is no good heuristic to pick $\alpha$ except that $\alpha_\add$ should be large enough compared to $r_\new$ and $\alpha_\del$ should be large enough compared to $r$. We used $\alpha=100$ and $K=12$ in our experiments.
We need $K$ to be large enough so that the new subspace is accurately recovered at the end of $K$ projection-PCA iterations. Thus, one way to set $K$ indirectly is as follows: do projection-PCA for at least $K_{\min}$ times, but after that stop when there is not much difference between $\Phat_{j,\new,k}{}' \lhatt$  and $\Phat_{j,\new,k+1}{}' \lhatt$ \cite{rrpcp_perf,rrpcp_tsp}. This, along with imposing an upper bound on $K$ works well in practice \cite{rrpcp_tsp}. We can set $\ghatp$ as suggested in \cite{rrpcp_perf}; by applying any clustering algorithm from literature, e.g., k-means clustering or split-and-merge and then finding the maximum condition number of any cluster. This can be applied to the empirical covariance matrix used in the clustering step of cluster-PCA.


\subsection{Simulated data}
Here we used simulated data to compare performance of PCP \cite{rpca}, mod-PCP \cite{zhan_pcp_jp}, GRASTA \cite{grass_undersampled}, RSL \cite{Torre03aframework} and Automatic ReProCS-cPCA. 
We generated data as explained in Sec. \ref{models}, with $n=256$, $J=3$, $r_0=40$, $t_\train=200$, $t_{\max} = 8200$. We generated $\ell_t$ as in Model \ref{cor_model} and Model \ref{clust_model} with $r_{j,\new}=4$, $r_{j,\old}=4$, $j=1,2,3$, $t_1=700$, $t_2 = 3700$, $t_3 = 6200$, $\vartheta = 3$,  $b=0.1$. The subspace $[\bm{P}_0, \bm{P}_{t_1,\new}, \bm{P}_{t_2,\new}, \bm{P}_{t_3,\new}]$ was generated by orthonormalizing an $n\times(r_0 + r_{1,\new} + r_{2,\new} + r_{3,\new})$ matrix of iid Gaussian entries.
The coefficients $a_{t,*} := \bm{P}_{j,*}^*\nu_{t}$, were generated as follows. They were divided into three clusters. The coefficeints of the first cluster were iid uniformly distributed over $[-100, 100]$, those of the second cluster were iid uniform over $[-10, 10]$, and those of the third cluster were iid uniform over $[-1,1]$. We generated $a_{t,\new} := \bm{P}_{j,\new}^*\nu_t$ iid uniform over $[-1, 1]$ for the first $1700$ time units after the subspace change. After that, it was in one of the three intervals. The sparse matrix $\S$ was generated as in Model \ref{sbyrho} with $s=10$, $\rho = 2$. The support of $x_t$ started from the top, and moved down by $5$ indices every $\beta=25$ time instants. Once it reached the bottom, it started from the top again. We set $(x_t)_i \thicksim \text{Unif}[x_{\min}, 3x_{\min}]$ for all  $i\in \T_t$ with $x_{\min} = 20$.
%
We ran Automatic ReProCS-cPCA with $\alpha=100$, $K=12$, $\xi = \sqrt{r_\new/2}\gamma_\new$, $\omega = (x_{\min} - 14\xi)/2$. We used $\hat{\bm{P}}_0$ for modified-PCP as partial knowledge. We solved PCP and modified-PCP every $200$ frames by using the observations of the last 200 frames as the matrix $\M$.
%
In Fig. \ref{sims}, where the averaged sparse part errors over 50 Monte Carlo simulations are shown, we can see Automatic ReProCS-cPCA outperforms all the other algorithms. We can also see jumps in the Automatic ReProCS-cPCA error at the time instants at which there is a subspace change, and then decays exponentially. This is what is seen from the bounds given in Theorem \ref{thm1_cor} and Corollary \ref{thm1_cor_corol}.

\begin{figure}[t]
\centerline{
\begin{subfigure}{0.3\linewidth} \label{online}
\centering
\includegraphics[width=\linewidth]{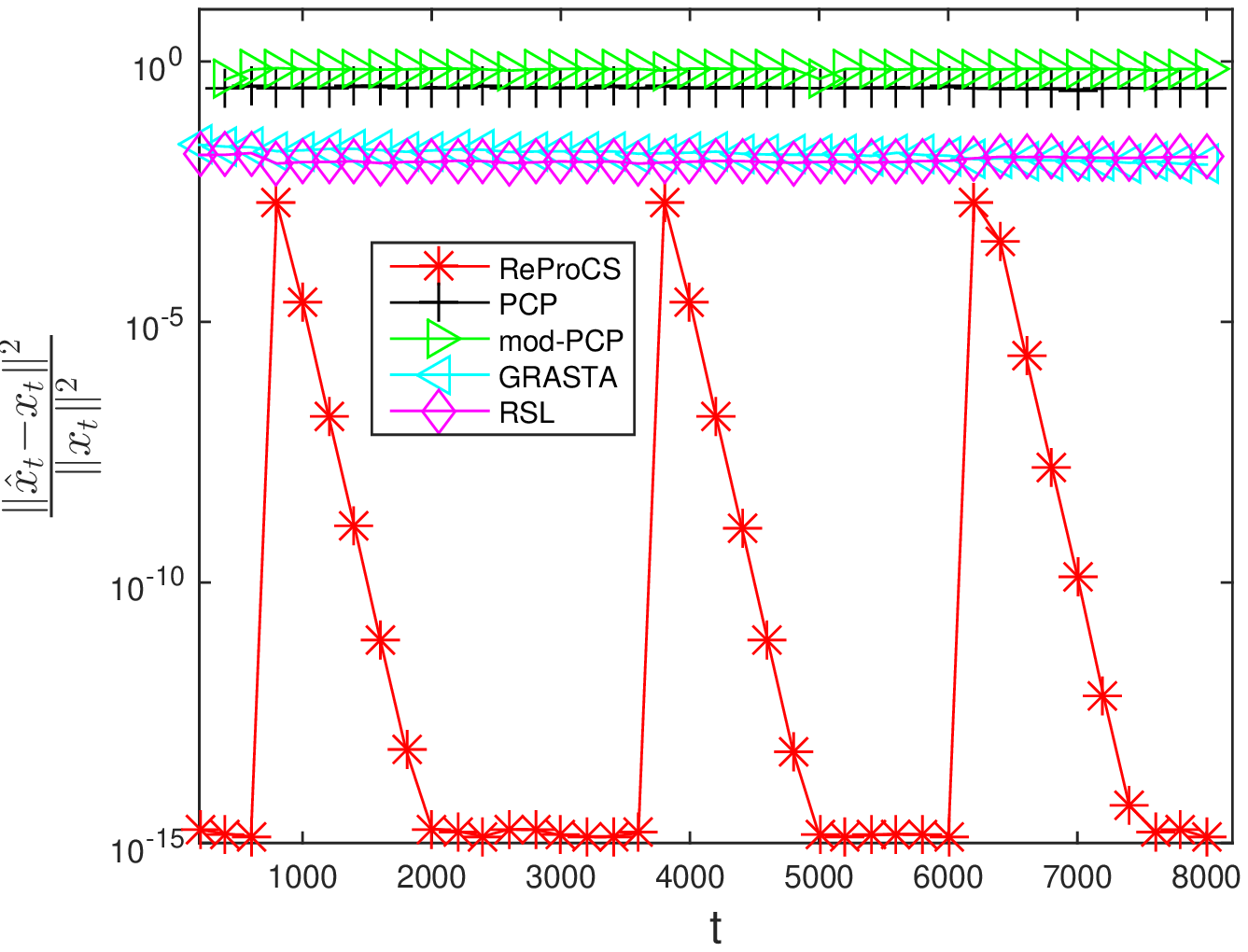}
\caption{\small{fully simulated data}}
\end{subfigure}
\begin{subfigure}{0.3\linewidth} \label{lake_nmse}
\centering
\includegraphics[width=\linewidth]{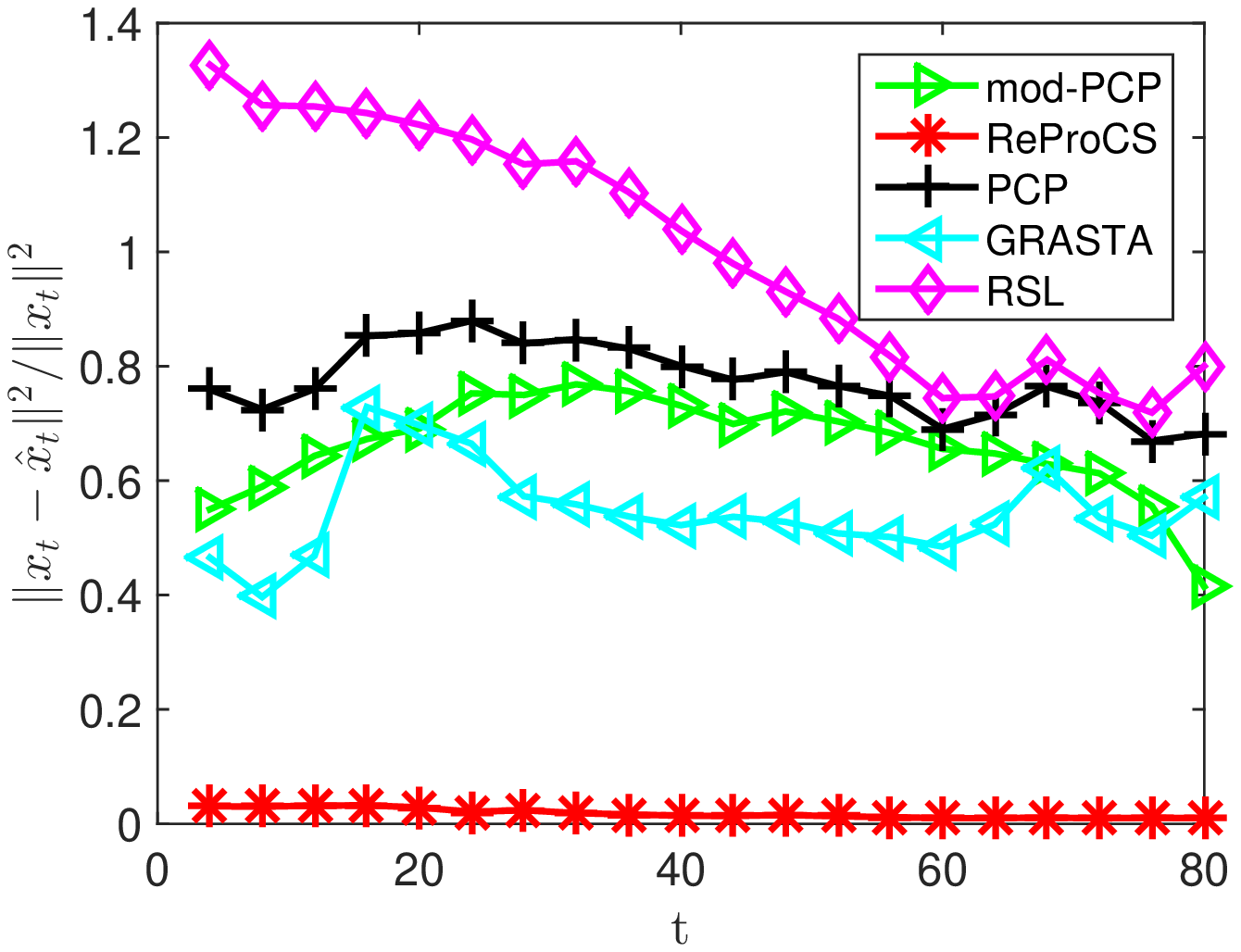}
\caption{\small{lake background}}
\end{subfigure}
}
\vspace{-0.1in}
\caption{\small{Average error comparisons for fully simulated data and for the sequence with the lake background and simulated block object}}
\label{online} \label{sims}
\end{figure}


\subsection{Lake background sequence with simulated foreground}
The lake background sequence used is the same as the one used in \cite{rrpcp_tsp}. The background consisted of a video of moving lake waters. The foreground is a simulated moving rectangular object.
The sequence is of size $72\times 90 \times 1500$, and we used the first $1420$ frames as training data (after subtracting the empirical mean of the training images). The rest 80 frames (after subtracting the same mean image) served as the background $\L$ for the test data. For the first frame of test data, we generated a rectangular foreground support with upper left vertex $(20,5+j_0)$ and lower right vertex $(40+i_1, 30+j_0)$, where $j_0 \sim \text{Unif}[0,30]$ and $i_1 \sim \text{Unif}[0,5]$, and the foreground moves to the right 1 column each time. Then we stacked each image as a long vector $\ell_t$ of size $6480\times1$. For each index $i$ belonging to the support set of foreground $x_t$, we assign $(x_t)_i = 185 - (\ell_t)_i$. We set $\bm{M} = \L + \S$.
For mod-PCP, ReProCS and GRASTA, we used the approach used in \cite{rrpcp_tsp} to estimate the initial background subspace (partial knowledge): do SVD on training data and keep the left singular vectors corresponding to $95\%$ energy as the matrix $\bm{P}_0$. 
The averaged normalized mean squared error (NMSE) of the sparse part over $50$ Monte Carlo realizations is shown in Fig. \ref{sims}. As can be seen, in this case, ReProCS performs the best. In Fig. \ref{lake_view}, we show the lake with simulated foreground at $t=20, 40, 60$, and corresponding foreground and background recovered by different algorithms, and we can see that ReProCS successfully separated foreground and background apart while others did not.

\section{Conclusions} \label{conclude}
In this work, we developed and studied the Automatic ReProCS-cPCA algorithm for incremental or recursive or dynamic or ``online" robust PCA. Our result needed the following assumptions: accurate initial subspace knowledge and a slow subspace change change assumption on the $\lt$'s; the basis vectors for its subspaces are dense (non-sparse) enough; the eigenvalues of the covariance matrix of $\lt$'s are clustered for a certain period of time (this would happen if data has variations across different scales); the outlier support sets $\T_t$ have {\em some} changes over time (as quantified in Model \ref{sbyrho} or Model \ref{general_model}); the square of the smallest outlier magnitude is large enough compared to the energy in the unstructured small noise plus the energy in the changed subspace; and the algorithm parameters are appropriately set.
Ongoing work includes studying the undersampled measurements' case, i.e., the case $\mt = A_t \xt + B_t \lt + \wt$. Besides this, we expect the cluster-PCA algorithm and the proof techniques developed here to apply to various other problems involving PCA with data and noise terms being correlated.

\appendices
\renewcommand{\thetheorem}{\thesection.\arabic{theorem}}

\section{Preliminaries} \label{prelim}

\begin{lem} \cite[Lemma 2.10]{rrpcp_perf}\label{lemma0}\label{hatswitch}
Suppose that $\bm{P}$, $\Phat$ and $\bm{Q}$ are three basis matrices. Also, $\bm{P}$ and $\Phat$ are of the same size, $\bm{Q}' \bm{P} = \bm{0}$ and $\|(\I-\Phat \Phat{}' ) \bm{P} \|_2 = \zeta_*$. Then,
\begin{enumerate}
  \item $\|(\I-\Phat\Phat{}')\bm{P}\bm{P}'\|_2 =\|( \I - \bm{P}\bm{P}' ) \Phat \Phat{}'\|_2 =  \|( \I - \bm{P} \bm{P}' ) \Phat \|_2 = \| ( \I - \Phat \Phat{}' ) \bm{P}\|_2 =  \zeta_*$
  \item $\|\bm{P} \bm{P}' - \Phat \Phat{}'\|_2 \leq 2 \|(\I-\Phat \Phat{}')\bm{P}\|_2 = 2 \zeta_*$
  \item $\|\Phat{}' \bm{Q}\|_2 \leq \zeta_*$ \label{lem_cross}
  \item $ \sqrt{1-\zeta_*^2} \leq \sigma_i\left((\I-\Phat \Phat{}')\bm{Q}\right)\leq 1 $
\end{enumerate}
\end{lem}

Weyl's inequality \cite{hornjohnson} (simplified version) states the following
\begin{theorem}
Given two Hermitian matrices $\mathcal{A}$ and $\mathcal{H}$,
\[
\lambda_i(\mathcal{A}) - \|\mathcal{H}\|_2 \le \lambda_i(\mathcal{A}) - \lambda_{\min}(\mathcal{H}) \le
\lambda_i(\mathcal{A} + \mathcal{H}) \le \lambda_i(\mathcal{A}) +  \lambda_{\max}(\mathcal{H}) \le  \lambda_i(\mathcal{A}) + \|\mathcal{H}\|_2
\]
\end{theorem}

Davis and Kahan's $\sin \theta$ theorem \cite{davis_kahan} studies the rotation of eigenvectors by perturbation.

\begin{theorem}[$\sin \theta$ theorem \cite{davis_kahan}] \label{sintheta}
Given two Hermitian matrices $\mathcal{A}$ and $\mathcal{H}$ and suppose that $\mathcal{A}$ satisfies
\begin{align}
 \label{sindecomp}
\mathcal{A} &= \left[ \begin{array}{cc} E & E_{\perp} \\ \end{array} \right]
\left[\begin{array}{cc} A\ & 0\ \\ 0 \ & A_{\perp} \\ \end{array} \right]
\left[ \begin{array}{c} E' \\ {E_{\perp}}' \\ \end{array} \right] \nn
\end{align}
where $[E \ E_{\perp}]$ is an orthonormal matrix. Suppose that $\mathcal{A}+\mathcal{H}$ can be decomposed as
\begin{align*}
\mathcal{A} + \mathcal{H}
&= \left[ \begin{array}{cc} F & F_{\perp} \\ \end{array} \right]
\left[\begin{array}{cc} \Lambda\ & 0\ \\ 0 \ & \Lambda_{\perp} \\ \end{array} \right]
\left[ \begin{array}{c} F' \\ {F_{\perp}}' \\ \end{array} \right] \nn
\end{align*}
where $[F\ F_{\perp}]$ is another orthonormal matrix and is such that $\rank(F)=\rank(E)$. Let $\mathcal{R} := (\mathcal{A}+\mathcal{H}) E - \mathcal{A}E = \mathcal{H} E $. If $ \lambda_{\min}(A) >\lambda_{\max}(\Lambda_{\perp})$, then
\beq
\|(I-F F')E \|_2 \leq \frac{\|\mathcal{R}\|_2}{\lambda_{\min}(A) - \lambda_{\max}(\Lambda_{\perp})} \le \frac{\|\mathcal{H}\|_2}{\lambda_{\min}(A) - \lambda_{\max}(\Lambda_{\perp})}.
\nn
\eeq
\end{theorem}

\begin{remark}
In the above theorem, let $r = \rank(F)$. If the decomposition of $\mathcal{A} + \mathcal{H}$ is obtained by EVD, then $\lambda_{\max}(\Lambda_{\perp}) = \lambda_{r+1}(\mathcal{A} + \mathcal{H}) \le
\lambda_{r+1}(\mathcal{A}) + \|\mathcal{H}\|_2$. The inequality follows using Weyl. Moreover, if $\lambda_{\min}(A) > \lambda_{\max}(A_\perp)$, then $\lambda_{r+1}(\mathcal{A}) = \lambda_{\max}(A_\perp)$. 
Thus a useful corollary of the above result is the following. If $ \lambda_{\min}(A)  - \lambda_{\max}(A_\perp) - \|\mathcal{H}\|_2 > 0$, then
\[
\|(I-F F')E \|_2 \leq \frac{\|\mathcal{H}\|_2}{\lambda_{\min}(A) - \lambda_{\max}(A_\perp) - \|\mathcal{H}\|_2}.
\]
\end{remark}

\begin{lem}[Cauchy-Schwarz for a sum of vectors]\label{CSsum}
For vectors $\bm{x}_t$ and $\bm{y}_t$,
\[
\left(\sum_{t=1}^{\alpha} {\bm{x}_t}'\bm{y}_t\right)^2 \leq \left( \sum_t \|\bm{x}_t\|_2^2 \right) \left( \sum_t \|\bm{y}_t\|_2^2 \right)
\]
\end{lem}


\begin{lem}[Cauchy-Schwarz for a sum of matrices]\label{CSmat}
For matrices $\bm{X}_t$ and $\bm{Y}_t$,
\[
\left\|\frac{1}{\alpha} \sum_{t=1}^{\alpha} \bm{X}_t {\bm{Y}_t}'\right\|_2^2 \leq \lambda_{\max}\left(\frac{1}{\alpha} \sum_{t=1}^{\alpha} \bm{X}_t {\bm{X}_t}'\right)
\lambda_{\max}\left(\frac{1}{\alpha} \sum_{t=1}^{\alpha} \bm{Y}_t {\bm{Y}_t}'\right)
\]
\end{lem}

\begin{proof}[Proof of Lemma \ref{CSmat}]
\begin{align*}
\left\| \sum_{t=1}^{\alpha} \bm{X}_t {\bm{Y}_t}'\right\|_2^2 &=
\max_{\substack{\|\bm{x}\|=1\\\|\bm{y}\|=1}} \left| \bm{x}'\left(\sum_{t}\bm{X}_t {\bm{Y}_t}'\right)\bm{y} \right|^2 \\
&= \max_{\substack{\|\bm{x}\|=1\\\|\bm{y}\|=1}} \left| \sum_{t=1}^{\alpha} ({\bm{X}_t}'\bm{x})'({\bm{Y}_t}'\bm{y})  \right|^2 \\
&\leq \max_{\substack{\|\bm{x}\|=1\\\|\bm{y}\|=1}} \left( \sum_{t=1}^{\alpha} \left\| {\bm{X}_t}'\bm{x} \right\|_2^2\right) \left( \sum_{t=1}^{\alpha} \left\| {\bm{Y}_t}'\bm{y} \right\|_2^2\right) \\
&= \max_{\|\bm{x}\|=1}  \bm{x}' \sum_{t=1}^{\alpha} \bm{X}_t{\bm{X}_t}' \ \bm{x}  \ \cdot \ \max_{\|\bm{y}\|=1} \bm{y}' \sum_{t=1}^{\alpha} \bm{Y}_t{\bm{Y}_t}' \ \bm{y}\\
&= \lambda_{\max}\left( \sum_{t=1}^{\alpha} \bm{X}_t{\bm{X}_t}' \right)\lambda_{\max}\left( \sum_{t=1}^{\alpha} \bm{Y}_t{\bm{Y}_t}' \right)
\end{align*}
The inequality is by Lemma \ref{CSsum}. The penultimate line is because $\|\bm{x}\|_2^2 = {\bm{x}'\bm{x}}$.
Multiplying both sides by $\left(\frac{1}{\alpha}\right)^2$ gives the desired result.
\end{proof}

\begin{lem}[Exchanging the order of a double sum]\label{sumswitch_0}
\[
\sum_{t=0}^{\alpha-1} \sum_{\tau=0}^{t} f_{t,\tau} = \sum_{\tau=0}^{\alpha-1} \sum_{t=\tau}^{\alpha-1} f_{t,\tau}
\]
\end{lem}

\begin{proof}
Define
$[\mathrm{statement}]$ to be the Boolean value of $\mathrm{statement}$
\begin{align*}
\sum_{t=0}^{\alpha-1} \sum_{\tau=0}^{t} f_{t,\tau} &= \sum_{t,\tau}[0\leq \tau \leq t][0\leq t \leq \alpha-1] f_{t,\tau} \\
&= \sum_{t,\tau} [0 \leq \tau \leq t \leq \alpha-1] f_{t,\tau} \\
&= \sum_{t,\tau} [0 \leq \tau \leq \alpha -1][ \tau \leq t \leq \alpha-1] f_{t,\tau} \\
&= \sum_{\tau=0}^{\alpha-1} \sum_{t=\tau}^{\alpha-1} f_{t,\tau}
\end{align*}
\end{proof}

The following lemma follows in an exactly analogous fashion.
\begin{lem}[Exchanging the order of a double sum]\label{sumswitch}
\[
\sum_{t=t_0}^{t_0+\alpha-1} \sum_{\tau=t_0}^{t} f_{t,\tau} = \sum_{\tau=t_0}^{t_0+\alpha-1} \sum_{t=\tau}^{t_0+\alpha-1} f_{t,\tau}
\]
\end{lem}

\begin{lem}[A summation used very often] \label{sum_often}
We have
\[
\frac{1}{\alpha} \sum_{t=t_0}^{t_0+\alpha-1} \sum_{\tau=t_0}^t  b^{2(t-\tau)} =  \frac{1}{1-b^2}(1 - \frac{1}{\alpha} \frac{b^2(1 - b^{2 \alpha})}{1-b^2})
\]
Thus
\[
 \frac{1}{1-b^2}(1 - \frac{1}{\alpha} \frac{b^2}{1-b^2}) \le  \frac{1}{\alpha} \sum_{t=t_0}^{t_0+\alpha-1} \sum_{\tau=t_0}^t  b^{2(t-\tau)} \le \frac{1}{1-b^2}
\]
\end{lem}
Proof: $\sum_{\tau=t_0}^t  b^{2(t-\tau)} = \frac{1}{1-b^2}(1-b^{2(t-t_0+1)})$. And $\sum_{t=t_0}^{t_0+\alpha-1}  \frac{1}{1-b^2}(1-b^{2(t-t_0+1)}) = \frac{1}{1-b^2}(\alpha- \frac{b^2(1-b^{2\alpha})}{1-b^2})$

\begin{lem}\label{ind_exp}
Let $W$, $Y$, and $Z$ be random variables.  Assume that $W$ is independent of $\{Y,Z\}$.  Then
\[
\E[WY|Z] = \E[W]\E[Y|Z]
\]
\end{lem}
\begin{proof}
By the chain rule, $f_{W,Y|Z}(w,y|z) = f_{W|Y,Z}(w|y,z) f_{Y|Z}(y|z)$.  Because $W$ is independent of both $Y$ and $Z$, $f_{W|Y,Z}(w|y,z) = f_W(w)$.
\end{proof}

\begin{lem}\label{event}
For an event $\mathcal{E}$ and random variable $X$, $\mathbb{P}(\mathcal{E}|X) \geq p$ for all $X \in \mathcal{C}$ implies that $\mathbb{P}(\mathcal{E}|X\in \mathcal{C}) \geq p$.
\end{lem}

\begin{theorem}[Matrix Azuma]\cite[Theorem 7.1]{tail_bound}\label{azuma}
Consider a finite adapted sequence $\bm{Z}_t$, $t=1,2,\dots \alpha$, of $n\times n$ Hermitian matrices, and a fixed sequence $\bm{A}_t$ of Hermitian matrices that satisfy
\[
\E[\bm{Z}_t|\bm{Z}_1, \bm{Z}_2, \dots, \bm{Z}_{t-1}] = \bm{0} \quad \text{and} \quad {\bm{Z}_t}^2 \preceq {\bm{A}_t}^2 \quad \text{with probability 1.}
\]
Define the variance parameter
\[
sigma^2 := \Big\| \sum_t {\bm{A}_t}^2 \Big\|_2.
\]
Then, for all $\epsilon > 0$,
\[
\Pr\left( \lambda_{\max}\left(\sum_t \bm{Z}_t \right) \geq \epsilon \right) \leq n\exp\left(\frac{-\epsilon^2}{8\sigma^2}\right)
\]
\end{theorem}

The following corollary extends the above result to the case where the conditional expectation is not zero and when we also condition on another random variable.

\begin{corollary}[Matrix Azuma conditioned on another random variable for a nonzero mean Hermitian matrix]\label{azuma_hermitian}
Consider an $\alpha$-length sequence $\{\bm{Z}_t\}_{t=1,2,\dots, \alpha}$ of random Hermitian matrices of size $n\times n$ and a random variable $X$ that we condition on. Assume that, for all $X\in\mathcal{C}$, (i) $\Pr(b_1 \I \preceq \bm{Z}_t \preceq b_2 \I|X) = 1$, for $1\leq t\leq \alpha$ and (ii) $b_3 \I \preceq \frac{1}{\alpha}\sum_{t=1}^{\alpha} \E[\bm{Z}_t| \bm{Z}_1,\bm{Z}_2, \dots, \bm{Z}_{t-1},X] \preceq b_4 \I $. Then for all $\epsilon > 0$,
\begin{align*}
&\Pr \left(\lambda_{\max}\left(\frac{1}{\alpha}\sum_{t=1}^{\alpha} \bm{Z}_t  \right) \leq b_4 + \epsilon \Big| X \right) \geq 1- n \exp\left(\frac{-\alpha \epsilon^2}{8(b_2-b_1)^2}\right) \\
&\Pr \left(\lambda_{\min}\left(\frac{1}{\alpha}\sum_{t=1}^{\alpha} \bm{Z}_t  \right) \geq b_3 -\epsilon \Big| X \right) \geq  1- n \exp\left(\frac{-\alpha \epsilon^2}{8(b_2-b_1)^2}\right)
\end{align*}
\end{corollary}

\begin{proof}
At certain places, where the meaning is clear, we use $\E_{t-1}[\bm{Z}_t|X]$ to refer to $\E[\bm{Z}_t| \bm{Z}_1,\bm{Z}_2, \dots, \bm{Z}_{t-1},X]$
\begin{enumerate}
\item
Let $\bm{Y}_t := \bm{Z}_t - \E_{t-1}(\bm{Z}_t|X)$.
Clearly $\E_{t-1}(\bm{Y}_t|X) = \bm{0}$. Since for all $X \in \calc$, $\Pr(b_1 \I \preceq \bm{Z}_t \preceq b_2 \I|X)=1$ and since for an Hermitian matrix, $\lambda_{\max}(.)$ is a convex function, and $\lambda_{\min}(.)$ is a concave function, $b_1 \I \preceq \E_{t-1}(\bm{Z}_t|X) \preceq b_2 \I$ for all $X \in \calc$. Therefore, $\Pr({\bm{Y}_t}^2 \preceq (b_2 -b_1)^2 \I|X) = 1$ for all $X \in \calc$.
Thus, for Theorem \ref{azuma}, $\sigma^2 = \|\sum_{t=1}^{\alpha} (b_2 - b_1)^2 \I\|_2 = \alpha (b_2-b_1)^2$. For any $X \in \calc$, applying Theorem \ref{azuma} for $\{\bm{Y}_t\}_{t = 1,\dots,\alpha}$ conditioned on $X$, we get that, for any $\epsilon > 0$,
\[
\Pr\left( \lambda_{\max} \left(\frac{1}{\alpha} \sum_{t=1}^{\alpha} \bm{Y}_t \right) \leq \epsilon\Big|X\right) > 1- n\exp\left(\frac{-\alpha \epsilon^2}{8  (b_2 -b_1)^2}\right) \ \text{for all} \ X \in \calc
\]
By Weyl's inequality, $\lambda_{\max} (\frac{1}{\alpha} \sum_{t=1}^{\alpha} \bm{Y}_t) = \lambda_{\max} (\frac{1}{\alpha} \sum_{t=1}^{\alpha} (\bm{Z}_t - \E_{t-1}(\bm{Z}_t|X))
\geq \lambda_{\max} (\frac{1}{\alpha} \sum_{t=1}^{\alpha} \bm{Z}_t) +  \lambda_{\min} (\frac{1}{\alpha} \sum_{t=1}^{\alpha} - \E_{t-1}(\bm{Z}_t|X))$.

 Since $\lambda_{\min} (\frac{1}{\alpha} \sum_{t=1}^{\alpha} -\E_{t-1}(\bm{Z}_t|X)) = - \lambda_{\max} (\frac{1}{\alpha} \sum_{t=1}^{\alpha} \E_{t-1}(\bm{Z}_t|X))\geq -b_4$, thus $ \lambda_{\max} (\frac{1}{\alpha} \sum_{t=1}^{\alpha} \bm{Y}_t) \geq \lambda_{\max} (\frac{1}{\alpha} \sum_{t=1}^{\alpha} \bm{Z}_t) - b_4$.
Therefore,
\[
\Pr\left( \lambda_{\max} \left(\frac{1}{\alpha} \sum_{t=1}^{\alpha} \bm{Z}_t\right) \leq b_4 + \epsilon \Big| X \right) > 1- n\exp\left(\frac{-\alpha \epsilon^2}{8  (b_2 -b_1)^2}\right) \ \text{for all} \ X \in \calc
\]

\item Now let $\bm{Y}_t = \E_{t-1}(\bm{Z}_t|X) - \bm{Z}_t$. As before, $\E_{t-1}(\bm{Y}_t|X) = 0$ and conditioned on any $X \in {\cal C}$, $\mathbf{P}(\bm{Y}_t^2 \preceq (b_2 -b_1)^2 \I|X) = 1$.  As before, applying Theorem \ref{azuma}, we get that for any $\epsilon >0$,
\[
\Pr \left( \lambda_{\max} \left(\frac{1}{\alpha} \sum_{t=1}^{\alpha} \bm{Y}_t\right) \leq \epsilon \Big| X \right) > 1- n\exp\left(\frac{-\alpha \epsilon^2}{8  (b_2 -b_1)^2}\right) \ \text{for all} \ X \in \calc
\]
By Weyl's inequality, $\lambda_{\max}(\frac{1}{\alpha}\sum_{t=1}^{\alpha} \bm{Y}_t) = \lambda_{\max}(\frac{1}{\alpha} \sum_{t=1}^{\alpha}(\E_{t-1}(\bm{Z}_t|X) - \bm{Z}_t)) \geq \lambda_{\min} (\frac{1}{\alpha} \sum_{t=1}^{\alpha} \E_{t-1}(\bm{Z}_t|X)) + \lambda_{\max} (\frac{1}{\alpha} \sum_{t=1}^{\alpha} - \bm{Z}_t) = \lambda_{\min} (\frac{1}{\alpha} \sum_{t=1}^{\alpha} \E_{t-1}(\bm{Z}_t|X)) - \lambda_{\min} (\frac{1}{\alpha} \sum_{t=1}^{\alpha} \bm{Z}_t) \geq b_3 - \lambda_{\min} (\frac{1}{\alpha} \sum_{t=1}^{\alpha} \bm{Z}_t)$
Therefore, for any $\epsilon >0$,
\[
\Pr \left(\lambda_{\min}\left(\frac{1}{\alpha}\sum_{t=1}^{\alpha} \bm{Z}_t\right) \geq b_3 -\epsilon \Big| X \right) \geq  1- n \exp\left(\frac{-\alpha \epsilon^2}{8(b_2-b_1)^2}\right) \ \text{for all} \ X \in \calc
\]
\end{enumerate}
\end{proof}

We can further extend this to the case of a matrix which is not necessarily Hermitian.

\begin{corollary}[Matrix Azuma conditioned on another random variable for an arbitrary matrix]\label{azuma_norm}
Consider an $\alpha$-length adapted sequence $\{\bm{Z}_t\}_{t=1,2,\dots, \alpha}$ of random matrices of size $n_1\times n_2$ and a random variable $X$ that we condition on. Assume that, for all $X\in\mathcal{C}$, (i) $\Pr(\|\bm{Z}_t\|_2 \le b_1|X) = 1$ and (ii) $\|\frac{1}{\alpha}\sum_{t=1}^\alpha \E[\bm{Z}_t| \bm{Z}_1,\bm{Z}_2, \dots, \bm{Z}_{t-1},X] \|_2 \leq b_2$. Then, for all $\epsilon >0$,
\[
\Pr \left(\Big\|\frac{1}{\alpha}\sum_{t=1}^{\alpha} \bm{Z}_t \Big\|_2 \leq b_2 + \epsilon \Big| X \right) \geq 1-(n_1+n_2) \exp\left(\frac{-\alpha \epsilon^2}{8 (2b_1)^2}\right)
\]
\end{corollary}

\begin{proof}
At certain places, where the meaning is clear, we use $\E_{t-1}[\bm{Z}_t|X]$ to refer to $\E[\bm{Z}_t| \bm{Z}_1,\bm{Z}_2, \dots, \bm{Z}_{t-1},X]$

Define the dilation of an $n_1 \times n_2$ matrix $\bm{M}$ as $\operatorname{dilation} (\bm{M}) := \left[\begin{array}{cc} \bm{0} & {\bm{M}}' \\ \bm{M} & \bm{0} \\\end{array} \right]$. Notice that this is an $(n_1+n_2) \times (n_1 +n_2)$ Hermitian matrix \cite{tail_bound} . As shown in \cite[equation 2.12]{tail_bound},
\begin{align}\label{dilM}
\lambda_{\max}\big(\operatorname{dilation}(\bm{M})\big) = \|\operatorname{dilation} (\bm{M})\|_2 = \|\bm{M}\|_2
\end{align}
Thus, the corollary assumptions imply that $\mathbf{P}(\|\operatorname{dilation} (\bm{Z}_t)\|_2 \leq b_1 |X) = 1$ for all $X \in \calc$.
By \eqref{dilM} and the definition of $\operatorname{dilation}$,
\[
\frac{1}{\alpha}\sum_t \E_{t-1}[\operatorname{dilation}(\bm{Z}_t ) | X] =
\operatorname{dilation}\left(  \frac{1}{\alpha}\sum_t \E_{t-1}[\bm{Z}_t | X] \right) \preceq
b_2\I
\]
Thus, applying Corollary \ref{azuma_hermitian} to the sequence $\{\operatorname{dilation} (\bm{Z}_t)\}_{t=1,\dots,\alpha}$, we get that,
\[
\Pr \left(\lambda_{\max}\left(\frac{1}{\alpha}\sum_{t=1}^{\alpha} \operatorname{dilation}(\bm{Z}_t)\right)\leq b_2 + \epsilon \Big| X\right) \geq 1- (n_1+n_2) \exp\left(\frac{-\alpha \epsilon^2}{32 b_1^2}\right) \ \text{for all} \ X \in \calc
\]
Using \eqref{dilM}, $\lambda_{\max}\big(\frac{1}{\alpha}\sum_{t=1}^{\alpha} \operatorname{dilation}(\bm{Z}_t) \big) = \lambda_{\max}\big(\operatorname{dilation}(\frac{1}{\alpha}\sum_{t=1}^{\alpha} \bm{Z}_t)\big)  = \|\frac{1}{\alpha}\sum_{t=1}^{\alpha} \bm{Z}_t\|_2$ gives the final result.
\end{proof}

\section{Proof of Lemma \ref{initsub_cor} (initial subspace is accurately recovered)} \label{proof_initsub_cor}
\begin{proof}[Proof of Lemma \ref{initsub_cor}]
Define $\mathcal{M}:= \frac{1}{t_{\train}}\sum_{t=1}^{t_{\train}} \mt{\mt}'$,  $\mathcal{A}:=\frac{1}{t_{\train}}\sum_{t=1}^{t_{\train}} \lt{\lt}'$
and
$\mathrm{perturb}:=  \mathcal{M} - \mathcal{A}$.

Using Theorem \ref{sintheta} (sin theta theorem) followed by Weyl's inequality for $\lambda_{\max}(\Lambda_{\perp})= \lambda_{\max}(\mathcal{M})$, if $\lambda_{r_0}(\mathcal{A}) - \lambda_{r_0+1}(\mathcal{A}) - \| \mathrm{perturb} \|> 0$, then
\begin{equation}\label{se_bnd}
\mathrm{dif}(\Phat_\train, \P_\train)
\le \frac{\| \mathrm{perturb} \|_2}{\lambda_{r_0}(\mathcal{A}) - \lambda_{r_0+1}(\mathcal{A}) - \| \mathrm{perturb} \|_2}
\end{equation}
We will use Azuma to lower and upper bound $\lambda_{r_0}(\mathcal{A})$, to upper bound $\lambda_{r_0+1}(\mathcal{A})$ and to upper bound $\|perturb\|_2$.
Let
\[
\epsilon = \frac{1}{1-b^2}0.001 r_\new \zeta \lambda^-
\]
To get the first three bounds, we need to bound $\lambda_{\max}(\mathcal{A}-\frac{1}{t_{\train}} \sum_{t=1}^{t_\train} \sum_{\tau=0}^t  b^{2(t-\tau)} \bm{\Sigma}_\tau)$ and then use Weyl's inequality.
Now $\mathcal{A} = \frac{1}{t_{\train}} \sum_{t=1}^{t_\train} \sum_{\tau=0}^t \sum_{\ttau=0}^t  b^{2t-\tau-\ttau} \bnu_\tau \bnu_\ttau{}'$. We proceed as in Section \ref{general_decomp} but with the difference that we include $-\frac{1}{t_{\train}} \sum_{t=1}^{t_\train} \sum_{\tau=0}^t  b^{2(t-\tau)} \bm{\Sigma}_\tau$ into $\mathrm{term21}$. Another difference is that $t_0=1$ and so $\l_{t_0-1} = 0$ (and so $\mathrm{term1}=0$ and $\mathrm{term3}=0$). Thus we get
\[
-3\epsilon \le \lambda_{\max}(\mathcal{A}-\frac{1}{t_{\train}} \sum_{t=1}^{t_\train} \sum_{\tau=0}^t  b^{2(t-\tau)} \bm{\Sigma}_\tau) \le 3\epsilon
\]
with probability $1 - 3\cdot(2n) \exp\left(\frac{-t_\train \epsilon^2 (1-b^2)^2(1-b)^2}{32 (2r \gamma^2)^2}\right)$.
Thus, with the above probability, using Weyl's inequality and Lemma \ref{sum_often},
\begin{align*}
\lambda_{r_0}(\mathcal{A}) & \ge \lambda_{r_0}(\frac{1}{t_{\train}} \sum_{t=1}^{t_\train} \sum_{\tau=0}^t  b^{2(t-\tau)} \bm{\Sigma}_\tau) - 3 \epsilon \ge \frac{1}{1-b^2}(1 - \frac{b^2}{t_{\train}(1-b^2)}) \lambda^- - 3\epsilon \\
\lambda_{r_0}(\mathcal{A})  &  \le \lambda_{r_0}(\frac{1}{t_{\train}} \sum_{t=1}^{t_\train} \sum_{\tau=0}^t  b^{2(t-\tau)} \bm{\Sigma}_\tau) + 3 \epsilon \le \frac{1}{1-b^2} \lambda^- +  3\epsilon \\
\lambda_{r_0+1}(\mathcal{A}) &  \le \lambda_{r_0+1}(\frac{1}{t_{\train}} \sum_{t=1}^{t_\train} \sum_{\tau=0}^t  b^{2(t-\tau)} \bm{\Sigma}_\tau) + 3 \epsilon  = 0 + 3\epsilon
\end{align*}
(the above follows because $\bm{\Sigma}_\tau$ has rank $r_0$ for all $t \le t_\train$).
Next consider $\|\mathrm{perturb}\|_2$. It is easy to see that
\[
\|\mathrm{perturb}\|_2 \le 2\|\frac{1}{t_\train} \sum_t \lt \wt'\| +  \|\frac{1}{t_\train} \sum_t \wt \wt'\|_2
\]
Proceeding as in Section \ref{general_decomp_w} for the first term and using the deterministic bound of $0.03 \zeta \lambda^-$ for the second term, we get
\[
\|\mathrm{perturb}\|_2 \le 0.03 \zeta \lambda^- + 2 \epsilon
\]
with probability $1 - (2n) \exp\left(\frac{-t_\train \epsilon^2 (1-b)^2}{32 \cdot r \gamma^2 \epsilon_w^2}\right)$.
Using the above bounds and Weyl's inequality, we can conclude that
\begin{align*}
\lamtrain:=\lambda_{r_0}(\mathcal{M}) & \le \lambda_{r_0}(\mathcal{A}) + \|\mathrm{perturb}\|_2 \le \frac{1}{1-b^2} \lambda^- + 0.08 r \zeta \lambda^- \\
\lamtrain:=\lambda_{r_0}(\mathcal{M}) & \ge \lambda_{r_0}(\mathcal{A}) - \|\mathrm{perturb}\|_2 \ge \frac{1}{1-b^2}(1 - \frac{b^2}{t_{\train}(1-b^2)}) \lambda^- - 0.08r \zeta \lambda^-
\end{align*}
w.p. at least $1 - 3\cdot(2n) \exp\left(\frac{-t_\train \epsilon^2 (1-b^2)^2(1-b)^2}{32(2r \gamma^2)^2}\right) - (2n) \exp\left(\frac{-t_\train \epsilon^2 (1-b)^2}{32r \gamma^2 \epsilon_w^2}\right) \ge
1 - 4 \cdot (2n) \exp\left(\frac{-t_\train \epsilon^2 (1-b^2)^2(1-b)^2}{32 (2r \gamma^2)^2}\right) \ge 1- n^{-10}$.
The last inequality follows because $t_\train \ge \frac{128 (r \gamma^2)^2}{ (1-b)^2 (0.001 r_\new\zeta \lambda^-)^2} (11 \log n + \log 8)$.

Thus, using the fact that $1/t_\train < (r\zeta)^2$, $(1 - \frac{b^2}{t_{\train}(1-b^2)}) \ge (1 - \frac{(r\zeta)^2 b^2}{(1-b^2)})$ and so we get: with probability at least $1- n^{-10}$,
\ben
\item  $\lamtrain \le  \frac{1}{1-b^2} \lambda^- + 0.08r \zeta \lambda^-  < 1.2 \frac{\lambda^-}{1-b^2}$

\item $\lamtrain \ge \frac{1}{1-b^2}(1 - \frac{(r\zeta)^2 b^2}{(1-b^2)}) \lambda^- - 0.08r \zeta \lambda^-  \ge 0.8 \frac{\lambda^-}{1-b^2}$
\item and
\[
\mathrm{dif}(\Phat_\train, \P_\train) \le \frac{0.03 \zeta \lambda^- + 2 \epsilon}{\frac{1}{1-b^2}(1 - \frac{b^2}{t_{\train}(1-b^2)}) \lambda^- - 0.08r \zeta \lambda^-} \le  0.031r_\new\zeta \le r_0 \zeta
\]
\een
\end{proof}

\section{Proof of Lemma \ref{zetadecay} (bounds on $\zeta_{j,\new,k}^+$ and $\tilde\zeta_{j,k}^+$)} 
\label{zetatil_bnd}

\begin{proof}[Proof of Lemma \ref{zetadecay}]
%

Proof of item 1 of the lemma: This follows directly from the bounds for $b_{\bm{A}}$, $b_{\bm{A},\perp}$, $b_{\bm{\mathcal{H}},k}$ in Fact \ref{ppca_bound}, and by using Lemma \ref{initsub_cor}.

Proof of item 2 of the lemma:
Recall that  $\ds\zeta_{j,\new,k}^+ := \frac{b_{\bm{\mathcal{H}},k}}{b_{\bm{A}} - b_{\bm{A},\perp} - b_{\bm{\mathcal{H}},k}} $ with the terms on the RHS defined in Lemmas \ref{Ak_cor}, \ref{Akperp_cor}, \ref{calHk_cor}.
The proof approach is similar to that of \cite[Lemma 6.1]{rrpcp_perf} and almost exactly the same as that of \cite[Lemma 6.14]{rrpcp_isit15}.
The proof is as follows. With the bound in Fact \ref{ppca_bound}, and since $\zeta_{\new, k}^+$ is an increasing function of $b_{\bm{A},\perp}$ and $b_{\bm{\mathcal{H}},k}$, and a decreasing function of $b_{\bm{A}}$, we have
\[
\zeta_{\new,1}^+ \leq \frac{0.156 + 0.1r_\new\zeta}{0.9999 - 0.005r_\new\zeta - (0.156 + 0.12r_\new\zeta)} < 0.19 \ \text{ because $r_\new\zeta \leq 10^{-4}$.}
\]
For $k\geq 2$, we have
\[
\zeta_{\new,k}^+ \leq \frac{0.073\zeta_{\new, k-1}^+ + 0.1r_\new\zeta}{0.9999 - 0.005r_\new\zeta - (0.073\zeta_{\new, k-1}^+ + 0.12r_\new\zeta)}
\]
Using $\zeta_{\new,1}^+  \le 0.19$ and $r_\new\zeta \leq 10^{-4}$ in the above, it is easy to see that $\zeta_{\new,k}^+ \le 0.19$ for all $k\ge1$. Using this and $r_\new\zeta \leq 10^{-4}$ to bound the denominator,  we can get
\[
\zeta_{\new, k}^+ \leq 0.1\zeta_{\new,k-1}^+ + 0.15r_\new\zeta.
\]

Proof of item 3 of the lemma:
Recall that
\[
\tilde{\zeta}_k^+ := \frac{b_{\tilde{\bm{\mathcal{H}}},k}}{b_{\tilde{\bm{A}},k}-b_{\tilde{\bm{A}},k,\perp}-b_{\tilde{\bm{\mathcal{H}}},k}}.
\]
Substituting in the bounds for $b_{\tilde{\bm{\mathcal{H}}},k}$, $b_{\tilde{\bm{A}},k}$, and $b_{\tilde{\bm{A}},k,\perp}$ from Fact \ref{del_bounds} gives
\[
\tilde{\zeta}_k^+ \leq \frac{0.072(r+r_\new)\zeta + 0.19r_\new\zeta}{0.9999 - (0.2 + 0.265r_\new\zeta + 0.072(r+r_\new)\zeta)} \leq 0.09(r+r_\new)\zeta + 0.119r_\new\zeta < 0.15(r+r_\new)\zeta
\]
where we use $r_\new < r$ to get the last inequality.
Using the theorem's assumption $r_{j,k}:= |\mathcal{G}_{j,k}| \ge 0.15 (r+r_\new)$, the claim follows.
\end{proof}

\section{Proof of Lemma \ref{cslem_cor} (Compressed Sensing Lemma)} \label{pf_cslem}
This proof's approach is similar to that of \cite[Lemma 6.4]{rrpcp_perf}. 
The proof uses the denseness assumption and subspace error bounds $\zeta_{j,*} \leq \zeta_{j,*}^{+}$ and $\zeta_{j,\new,k-1} \leq \zeta_{j,\new,k-1}^{+}$, that hold when $X_{\hat{u}_j+k-1}\in\Gamma_{j,k-1}^{\hat{u}_j}$ for $\hat{u}_j=u_j$ or $\hat{u}_j = u_j+1$, to obtain bounds on the restricted isometry constant (RIC) of the sparse recovery matrix $\bm{\Phi}_t$ and the sparse recovery error $\| \bm{b}_t \|_2$.
Applying the noisy compressed sensing (CS) result from \cite{candes_rip} and the assumed bounds on $\zeta$ and $\gamma$, the lemma follows.


\begin{lem}[{Bounding the RIC of $\bm{\Phi}_{t}$ \cite[Lemma 6.6]{rrpcp_perf}, \cite{rrpcp_isit15}}] \label{RIC_bnd}
Recall that $\zeta_{j,*}:= \|(\I-\Phat_{(j),*}\Phat_{(j),*}{}')\bm{P}_{(j),*}\|_2$.
\begin{enumerate}
\item Suppose that a basis matrix $\bm{P}$ can be split as $\bm{P} = [\bm{P}_1 \ \bm{P}_2]$ where $\bm{P}_1$ and $\bm{P}_2$ are also basis matrices. Then $\kappa_s^2 (\bm{P}) = \max_{\mathcal{T}: |\mathcal{T}| \leq s} \|{\I_{\mathcal{T}}}'\bm{P}\|_2^2 \le \kappa_s^2 (\bm{P}_1) + \kappa_s^2 (\bm{P}_2)$.
\item $\kappa_s^2(\Phat_{(j),*}) \leq (\kappa_{s,*})^2 + 2\zeta_*$ for all $j$
\item $\kappa_s (\Phat_{(j),\rmnew,k}) \leq \kappa_{s,\rmnew} + \zeta_{j,\new,k} + \zeta_{j,*}$ for all $j$ and $k$.
\item For $t\in[(u_{j-1}+K)\alpha+1,(\hat{u}_j + 1)\alpha)$, $\delta_{s} (\bm{\Phi}_{t}) = \kappa_s^2 (\Phat_{(j),*}) \leq  (\kappa_{s,*})^2 + 2 \zeta_{j,*}$.
\item For $k=1,\dots,K-1$, for $t\in[(\hat{u}_j+k)\alpha+1,(\hat{u}_{j}+k+1)\alpha] $ $\delta_{s}(\bm{\Phi}_{t})  = \kappa_s^2 ([\Phat_{(j),*} \ \Phat_{(j),\rmnew,k}]) \leq \kappa_s^2 (\Phat_{(j),*}) + \kappa_s^2 (\Phat_{(j),\rmnew,k}) \leq (\kappa_{s,*})^2 + 2\zeta_{j,*} + (\kappa_{s,\rmnew} +  \zeta_{j,\new,k} + \zeta_{j,*})^2$.
\end{enumerate}
\end{lem}

\begin{corollary}\label{RICnumbnd}
\
\begin{enumerate}
\item Conditioned on $\Gamma_{j-1,\rmend}$, for $t\in[t_j,(\uhat_j + 1)\alpha]$, $\delta_s(\bm{\Phi}_{t}) \leq \delta_{2s}(\bm{\Phi}_{t})  \leq (\kappa_{2s,*})^2 + 2\zeta_{j,*}^+ < 0.1 < 0.1479$, and $\|  [ ({\bm{\Phi}_{t})_{\mathcal{T}_t}}'(\bm{\Phi}_{t})_{\mathcal{T}_t}]^{-1} \|_2 \le \frac{1}{1-\delta_s(\bm{\Phi}_{t})} < 1.2 := \phi^+$.

\item For $k=2,\dots,K$ and $\hat{u}_j = u_j$ or $\hat{u}_j = u_j + 1$, conditioned on $\Gamma_{j,k-1}^{\hat{u}_j}$, for $t\in[(\hat{u}_{j}+k-1)\alpha+1,(\hat{u}_j + k)\alpha]$, $\delta_s(\bm{\Phi}_{t}) \leq \delta_{2s}(\bm{\Phi}_{t}) \leq (\kappa_{2s,*})^2 + 2\zeta_{j,*}^+ + (\kappa_{2s,\rmnew} +  \zeta_{j,\new,k-1}^+ + \zeta_{j,*}^+)^2 < 0.1479$, and $\|  [ ({\bm{\Phi}_{t})_{\mathcal{T}_t}}'(\bm{\Phi}_{t})_{\mathcal{T}_t}]^{-1} \|_2 \le \frac{1}{1-\delta_s(\bm{\Phi}_{t})} < 1.2 := \phi^+$.

\item For $\hat{u}_j = u_j$ or $\hat{u}_j = u_j + 1$, conditioned on $\Gamma_{j,K}^{\hat{u}_j}$, for $t\in[(\hat{u}_j+K)\alpha+1,t_{j+1}-1]$,  $\delta_s(\bm{\Phi}_{t}) \leq \delta_{2s}(\bm{\Phi}_{t})  \leq (\kappa_{2s,*})^2 + 2\zeta_{j,*}^+ < 0.1 < 0.1479$, and $\|  [ ({\bm{\Phi}_{t})_{\mathcal{T}_t}}'(\bm{\Phi}_{t})_{\mathcal{T}_t}]^{-1} \|_2 \le \frac{1}{1-\delta_s(\bm{\Phi}_{t})} < 1.2 := \phi^+$.
\end{enumerate}
\end{corollary}

\begin{proof}
This follows using Lemma \ref{RIC_bnd}, the definitions of $\Gamma_{j-1,\rmend}$ and $\Gamma_{j,k}^{\hat{u}_j}$, and Fact \ref{d_large}. 
\end{proof}



\begin{proof}[Proof of Lemma \ref{cslem_cor}]
We will prove claim 2).  The others are done in the same way.

By Fact \ref{d_large}, $\Gamma_{j,k-1}^{\hat{u}_j}$ implies that $\zeta_{j,*} \leq \zeta_{j,*}^+$ and $\zeta_{j,\new,k-1}\leq \zeta_{j,\new,k-1}^+$.

\begin{enumerate}[a)]
\item For $t \in [(\uhat_j+k-1)\alpha+1, (\hat{u}_j+k)\alpha]$, $\bm{b}_t := ( \I - \Phat_{t-1} \Phat_{t-1} {}') (\lt + \wt)$. Thus, using Fact \ref{ltbnds}, 
\begin{align*}
\|\bm{b}_t\|_2 & \leq \xi_{\text{cor}} = \xi 
\end{align*}

\item By Corollary \ref{RICnumbnd}, $\delta_{2s} (\bm{\Phi}_{t}) < 0.15 < \sqrt{2}-1$. Given $|\mathcal{T}_t| \leq s$, $\|\bm{b}_t\|_2 \leq \xi$, by  \cite[Theorem 1.1]{candes_rip}, the CS error satisfies
\[
\|\hat{\bm{x}}_{t,\cs} - \xt \|_2 \leq  \frac{4\sqrt{1+\delta_{2s}(\bm{\Phi}_{t})}}{1-(\sqrt{2}+1)\delta_{2s}(\bm{\Phi}_{t})} \xi < 7 \xi.
\]

\item Using the above, $\|\hat{\bm{x}}_{t,\cs} - \xt\|_{\infty} \leq 7  \xi$. Since $\min_{i\in \mathcal{T}_t} |(\xt)_{i}| \geq x_{\min}$ and $(\xt)_{\mathcal{T}_t^c} = 0$, $\min_{i\in \mathcal{T}_t} |(\hat{\bm{x}}_{t,\cs})_i| \geq x_{\min} - 7 \xi$ and $\max_{i \in \bar{\mathcal{T}_t}} |(\hat{\bm{x}}_{t,\cs})_i| \leq 7 \xi$. If $\omega \leq x_{\min} - 7 \xi$, then $\hat{\mathcal{T}}_t \supseteq \mathcal{T}_t$. On the other hand, if $\omega \geq 7 \xi$, then $\hat{\mathcal{T}}_t \subseteq \mathcal{T}_t$. Since $\omega$ satisfies $7 \xi \leq \omega \leq x_{\min} -7 \xi$, the support of $\xt$ is exactly recovered, i.e. $\hat{\mathcal{T}}_t = \mathcal{T}_t$.

\item
Given $\hat{\mathcal{T}}_t = \mathcal{T}_t$, the least squares estimate of $\xt$ satisfies $(\hat{\bm{x}}_t)_{\mathcal{T}_t} = [(\bm{\Phi}_{t})_{\mathcal{T}_t}]^{\dag} \bm{y}_t = [ (\bm{\Phi}_{t})_{\mathcal{T}_t}]^{\dag} (\bm{\Phi}_{t} \xt + \bm{\Phi}_{t} \lt + \bm{\Phi}_{t} \wt )$ and $(\hat{\bm{x}}_t)_{\bar{\mathcal{T}_t}} = \bm{0}$.
Also,  ${(\bm{\Phi}_{t})_{\mathcal{T}_t}}' \bm{\Phi}_{t} = {\I_{\mathcal{T}_t}}' \bm{\Phi}_{t}$ (this follows since $(\bm{\Phi}_{t})_{\mathcal{T}_t} = \bm{\Phi}_{t} \I_{\mathcal{T}_t}$ and ${\bm{\Phi}_{t}}'\bm{\Phi}_{t} = \bm{\Phi}_{t}$).
Using this, the error $\et := \hat{\bm{x}}_t - \xt$ satisfies (\ref{etdef0_cor}).

\item Using Fact \ref{ltbnds} we get the bound on $\|\et\|_2$. 
\end{enumerate}
\end{proof}

\section{Proof of Lemmas \ref{falsedet_del}, \ref{det}, \ref{pPCA}} \label{old_main_lemmas}

\begin{proof} [Proof of Lemma \ref{falsedet_del}] This proof is similar to that of Lemma 6.16 of \cite{rrpcp_isit15}.

Notice that
$\Pr(\mathrm{NODETS}_j^{\uhat_j} \ | \ \tilde\Gamma_{j,\vartheta}^{\uhat_j} )
= \Pr\Big(\lambda_{\max} \left( \frac{1}{\alpha}\bm{\mathcal{D}}_u \bm{\mathcal{D}}_u{}' \right) < \thresh  \text{ for all } u \in [\uhat_j + K + (\vartheta+1) +1, u_{j+1}-1] \ | \ \tilde\Gamma_{j,\vartheta}^{\uhat_j}\Big)$
for $\uhat_j = u_j$ or $\uhat_j = u_j + 1$.

Recall that $\Gamma_{j,\rmend}:= \Big( \tilde\Gamma_{j,\vartheta}^{u_j} \cap \mathrm{NODETS}_j^{u_j}\Big) \cup \left( \tilde\Gamma_{j,\vartheta}^{u_j+1} \cap \mathrm{NODETS}_j^{u_j+1}\right)$. Recall from Fact \ref{d_large} that $\Gamma_{j,\rmend}$ implies that $\mathrm{dif}(\Phat_{(j),*},\bm{P}_{(j),*}) \leq r\zeta$.


Also, for $u\in  [\uhat_j + K + (\vartheta+1) +1, u_{j+1}-1]$, $\Phat_{u\alpha-1,*} =\Phat_{(j+1),*}$ and for all  $t \in \J_u$ for these $u$'s, $\bnu_t= \bm{P}_{(j)} \at = \bm{P}_{(j+1),*} \at$.

Using Lemma \ref{cslem_cor}, under the given conditioning, $\|\bm{e}_t\|_2 \leq  \frac{\phi^+}{1-b} (2\zeta_{j,*}^+ \sqrt{r}\gamma +  2\epsilon_w ) $ for times $t \in \J_u$ for all these $u$'s.
Therefore,
\begin{align*}
\lambda_{\max} \left( \frac{1}{\alpha}  \bm{\mathcal{D}}_u \bm{\mathcal{D}}_u{}' \right) &= \lambda_{\max}\left(\frac{1}{\alpha}\sum_{t\in\mathcal{J}_u} (\I - \hat{\bm{P}}_{u\alpha-1,*}\hat{\bm{P}}_{u\alpha-1,*}{}') \lhat_t\lhat_t{}'(\I - \hat{\bm{P}}_{u\alpha-1,*}\hat{\bm{P}}_{u\alpha-1,*}{}')\right)\\
&= \lambda_{\max}\left( \frac{1}{\alpha}\sum_{t\in\mathcal{J}_u} (\I - \Phat_{(j+1),*}\Phat_{(j+1),*}{}') (\lt-\et)(\lt-\et){}'(\I - \Phat_{(j+1),*}\Phat_{(j+1),*}{}') \right) \\
&\leq \frac{(2\zeta_{j,*}^+)^2r\gamma^2}{(1-b)^2} + 2\phi^+ (2\zeta_{j,*}^+ \sqrt{r}\gamma +  2\epsilon_w )\frac{2\zeta_{j,*}^+\sqrt{r}\gamma}{(1-b)^2}  +  \frac{( \phi^+ (2\zeta_{j,*}^+ \sqrt{r}\gamma +  2\epsilon_w ) )^2}{(1-b)^2}\\
&\leq \frac{0.05\zeta\lambda^-}{0.81} + \frac{2.4\cdot(0.05+\sqrt{0.006})\zeta\lambda^-}{0.81} + \frac{1.44\cdot(\sqrt{0.05}+\sqrt{0.03})^2\zeta\lambda^-}{0.81}
< 0.5 \lamtrain = \thresh  
\end{align*}
The first inequality uses  the bound on $\|(\I - \Phat_{(j+1),*}\Phat_{(j+1),*}{}') \lt\|_2 = \|\bm\Phi_{(j+1),0}\lt\|_2 $ from Fact \ref{ltbnds} and the bound on $\|\et\|_2$ from Lemma \ref{cslem_cor}.
The second inequality uses the bound on $\epsilon_w$ from Model \ref{wt_model} and the Theorem; the bound $\zeta\leq\frac{0.05\lammin}{(r+r_{\new})^3\gamma^2}$ from the Theorem and the lower bound on $\lamtrain$ from Lemma \ref{initsub_cor}.
\end{proof}

\begin{proof}[Proof of Lemma \ref{det}] This proof is similar to that of the corresponding lemma from \cite{rrpcp_isit15}.
We will prove that $\Pr\left(\mathrm{DET}^{u_j+1} \ |\ X_{u_j} \right) > p_{\det,1}$ for all $X_{u_j} \in \Gamma_{j-1,\rmend}$. In particular, this will imply that $\Pr(\mathrm{DET}^{u_j+1} \ |\ X_{u_j}) > p_{\det,1}$ for all $X_{u_j} \in \Gamma_{j-1,\rmend} \cap \overline{\mathrm{DET}^{u_j}}$ and so, by Lemma \ref{event}, we can conclude that  $\Pr(\mathrm{DET}^{u_j+1} \ |\ \Gamma_{j-1,\rmend}, \overline{\mathrm{DET}^{u_j}}) > p_{\det,1}$.

The following claim is a direct corollary of Lemmas \ref{Ak_cor} and \ref{calHk_cor}. It follows exactly as the proof of these lemmas for the $k=1$ case but with using $u = u_j+1$ instead of $u = \uhat_j+1$.
\begin{align*}
& \Pr\left( \lambda_{\min} \left(\bm{A}_{u_j + 1} \right)\geq b_{\bm{A}} \ \big| \ X_{u_j} \right) \geq 1 -  p_{\bm{A}}, \\
&\Pr \left(\|\bm{\mathcal{H}}_{u_j+1}\|_2 \leq b_{\bm{\mathcal{H}},1}  \ \big| \ X_{u_j}\right)\geq 1 - p_{\bm{\mathcal{H}}}
\end{align*}
for all $X_{u_j}\in\Gamma_{j-1,\rmend}$.
By Lemma \ref{bounds_b_A}, $b_{\bm{A}} - b_{\bm{\mathcal{H}},1}  \geq \thresh$.

From the algorithm, notice that, $\bm{\mathcal{M}}_u = \frac{1}{\alpha}\bm{\mathcal{D}}_u{\bm{\mathcal{D}}_u}'$. Thus, 
\begin{align*}
\Pr\left(\mathrm{DET}^{u_j+1} \ |\ X_{u_j} \right) &= \Pr\left( \lambda_{\max}(\M_{u_{j}+1}) > \thresh \ | \ X_{u_j} \right)
\end{align*}
By Weyl's inequality and the above,
\begin{align*}
\lambda_{\max}(\bm{\mathcal{M}}_{u_{j}+1}) 
&\geq \lambda_{\max}(\bm{A}_{u_{j}+1}) - \|\bm{\mathcal{H}}_{u_{j}+1}\|_2 \\
&\geq \lambda_{\min}(\bm{A}_{u_{j}+1}) - \|\bm{\mathcal{H}}_{u_{j}+1}\|_2 \\
&\geq b_{\bm{A}} - b_{\bm{\mathcal{H}},1} \geq \thresh
\end{align*}
with probability at least $1-p_{\bm{A}}-p_{\bm{\mathcal{H}}} = p_{\mathrm{det},1}$, whenever $X_{u_{j}} \in \Gamma_{j-1,\rmend}$. Thus the result follows.
%
\end{proof}

\begin{proof}[Proof of Lemma \ref{pPCA} (p-PCA lemma)]  This proof is similar to that of the corresponding lemma from \cite{rrpcp_isit15}.
To prove this lemma we need to show two things.  First, conditioned on $\Gamma_{j,k-1}^{\hat{u}_j}$, the $k^{\text{th}}$ estimate of the number of new directions is correct.  That is: $\hat{r}_{j,\new,k} = r_{j,\new}$.  Second, we must show $\zeta_{j,\new,k}\leq\zeta_{j,\new,k}^+$, again conditioned on $\Gamma_{j,k-1}^{\uhat_j}$.

Notice that $\hat{r}_{j,\new,k} = \rank(\Phat_{(j),\new,k})$.  To show that $\rank(\Phat_{(j),\new,k}) = r_{j,\new}$, we need to show that for $u = \hat{u}_j+k$, $k=1,\dots,K$, $\lambda_{r_{j,\new}}(\M_u)>\thresh$ and $\lambda_{r_{j,\new}+1}(\M_u)<\thresh$.
Observe that, $\bm{\mathcal{M}}_u = \bm{\mathcal{A}}_u + \bm{\mathcal{H}}_u$.
By Lemma \ref{bounds_b_A}, Lemmas \ref{Ak_cor} and \ref{Akperp_cor} followed by Lemma \ref{event}, $\lambda_{\min}(\bm{A}_u) \ge b_{\bm{A}} > b_{\bm{A},\perp} \ge \lambda_{\max}(\bm{A}_{u,\perp})$ with probability at least $1-p_{\bm{A}}-p_{\bm{A},\perp}$ under the given conditioning.
Since $\bm{A}_u$ is of size $r_{j,\new}\times r_{j,\new}$, this means that $\lambda_{r_{j,\new}}(\bm{\mathcal{A}}_u) = \lambda_{\min}(\bm{A}_u)$ and $\lambda_{r_{j,\new}+1}(\bm{\mathcal{A}}_u) = \lambda_{\max}(\bm{A}_{u,\perp})$. Using these facts, Weyl's inequality,  Lemmas \ref{Ak_cor}, \ref{Akperp_cor} and \ref{calHk_cor}, and the bounds from Lemma \ref{zetadecay}, we can conclude that
with probability at least $p_{\mathrm{ppca}}$, under the given conditioning,
\begin{align*}
\lambda_{r_{j,\new}}(\bm{\mathcal{M}}_u) 
&\geq \lambda_{r_{j,\new}}(\bm{\mathcal{A}}_u) - \|\bm{\mathcal{H}}_u\|_2 \\
&= \lambda_{\min}(\bm{A}_u) - \|\bm{\mathcal{H}}_u\|_2 \ge b_{\bm{A}} - b_{\bm{\mathcal{H}},k} \ge \thresh
\end{align*}
and
\begin{align*}
\lambda_{r_{j,\new}+1}(\bm{\mathcal{M}}_u) 
&\leq \lambda_{r_{j,\new}+1}(\bm{\mathcal{A}}_u) + \|\bm{\mathcal{H}}_u\|_2 \\
&= \lambda_{\max}(\bm{A}_{u,\perp}) + \|\bm{\mathcal{H}}_u\|_2 \le b_{\bm{A},\perp} + b_{\bm{\mathcal{H}},k} <  \thresh
\end{align*}
Therefore $\rank(\Phat_{(j),\new,k})=r_{j,\new}$ with probability greater than $p_{\mathrm{ppca}}$ under the given conditioning.

To show that $\zeta_{j,\new,k} \leq \zeta_{j,\new,k}^+$, we use Lemmas \ref{Ak_cor}, \ref{Akperp_cor}, and \ref{calHk_cor}. Using $\rank(\Phat_{(j),\new,k})=r_{j,\new}$ and applying Lemma \ref{zetakbnd} with these bounds; using $\lambda_\new^- \ge \lambda^-$; and finally using Lemma \ref{event} gives the desired result.
\end{proof}


\section{Proof of Theorem \ref{thm1_cor_nodel}} \label{proof_thm1_cor_nodel}
The proof follows with the following re-definitions.
Redefine $\Gamma_{j,\rmend}$ as
\[
\ds \Gamma_{j,\rmend} := \Big( \Gamma_{j,K}^{u_j} \cap \mathrm{NODETS}_j^{u_j} \Big) \cup \left(  \Gamma_{j,K}^{u_j +1} \cap \mathrm{NODETS}_j^{u_j +1} \right).
\]
%
%
We get Corollary \ref{main_cor} then by just combining Lemmas \ref{falsedet_del}, \ref{det}, \ref{pPCA}. The rest of the argument needed is the same as that used to prove Theorem \ref{thm1_cor} in Sec. \ref{pf_thm}.
%
Proofs of Lemmas \ref{falsedet_del}, \ref{det}, \ref{pPCA} follow using the following redefinitions. Re-define
\ben
\item $\P_{(j),*}:= [\P_0, \P_{t_1,\new}, \P_{t_2,\new}, \dots \P_{t_{j-1},\new}]$. By the assumption given in the theorem, $\P_{(j),*}$ is orthogonal to $\P_{(j),\new}$. 
\item $\Phat_{(j+1),*}:= \Phat_{\that_j+K \alpha}$. Thus, given all subspace change times are correctly detected, $\Phat_{(j+1),*} = [\Phat_{(j),*}, \Phat_{(j),\new,K}]$. Thus, $\Gamma_{j,\rmend}^a$ implies $\zeta_{j+1,*} \le \zeta_{j,*} + \zeta_{j,\new,K}$. Use this fact to replace the third item of Fact \ref{d_large}.

\item $\zeta_{j,*}^+: = (r_0+(j-1)r_\new) \zeta$ and $\zeta_{j,\add}^+: = (r_0+ jr_\new) \zeta$. Thus, $\zeta_{j+1,*} \le \zeta_{j,*} + \zeta_{j,\new,K} \le  \zeta_{j+1,*}^+$.
\een

\bibliographystyle{IEEEbib} 
\bibliography{../bib/tipnewpfmt_kfcsfullpap}
\end{document}